\documentclass{article}

\usepackage[margin=1in]{geometry}
\usepackage[T1]{fontenc}
\usepackage{amsmath,amsthm,amssymb,hyperref}
\hypersetup{
  colorlinks=true,
  linkcolor=black,
  citecolor=blue
}
\usepackage{comment}
\usepackage[utf8]{inputenc}
\usepackage{graphicx}
\usepackage{lineno,xcolor}
\usepackage[style=alphabetic, backend=biber, maxcitenames=99, maxbibnames=99, minalphanames=3, labelalpha=true]{biblatex}
\usepackage{etoolbox}
\usepackage{float}
\usepackage{subcaption}
\usepackage{tikz}
\usetikzlibrary{decorations.pathreplacing}

\theoremstyle{definition}
\newtheorem{definition}{Definition}[section]

\theoremstyle{plain}
\newtheorem{theorem}{Theorem}[section]
\newtheorem{lemma}[theorem]{Lemma}
\newtheorem{corollary}[theorem]{Corollary}
\newtheorem{proposition}[theorem]{Proposition}

\theoremstyle{remark} 
\newtheorem*{remark}{Remark}
\newtheorem{claim}{Claim}

\usepackage[ruled,vlined]{algorithm2e}
\author{
    Zhengyang Liu\thanks{Beijing Institute of Technology. Email: zhengyang@bit.edu.cn} \and
    Ying Qin\thanks{Renmin University of China. Email: qinying0420@ruc.edu.cn} \and
    Zeyu Ren\thanks{Renmin University of China. Email: cs.zeyu.ren@outlook.com} \and
    Zihe Wang\thanks{Renmin University of China. Email: wang.zihe@ruc.edu.cn}
}

\begin{document}

\title{Second-Best Bilateral Trade is $1/2$ Efficient} 
\date{}
\pagenumbering{roman}
\maketitle

\begin{abstract}
    The landmark Myerson-Satterthwaite Theorem establishes a fundamental impossibility in bilateral trade: no Bayesian incentive-compatible mechanism can simultaneously achieve ex-post efficiency, individual rationality, and strong budget balance. We resolve a long-standing open question regarding the efficiency loss imposed by these constraints. Specifically, we prove that the Bayesian-optimal (second-best) mechanism always captures at least half of the first-best gains from trade ($\mathrm{SB}\ge\frac{1}{2}\mathrm{FB}$). This result is tight, definitively closing the gap between the previously best-known bounds of $0.317$ and $0.736$.
\end{abstract}

\tableofcontents
\clearpage
\pagenumbering{arabic}

\section{Introduction}
The design of efficient mechanisms is a foundational challenge at the intersection of economics and computer science. In its most elemental form, {\em bilateral trade} involves a single seller and a single buyer with private valuations, $b$ and $s$, drawn from public prior distributions $F_B$ and $F_S$. The goal is to maximize the {\em Gains from Trade} (GFT), defined as the expected surplus $b-s$ generated whenever a trade occurs.

In an ideal {\em first-best} setting, trade occurs if and only if $b\ge s$, thereby capturing the full social surplus. However, the landmark Myerson-Satterthwaite Theorem~\cite{MYERSON1983265} establishes a celebrated impossibility: no Bayesian incentive-compatible (BIC) mechanism can simultaneously satisfy ex-post efficiency, individual rationality (IR), and strong budget balance (SBB). This ``trilemma'' implies that any mechanism operating without an external subsidy must inevitably sacrifice some efficiency to maintain strategic integrity and financial self-sufficiency.

Faced with this impossibility, the classical economic approach seeks the {\em second-best} mechanism, that is, the Bayesian-optimal design that maximizes GFT for a given pair of prior distributions. While mathematically elegant, second-best mechanisms are often notoriously complex and brittle, relying heavily on the exact specification of priors. This fragility raises a fundamental question that defines this line of research: 
\begin{quote}
    To what extent can a second-best mechanism approximate the first-best efficiency?
\end{quote}

To circumvent the analytical barriers posed by the exact second-best mechanism, the field of algorithmic game theory has extensively studied simple and robust alternatives. A recent breakthrough by Deng, Mao, Sivan and Wang~\cite{deng2022approximately} demonstrated that the random-offerer (RO) mechanism, which randomizes between buyer-optimal and seller-optimal take-it-or-leave-it offers, achieves at least a $\frac{1}{8.23}\approx 0.121$ fraction of the first-best efficiency. This approximation guarantee was subsequently improved to $\frac{1}{3.15}$ by Fei~\cite{fei2022improved}. Complementing these advances, Hartline and Wang~\cite{hartline2025geometricanalysisgainstrade} provided a geometric analysis of these bounds, offering deeper intuition into their performance. Concurrently, a series of negative results have established theoretical limits on the RO mechanism's maximum potential~\cite{babaioff2021notegainstraderandomofferer,cai2021multi,cai2026newlowerboundrandom}.

Despite the rapid progress in understanding simple mechanisms, our structural understanding of the second-best benchmark remains surprisingly limited. To date, the only established theoretical upper bound indicates that the second-best is at most $\frac{2}{e}$-efficient relative to the first-best~\cite{blumrosen16approx_gains_trade_bilat_tradin}. Consequently, the worst-case approximation ratio for the second-best mechanism is currently known only to lie within the wide interval:
$$
\left[ \frac{1}{3.15}\approx 0.317, \frac{2}{e}\approx 0.736 \right].
$$
This substantial gap between the best known lower bound for a simple mechanism and the best known upper bound for the optimal mechanism highlights a significant void in our understanding of the fundamental limits of bilateral trade.

In this paper, we resolve the long-standing uncertainty regarding the efficiency of the second-best benchmark by closing the approximation gap.
\begin{theorem}[Main Result, informal]
  For every pair of independent buyer and seller distributions,
  $$
  \mathrm{SB}\ge \frac{1}{2} \mathrm{FB}.
  $$
  Moreover, the ratio $1/2$ is tight: there exist distributions where no Bayesian incentive compatible, individually rational, and strongly (or even weakly) budget-balanced mechanism can recover more than half of the first-best surplus.
\end{theorem}

\subsection{More Related Work}

Parallel to the pursuit of the first-best benchmark, a prolific line of work has explored fixed-price mechanisms. The primary motivation for this focus stems from the characterization by Hagerty and Rogerson~\cite{hagerty87robus_tradin_mechan}, who proved that fixed-price mechanisms (and distributions over them) are the {\em only} mechanisms that are dominant-strategy incentive compatible, individually rational, and strongly budget-balanced. Early progress in this domain was made by McAfee~\cite{mcafee08the_gains_trade_under_fixed_price_mechan}, showing that median-based pricing yields a $0.5$-approximation to the second-best GFT for symmetric distributions. Blumrosen and  Mizrahi~\cite{blumrosen16approx_gains_trade_bilat_tradin} later contributed a $1/e$ approximation for the first-best GFT under the monotone hazard rate assumption, while $1/(e-1)$ is the tight ratio due to Fei~\cite{fei2022improved}.
These paradigms were extended to more general two-sided markets by~\cite{brustle17approx_gains_trade_two_sided}, and more recently, \cite{hajiaghayi25gains_trade_bilat_trade_broker} analyzed the approximation limits of posted-pricing in broker-mediated environments.

Recent strides in evaluating the welfare bounds of fixed-price models have produced increasingly tight approximation ratios. Blumrosen and Dobzinski~\cite{blumrosen21almos_effic_mechan_bilat_tradin} proposed a $(1-1/e)$-approximation to the first-best welfare, which is improved by Kang, Pernice and Vondr{\'a}k~\cite{kang22fixed_price_approx_bilat_trade} to $1-1/e+0.0001$. 
In independent works, Liu, Ren and Wang~\cite{liu2023improved} and Cai and Wu~\cite{cai23on_optim_fixed_price_mechan_bilat_trade} significantly improved both approximation ratios and lower bounds for social welfare using different techniques. Huang and Zhang~\cite{huang25_fixed_price_mechanism_bilateral_trade} subsequently improved the efficiency of the algorithms in~\cite{liu2023improved}. Most recently, Giambartolomei and de Keijzer~\cite{giambartolomei2026approximation} refined the mathematical-programming approach of Cai and Wu~\cite{cai23on_optim_fixed_price_mechan_bilat_trade}, improving the best known guarantee for optimal fixed-price mechanisms in the single-unit social-welfare model to (0.7292), and extending the framework to multi-unit bilateral trade.

\subsection{Technical Overview}
\label{sec:technical-overview}


Recent work on bilateral trade has primarily shown that remarkably simple mechanisms achieve constant approximations to the first-best. In contrast, we establish a sharp universal guarantee for the exact second-best benchmark. Although the second-best mechanism is defined only implicitly and can be highly complex, its worst-case performance admits a clean and sharp characterization.

Analyzing this benchmark requires us to reason directly about the solution to a constrained optimization problem over all BIC, IR, and WBB mechanisms. The difficulty is not to evaluate a given allocation rule, but to uncover enough structure in this implicit benchmark to prove a tight worst-case guarantee.

To overcome these barriers, our approach proceeds in two main stages.
\begin{enumerate}
\item 
\textbf{We reformulate the approximation guarantee as a statement regarding the feasibility of a greedy allocation procedure.} After a standard finite-type approximation, we take a Lagrangian view of the budget-balance constraint and reduce the task of proving
\(
\mathrm{SB}\ge \frac{1}{2}\cdot \mathrm{FB}
\)
to showing that the optimal value of an auxiliary unconstrained optimization problem is nonnegative. We then exploit the structure of this optimal-value problem: its nonnegativity can be certified by constructing an allocation function with the required monotonicity. This turns the original inequality into an allocation-construction problem. We provide such a construction through a greedy procedure, so the remaining task is to prove that this algorithmic construction never fails.

\item \textbf{We prove that the greedy construction never fails through a state-compression and dominance argument.} Although the greedy procedure generates a high-dimensional allocation rule, its future evolution can be summarized by a low-dimensional state. This compression allows us to compare different executions of the procedure without tracking the entire allocation history. We then introduce a dominance relation that captures which states are closer to failure, and prove truncation lemmas showing that any hypothetical failure would imply a smaller and more demanding failing instance. A minimal-counterexample argument therefore rules out failure altogether. Hence, the greedy construction is always feasible, which proves the nonnegativity of the auxiliary optimization problem and establishes the universal lower bound. Finally, we prove tightness by constructing a family of distributions whose second-best-to-first-best ratio converges to \(1/2\).

\end{enumerate}
\section{Preliminaries}
\paragraph{Notation.}
We denote the sets of non-negative real numbers and non-negative integers by $\mathbb{R}_{\ge0}$ and $\mathbb{Z}_{\ge 0}$, respectively.

The standard bilateral trade problem involves a seller and a buyer.
The seller holds a single indivisible good with a private valuation $s$. The buyer has a private valuation $b$ for acquiring the good.
These valuations are independent random variables drawn from prior distributions $F_S$ and $F_B$, respectively.
Without loss of generality, we assume that $F_S$ and $F_B$ are supported on $[0,1]$, as our analysis focuses solely on the scale-invariant efficiency ratio.
The mechanism designer knows these prior distributions but does not observe the realized valuations $s$ and $b$. Our primary objective is to maximize the expected surplus, or the Gains from Trade (GFT), $b-s$, generated upon a successful transaction.

A mechanism is fully characterized by an {\em allocation rule} $x: [0,1]^2 \to [0,1]$, which specifies the probability of trade given the reported valuations, and a {\em payment rule}, which dictates the financial transfers. We denote the buyer's expected payment by $p_B: [0,1]^2 \to \mathbb{R}$ and the seller's expected revenue by $p_S: [0,1]^2 \to \mathbb{R}$. The expected gain from trade achieved by an allocation rule $x$ is formally defined as:
\[\mathrm{GFT}_{F_S, F_B}(x)=\int_0^1\int_0^1(b-s)x(s,b)\mathrm{d}F_B(b)\mathrm{d}F_S(s).\]

The ideal benchmark is the {\em first-best} mechanism, which achieves full ex-post efficiency by ensuring trade occurs if and only if the buyer's value meets or exceeds the seller's. Formally, this is defined by the allocation rule $x(s,b)=1$ when $b\ge s$, and $0$ otherwise.

The gain from trade in the first-best outcome, denoted by $\mathrm{FB}_{F_S,F_B}$ is defined as:
\[ \mathrm{FB}_{F_S,F_B}:=\int_0^{1}\int_0^b(b-s)\mathrm{d}F_S(s)\mathrm{d}F_B(b).\]

In standard mechanism design, we require mechanisms to satisfy key strategic and participation constraints.
First, we define the interim expected {\em utility} for a buyer with true valuation $b$ reporting $b'$, and a seller with true valuation $s$ reporting $s'$:
\begin{align*}
    u_B(b,b') &:= \mathbb{E}_{s \sim F_S}[x(s,b') \cdot b - p_B(s,b')];\\
u_S(s,s') &:= \mathbb{E}_{b \sim F_B}[p_S(s',b) - x(s',b) \cdot s].
\end{align*}
A mechanism is {\em Bayesian incentive-compatible} (BIC) if truth-telling maximizes each agent's expected utility. That is, for any $b, b', s, s' \in [0,1]$:
$$u_B(b,b) \ge u_B(b,b') \quad \text{and} \quad u_S(s,s) \ge u_S(s,s').$$
Furthermore, a mechanism is {\em interim individually rational} (IR) if both agents are guaranteed a non-negative expected utility from participating, meaning $u_B(b,b) \ge 0$ and $u_S(s,s) \ge 0$ for all $b, s \in [0,1]$. 

Finally, to ensure the mechanism is financially self-sufficient and does not rely on external subsidies, we impose a budget balance constraint. In this paper, we formalize this requirement as follows:
\begin{definition}[Budget Balance Constraints]
A mechanism is \emph{strongly budget-balanced} (SBB) if the payment collected from the buyer exactly equals the revenue paid to the seller for every valuation profile $(s, b)$:
$$p_B(s,b) = p_S(s,b).$$
A mechanism is \emph{weakly budget-balanced} (WBB) if the payment collected from the buyer is at least the revenue paid to the seller, meaning the mechanism never runs a deficit:
$$p_B(s,b) \ge p_S(s,b).$$
\end{definition}

In their seminal work, Myerson and Satterthwaite~\cite{MYERSON1983265} characterized the second-best mechanism, which maximizes the expected gains from trade subject to Bayesian incentive compatibility (BIC), interim individual rationality (interim IR), and strong budget balance (SBB) constraints. Assuming the prior distributions $F_S$ and $F_B$ admit continuous and strictly positive density functions $f_S$ and $f_B$, a central result of their analysis is that an allocation rule $x(s,b)$ is implementable under these constraints if and only if its interim allocation probabilities are monotonic and it satisfies the following budget balance discriminant:
\[\Lambda(x) = \int_0^{1}\int_0^{1}\left(b-\frac{1-F_B(b)}{f_B(b)}-s-\frac{F_S(s)}{f_S(s)}\right)x(s,b)f_B(b)f_S(s)\mathrm{d}s\mathrm{d}b\ge 0.\]
We denote by $\mathrm{SB}_{F_S,F_B}$ the expected gains from trade achieved by the second-best mechanism, formally defined as the supremum of the following optimization program:
\begin{align}\label{op:original}
\sup &~\mathrm{GFT}_{F_S,F_B}(x) \notag\\
\text{s.t.}&~
\Lambda(x)\geq 0, \notag\\
&~x_B(b)\le x_B(b'),\quad \forall b,b'\in[0,1], b< b',\\
&~x_S(s)\ge x_S(s'),\quad \forall s,s'\in[0,1], s< s', \notag
\end{align}
where $x_B(b)=\int_0^{1}x(t,b)f_S(t)\mathrm{d}t$ and $x_S(s)=\int_0^{1}x(s,t)f_B(t)\mathrm{d}t$.

In this paper, we study the worst-case approximation ratio between the second-best and first-best outcomes, that is,
\[\inf_{F_B,F_S\in\Delta([0,1])}\frac{\mathrm{SB}_{F_S,F_B}}{\mathrm{FB}_{F_S,F_B}},\]
where $\Delta([0,1])$ denotes the set of all probability measures over the interval $[0,1]$.

\begin{remark}
Hereafter, we drop the subscripts on FB, SB, GFT, and related terms whenever the underlying measures $F_B$ and $F_S$ are clear from the context.
\end{remark}

\section{From Continuous to Discrete Distributions}\label{sec:dis}

Since the original optimization problem~\eqref{op:original} is formulated over a continuum of type profiles, our first step is to extend it to the setting of Borel probability measures. Subsequently, we specialize this formulation to discrete distributions by restricting the measures to finite supports. All omitted proofs from this section are deferred to Appendix~\ref{app:dis}.

\begin{proposition}\label{prop:general_mixed}
Let \(F_S\) and \(F_B\) be Borel probability measures on \([0,1]\), the gain from trade of the second-best mechanism with BIC, interim IR, SBB (or WBB) constraint is the optimal value of the following optimization problem:

\begin{align*}
\sup_x\quad
& \int_{[0,1]^2}(b-s)\,x(s,b)\,\mathrm{d}F_S(s)\,\mathrm{d}F_B(b) \\
\text{s.t.}\quad
& \Lambda_{\mathrm{gen}}(x)=
\int_{[0,1]}\!\!
\left(
b\,x_B(b)-\int_{0}^{b}x_B(t)\,\mathrm{d}t
\right)\mathrm{d}F_B(b)
-
\int_{[0,1]}\!\!
\left(
s\,x_S(s)+\int_{s}^{1}x_S(t)\,\mathrm{d}t
\right)\mathrm{d}F_S(s)\ge 0,\\
& x_B(b)\le x_B(b'), \qquad \forall b,b'\in[0,1], b< b',\\
& x_S(s)\ge x_S(s'), \qquad \forall s,s'\in[0,1], s< s',
\end{align*}
where
\[
x_B(b)=\int_{[0,1]}x(s,b)\,\mathrm{d}F_S(s),\qquad
x_S(s)=\int_{[0,1]}x(s,b)\,\mathrm{d}F_B(b).
\]
\end{proposition}

Proposition~\ref{prop:general_mixed} establishes a crucial equivalence: under BIC and interim IR, the distinction between strong and weak budget balance does not affect the set of implementable allocation rules. Both budget constraints are characterized by the identical condition \(\Lambda_{\mathrm{gen}}(x) \ge 0\).\footnote{Consequently, the classical conclusions regarding the second-best mechanism in Myerson and Satterthwaite~\cite{MYERSON1983265} and Blumrosen and  Mizrahi~\cite{blumrosen16approx_gains_trade_bilat_tradin} apply equally under both SBB and WBB constraints.} Hereafter, we refer to this simply as the budget balance (BB) constraint.

To facilitate our proof, we restrict our attention to discrete distributions defined on grid points. As we will demonstrate, this restriction is entirely without loss of generality. We now specialize Proposition~\ref{prop:general_mixed} to this discrete setting. Given a positive integer $N$, let the seller's values be supported on $\{s_i\}_{i=0}^N$ with $0 \le s_0 < s_1 < \dots < s_N \le 1$, and the buyer's values on $\{b_j\}_{j=0}^N$ with $0 \le b_0 < b_1 < \dots < b_N \le 1$. Let $f_S(\cdot)$ and $f_B(\cdot)$ denote the respective probability mass functions, with their associated cumulative distribution functions defined as:
\[
F_S(s_i)=\sum_{t\le i} f_S(s_t),\qquad
F_B(b_j)=\sum_{t\le j} f_B(b_t).
\]

\begin{corollary}\label{prop:discrete}
For a given \(N>0\), suppose the seller’s and buyer’s values are drawn from \(F_S, F_B\), supported on \(\{s_i\}_{i=0}^N\) and \(\{b_j\}_{j=0}^N\), respectively. Then, under this discretization, second-best mechanism with BIC, interim IR, SBB (or WBB) can be reformulated as:
\begin{align*}
\sup_{x}\quad 
& \sum_{i,j=0}^N x(s_i,b_j)\,f_S(s_i)\,f_B(b_j)\,(b_j-s_i) \\
\text{s.t.}\quad
& \sum_{i,j=0}^N x(s_i,b_j)\, f_S(s_i)\, f_B(b_j)\,
\Bigg(
b_j+(b_j-b_{j+1})\frac{1-F_B(b_j)}{f_B(b_j)}
-s_i-(s_i-s_{i-1})\frac{F_S(s_{i-1})}{f_S(s_i)}
\Bigg)\ge 0, \\
& x_S(s_i)\ge x_S(s_{i+1}), \qquad \forall\, i=0,\ldots,N-1,\\
& x_B(b_j)\le x_B(b_{j+1}), \qquad \forall\, j=0,\ldots,N-1,
\end{align*}
where \(s_{-1}:=s_0\), \(b_{N+1}:=b_N\) and
\[
x_B(b_j)=\sum_{i=0}^N x(s_i,b_j)f_S(s_i),\qquad
x_S(s_i)=\sum_{j=0}^N x(s_i,b_j)f_B(b_j).
\]
\end{corollary}

To discretize the Borel probability measures, we restrict our analysis to an equidistant grid. Normalizing the support to the unit interval $[0,1]$, we define the grid points as $s_i = i\epsilon$ and $b_j = j\epsilon$ for $i,j \in \{0, \dots, N\}$, where $\epsilon = 1/N$.

Given Borel probability measures $F_S$ and $F_B$, we construct their discrete counterparts $\hat{F}_S$ and $\hat{F}_B$ by concentrating the probability mass onto this grid. Specifically, we assign the seller's mass on the interval $(s_{i-1}, s_i]$ to the upper endpoint $s_i$, and the buyer's mass on $[b_{j-1}, b_j)$ to the lower endpoint $b_{j-1}$. The discrete probability mass functions are thus defined as:
$$\hat{f}_S(s_i) = F_S(s_i) - F_S(s_{i-1}), \qquad \hat{f}_B(b_{j-1}) = F_B(b_j^-) - F_B(b_{j-1}^-),$$
where $F_B(v^-)$ denotes the left-hand limit of the cumulative distribution function at $v$.

Let $\hat{x}$ denote the optimal allocation rule for the discrete optimization problem (Corollary~\ref{prop:discrete}) under the distributions $\hat{F}_S$ and $\hat{F}_B$. We extend $\hat{x}$ back to the continuous domain by defining a piecewise-constant step function, $\tilde{x}$. Formally, for any continuous type profile $(s,b) \in (s_{i-1}, s_i] \times [b_{j-1}, b_j)$, we set the extended allocation to:
$$\tilde{x}(s,b) = \hat{x}(s_i, b_{j-1}).$$

\begin{lemma}\label{lem:feasibility-preservation}
The continuous extension $\tilde{x}$ satisfies BIC, interim IR, and BB constraints under the original distribution $F_S$ and $F_B$.
\end{lemma}

Lemma~\ref{lem:feasibility-preservation} shows that feasibility is preserved
when a discrete allocation rule is extended to the continuous domain. It remains
to verify that the corresponding objective values are also consistent in the continuum limit. We establish this convergence by employing a sequence of grids.

Throughout the remainder of the analysis, we operate on a sequence parameterized by a positive integer $M$. For each $M$, we define the grid points $s_i^{(M)} = i\cdot 2^{-M}$ and $b_j^{(M)} = j\cdot 2^{-M}$. We partition the unit interval into cells 
\[
I_0^{(M)} := \{0\},
\qquad
I_i^{(M)} := (s_{i-1}^{(M)}, s_i^{(M)}]
\text{ for } i = 1,\ldots,2^M,
\] and \[
J_j^{(M)} := [b_j^{(M)}, b_{j+1}^{(M)})
\text{ for } j = 0,\ldots,2^M-1,
\qquad
J_{2^M}^{(M)} := \{1\}.
\]
Let $\widehat{f}_S^{(M)}(s_i^{(M)})$ and $\widehat{f}_B^{(M)}(b_{j-1}^{(M)})$ denote the integrated probability masses over $I_i^{(M)}$ and $J_j^{(M)}$.

Let \(\mathcal X_M\) denote the feasible set of discrete allocation rules on
the \(2^{-M}\)-grid. The discretized second-best expected surplus is given by:
\[
\widehat{\mathrm{SB}}_M
:=
\max_{x\in\mathcal X_M}
\sum_{j=1}^{2^M}
\sum_{i=1}^{2^M}
\left(b_{j-1}^{(M)}-s_i^{(M)}\right)
x\left(s_i^{(M)},b_{j-1}^{(M)}\right)
\widehat f_S^{(M)}\left(s_i^{(M)}\right)
\widehat f_B^{(M)}\left(b_{j-1}^{(M)}\right).
\]
Let $\widehat{x}_M \in \mathcal{X}_M$ be an optimal allocation. We denote its step-function extension by $\widetilde{x}_M$, where $\widetilde{x}_M(s,b) = \widehat{x}_M(s_i^{(M)}, b_{j-1}^{(M)})$ for all $(s,b) \in I_i^{(M)} \times J_j^{(M)}$. Symmetrically, we define the discretized first-best benchmark as:
\[
\widehat{\mathrm{FB}}_M
:=
\sum_{j=1}^{2^M}
\sum_{i=1}^{j-1}
\left(b_{j-1}^{(M)}-s_i^{(M)}\right)
\widehat f_S^{(M)}\left(s_i^{(M)}\right)
\widehat f_B^{(M)}\left(b_{j-1}^{(M)}\right).
\]

We rely on a fundamental consistency property between successive discretizations: because the $2^{-M}$-grid strictly refines the $2^{-(M-1)}$-grid, allocations map seamlessly between them. Any feasible allocation on the coarser $2^{-(M-1)}$-grid can be carried to the finer $2^{-M}$-grid by duplicating its value across all corresponding sub-rectangles. Conversely, any feasible allocation on the finer grid can be pooled to the coarser grid by taking the mass-weighted average over the sub-rectangles that comprise each cell of the coarser partition.

\begin{lemma}\label{lem:convergence}
There exists an increasing sequence \(\{M_m\}_{m=1}^{\infty}\) such that
\[
\lim_{m\to\infty}
\frac{
\mathrm{GFT}_{F_S,F_B}(\widetilde x_{M_m})
}{
\mathrm{FB}_{F_S,F_B}
}
\ge
\lim_{M\to\infty}
\frac{
\widehat{\mathrm{SB}}_M
}{
\widehat{\mathrm{FB}}_M
}.
\]
\end{lemma}

We can now lift the discrete performance guarantee to the Borel measure setting. By Lemma~\ref{lem:feasibility-preservation}, each step-function extension $\widetilde{x}_{M_m}$ is a feasible allocation that satisfies three constraints for the original Borel measure. Therefore, the Borel measure second-best benchmark naturally satisfies:

\[\frac{\mathrm{SB}_{F_S,F_B}}{\mathrm{FB}_{F_S,F_B}}
\ge
\frac{
\mathrm{GFT}_{F_S,F_B}(\widetilde x_{M_m})
}{
\mathrm{FB}_{F_S,F_B}
}.\]
Taking \(m\to\infty\) and applying Lemma~\ref{lem:convergence} yields:
\[
\frac{\mathrm{SB}_{F_S,F_B}}{\mathrm{FB}_{F_S,F_B}}
\ge
\lim_{M\to\infty}
\frac{
\widehat{\mathrm{SB}}_M
}{
\widehat{\mathrm{FB}}_M
}.
\]
Thus, any uniform lower bound established for the discrete instances immediately applies to the original Borel measures.

Conversely, any discrete measure defined on a finite grid is indeed a valid Borel probability measure on $[0,1]$. Consequently, the worst-case efficiency ratio over all general Borel instances cannot exceed the worst-case ratio over finite-grid instances. Combining these two inequalities, we conclude that the discrete and continuous worst-case approximation ratios perfectly coincide.

Next, we simplify the discrete domain. For a fixed $\epsilon = 1/N$, the $\epsilon$-equidistant grid on $[0,1]$ is isomorphic to the integer grid $\{0, 1, \ldots, N\}$ up to a scaling factor of $\epsilon$. Because the surplus objective $b-s$ is homogeneous of degree one in the agents' valuations, both $\widehat{\mathrm{SB}}$ and $\widehat{\mathrm{FB}}$ scale by this exact same factor under the linear transformation. As a result, the ratio $\widehat{\mathrm{SB}}/\widehat{\mathrm{FB}}$ is strictly scale-invariant. This allows us to equivalently evaluate any finite equidistant grid on $[0,1]$ as a finite subset of the integers. Let $\Delta(\mathbb{Z}_{W})$ denote the space of all finitely supported probability distributions on the non-negative integers $0,1,\ldots,W$.

\begin{theorem}\label{thm:main_reduction}
Let $\Delta([0,1])$ be the set of all probability measures supported on $[0,1]$. The worst-case efficiency ratio over the continuous domain is strictly equivalent to that over the discrete integer domain:
$$\inf_{F_B,F_S \in \Delta([0,1])} \frac{\mathrm{SB}}{\mathrm{FB}}
\ =
\ \inf_{W\ge 0}
\ \inf_{\widehat{F}_B,\widehat{F}_S \in \Delta(\mathbb{Z}_{W})} 
\ \frac{\widehat{\mathrm{SB}}}{\widehat{\mathrm{FB}}}.$$
\end{theorem}
Accordingly, for the remainder of this paper, we restrict our analysis without loss of generality to discrete distributions supported on $\mathbb{Z}_{W}$.

\section{From Optimization to Greedy Algorithm}\label{sec:new-problem}

In this section, we develop the reduction that turns the second-best
efficiency problem into a constructive greedy feasibility test. The main
difficulty is that the second-best benchmark is defined by an optimization
problem subject to incentive and budget-balance constraints, while the desired
guarantee is a uniform lower bound over all buyer and seller distributions on
integer domains. Our goal is to replace this global optimization problem by a
sequence of local inequalities that can be checked algorithmically.

The reduction has three steps. First, we dualize the feasibility constraint
from Corollary~\ref{prop:discrete}. For a Lagrange multiplier
\(\alpha\ge 0\), the relevant objective does not lead to the ordinary Myerson virtual values. Rather, the regularity notion used in this section is \(\alpha\)-weak regularity: the seller-side
\(\alpha\)-virtual value
must be weakly increasing in \(s\). When this monotonicity fails, we apply an
\(\alpha\)-ironing procedure. The ironing step is a worst-case reduction: it
preserves the second-best value while weakly increasing the first-best value,
so it can only make the approximation ratio harder to prove.

Second, after this regularization, we apply Sion's Minimax Theorem to exchange
the order of optimization between the allocation rule and the buyer
distribution. This step is what converts a distributionally robust lower-bound
problem into a mechanism-construction problem. 

Third, once \(F_S\) and \(\alpha\) are fixed, the post-minimax objective is
linear in the buyer probability masses \(f_B(b)\). Therefore, robustness
against all buyer distributions is equivalent to a family of buyer-wise local
constraints, one for each buyer value \(b\). This yields a greedy
allocation procedure. We then prove that the greedy procedure succeeds if and
only if the original maximin feasibility problem is feasible.

For brevity, all omitted proofs for this section are deferred to
Appendix~\ref{app:new-problem}.
\subsection{The Lagrangian Dual Representation}
Following Corollary~\ref{prop:discrete}, we characterize the feasibility of a mechanism through constraints.
For a given pair of distributions $(F_S,F_B)$, we use $\Lambda_D(x)\ge 0$ to represent the first constraint, where
\[
\Lambda_D(x) =\sum_{s,b} x(s,b)\, f_S(s)\, f_B(b)\,
\left(
b-\frac{1-F_B(b)}{f_B(b)}
-s-\frac{F_S(s{-1})}{f_S(s)}
\right)\ge 0.
\]
The remaining two constraints, namely the monotonicity conditions
\(
x_B(b)\le x_B(b{+1})
\)
and
\(
x_S(s)\ge x_S(s{+1})
\)  for any $s,b\in \mathbb{Z}_{W}$,
are incorporated into the feasible domain. 
Let
\[
\mathcal R_1(W):=
\left\{x:\mathbb{Z}_{W}^2\to [0,1] \;\middle|\; \
\begin{array}{l} 
x_B(b) \le x_B(b') \text{ for all } b < b' \\ 
x_S(s) \ge x_S(s') \text{ for all } s < s' 
\end{array} \;\
\right\}.
\]
For finite-supported \(F_S,F_B\), the program
\[
\sup_{x\in\mathcal R_1(W)} \mathrm{GFT}(x)
\quad\text{s.t.}\quad \Lambda_D(x)\ge0
\]
is a finite-dimensional linear program. Hence, by LP strong duality, we introduce the Lagrangian $L(x,\alpha)$ for a multiplier $\alpha \ge 0$ with the following lemma:
\[
L(x,\alpha):=\mathrm{GFT}(x)+\alpha\,\Lambda_D(x).
\]

\begin{lemma}\label{lem:Lagrangian}
The second-best efficiency admits the following minimax representation: $$\mathrm{SB}_{F_S,F_B}
=
\inf_{\alpha\ge0}
\ \sup_{x\in\mathcal R_1    (W)}
\ L(x,\alpha).$$
\end{lemma}

Fix a target efficiency factor $\beta\in(0,1)$, and consider the task of proving
\[
    \mathrm{SB}_{F_S,F_B}\ge \beta\,\mathrm{FB}_{F_S,F_B}
\]
uniformly over all finite-support buyer and seller distributions on integer
domains. By Lemma~\ref{lem:Lagrangian}, for each fixed pair
\((F_S,F_B)\), the second-best benchmark can be written as the infimum over
the Lagrange multiplier \(\alpha\ge 0\) of the corresponding Lagrangian
objective. Therefore, the desired universal guarantee is equivalent to
requiring that the worst-case Lagrangian surplus above the target
\(\beta\cdot \mathrm{FB}\) be nonnegative:
\[
\inf_{W\ge0}
\ \inf_{F_S,F_B\in\Delta(\mathbb Z_{W})}
\ \inf_{\alpha\ge 0}
\ \sup_{x\in\mathcal R_1(W)}
\ \big[
\mathrm{GFT}(x)+\alpha\Lambda_D(x)-\beta\cdot\mathrm{FB}
\big]
\ge 0.
\]
The multiplier \(\alpha\) is an auxiliary dual variable whose feasible set
does not depend on \(W\), \(F_S\), or \(F_B\). Hence the infima are taken over
a product domain and can be reordered. We shall use the following equivalent
form, which is more convenient because it allows us to fix \(\alpha\) first
and then study the \(\alpha\)-dependent structure induced by the Lagrangian:
\[
\inf_{\alpha\ge 0}
\ \inf_{W\ge0}
\ \inf_{F_S,F_B\in\Delta(\mathbb Z_{W})}
\ \sup_{x\in\mathcal R_1(W)}
\ \big[
\mathrm{GFT}(x)+\alpha\Lambda_D(x)-\beta\cdot\mathrm{FB}
\big]
\ge 0.
\]

\subsection{\texorpdfstring{Discrete $\alpha$-Weak Regularity and Ironing}{Discrete alpha-Weak Regularity and Ironing}}

We next make two reductions that simplify the class of instances and the
allocation rules that need to be considered. The first is a support
normalization. Shifting both buyer and seller valuations by a common constant
does not affect gains from trade. The only subtlety is that, if the seller's
lowest type is \(m>0\), a direct common shift by \(-m\) may move buyer types
below \(m\) to negative values. This issue can be handled by a careful
decomposition of the buyer support and the corresponding allocation
constraints. We therefore obtain the following normalization result.

\begin{lemma}
\label{lem:wlog-fS0}
It suffices to prove the target guarantee for instances satisfying $f_S(0)>0$.
\end{lemma}

We now turn to the second
reduction, which concerns the regularity structure generated by the
Lagrangian. For a fixed multiplier \(\alpha\), the coefficient of
\(x(s,b)\) in
\(\mathrm{GFT}(x)+\alpha\Lambda_D(x)\) can be written as an
\(\alpha\)-dependent virtual surplus:
\[
\left((1+\alpha)b-\alpha\frac{1-F_B(b)}{f_B(b)}\right)
-
\left((1+\alpha)s+\alpha\frac{F_S(s-1)}{f_S(s)}\right).
\]
Thus the virtual value that matters here is not the ordinary virtual value, but rather the following \(\alpha\)-weak regularity.

\begin{definition}[$\alpha$-Virtual Value and $\alpha$-Weak Regularity]\label{def:virtual-value}
Given a parameter $\alpha \ge 0$ and a discrete seller distribution $F_S$, we define the seller-side $\alpha$-virtual value as:
$$\phi_S(s) := (1+\alpha)s + \alpha \frac{F_S(s-1)}{f_S(s)} \quad \forall s \in \mathbb{Z}_{\ge0},$$
where we adopt the convention $F_S(-1) = 0$. We say that $F_S$ is \emph{$\alpha$-weakly regular} if $\phi_S(s)$ is weakly increasing in $s$.
\end{definition}

The seller distribution need not be \(\alpha\)-weakly regular a priori. When
\(\phi_S(s)\) fails to be weakly increasing, we replace it by an ironed
version. The purpose of ironing is to convexify the cumulative
\(\alpha\)-virtual surplus so that the resulting slopes, i.e., the ironed
\(\alpha\)-virtual values, become monotone.

\begin{definition}[Ironed $\alpha$-Virtual Value]\label{def:ironing}
Let $F_S$ be a discrete seller distribution. Let $q_s := F_S(s)$ for $s \in \mathbb{Z}_{\ge0}$ (with $q_{-1} := 0$), and let $\Phi_S: [0,1] \to \mathbb{R}_{\ge 0}$ be the piecewise-linear function satisfying:
$$\Phi_S(q_s) := \sum_{s'=0}^s \phi_S(s')f_S(s').$$
Let $\overline{\Phi}_S$ denote the convex envelope (i.e., the greatest convex minorant) of $\Phi_S$ on $[0,1]$. The \emph{ironed $\alpha$-virtual value}, denoted $\phi_S^{\mathrm{ir}}$, is defined as the slope of $\overline{\Phi}_S$ on each interval $(q_{s-1}, q_s]$:
$$\phi_S^{\mathrm{ir}}(s) := \frac{\overline{\Phi}_S(F_S(s)) - \overline{\Phi}_S(F_S(s-1))}{f_S(s)}.$$
By construction, $\phi_S^{\mathrm{ir}}$ is weakly increasing in $s$. If $F_S$ is $\alpha$-weakly regular, then $\overline{\Phi}_S = \Phi_S$, yielding $\phi_S^{\mathrm{ir}} = \phi_S$. We write $F_S^{\mathrm{ir}}$ for the $\alpha$-weakly regular distribution associated with $\phi_S^{\mathrm{ir}}$.
\end{definition}

With the ironing operation in place, we can prove that there exists a solution $x$ possessing the following monotonicity property. After ironing, the
Lagrangian objective admits an optimizer with a stronger point-wise monotonicity
property, which is the form needed for the later greedy construction.

\begin{lemma}\label{lem:position-wise-monotone}
For any \(\alpha\ge 0\), the Lagrangian problem can be restricted, without loss of generality, to the set of point-wise monotone allocations:
$$\mathcal{R}_2(W) := \left\{ x : \mathbb{Z}_W^2 \to [0,1] \;\middle|\; 
\begin{array}{l} 
x(s, b) \le x(s, b') \text{ for all } b < b' \\ 
x(s, b) \ge x(s', b) \text{ for all } s < s' 
\end{array} 
\right\}.$$
\end{lemma}

Applying Lemma~\ref{lem:position-wise-monotone} to the preceding worst-case
program, the supremum over \(\mathcal R_1(W)\) may be replaced by a supremum over
\(\mathcal R_2(W)\). Hence it suffices to verify
\[
\inf_{\alpha\ge 0}
\ \inf_{W\ge0}
\ \inf_{F_S,F_B\in\Delta(\mathbb Z_{W})}
\ \sup_{x\in\mathcal R_2(W)}
\ \big[
\mathrm{GFT}(x)+\alpha\Lambda_D(x)-\beta\cdot\mathrm{FB}
\big]
\ge 0.
\]
 
Furthermore, we can restrict our focus entirely to $\alpha$-weakly regular distributions. The following lemma establishes that applying our ironing procedure to an irregular distribution constitutes a worst-case reduction, as it can only decrease the overall approximation ratio.

\begin{lemma}\label{lem:ironing-ratio}
Fix \(\alpha\ge 0\). Replacing \(F_S\) with its ironed version \(F_S^{\text{ir}}\) leaves \(\mathrm{SB}\) unchanged, but weakly increases \(\mathrm{FB}\). Consequently,
\[
\frac{\mathrm{SB}_{F_S^{\text{ir}},F_B}}{\mathrm{FB}_{F_S^{\text{ir}},F_B}}
\le
\frac{\mathrm{SB}_{F_S,F_B}}{\mathrm{FB}_{F_S,F_B}}.
\]
\end{lemma}

Combining the support normalization, the point-wise monotonicity reduction,
and the ironing reduction, the universal efficiency lower bound is reduced to
checking the following minimax inequality over \(\alpha\)-weakly regular seller
distributions:

\[
\inf_{\alpha\ge 0}
\ \inf_{W\ge0}
\ \inf_{\substack{
    \alpha\text{-weakly regular }F_S\\
    F_S\in\Delta(\mathbb Z_{W})
}}
\ \inf_{F_B\in\Delta(\mathbb Z_{W})}
\ \sup_{x\in\mathcal R_2(W)}
\ \big[
\mathrm{GFT}(x)+\alpha\Lambda_D(x)-\beta\cdot\mathrm{FB}
\big]
\ge 0.
\]

\subsection{The Minimax Swap}
To establish the universal lower bound $\beta$, we apply Sion's Minimax Theorem~\cite{sion1958general} to the inner optimization. Note that the space of allocations ${\cal R}_2$ is compact and convex, the domain of buyer distributions $\Delta(\mathbb{Z}_{W})$ is compact and convex, and our objective function below is a linear function of the probabilities of the discrete values, satisfying the necessary continuity and curvature requirements.
\begin{align*}
    &\mathrm{GFT}(x)+\alpha\cdot\Lambda_D(x)-\beta\cdot \mathrm{FB} \\
    =&\sum_{s,b} x(s,b)\,f_S(s)\,f_B(b)\,(b-s) +\alpha\cdot\sum_{s,b} x(s,b)
    \left[
        f_S(s)f_B(b)(b-s)
        -f_S(s)\sum_{b'>b} f_B(b')
        -f_B(b)F_S(s{-}1)
    \right] \\
    &-\beta\cdot\sum_{b}\sum_{s<b}(b-s)f_B(b)f_S(s).
\end{align*}

\begin{theorem}[Sion's Minimax Theorem~\cite{sion1958general}]
\label{thm:sion}
Let $X$ be a compact convex subset of a linear topological space, and let $Y$ be a convex subset of a linear topological space. Let
\[
f:X\times Y\to \mathbb{R}
\]
be a function such that:

\begin{enumerate}
    \item for every $y\in Y$, the function $x\mapsto f(x,y)$ is upper semicontinuous and quasiconcave on $X$;
    \item for every $x\in X$, the function $y\mapsto f(x,y)$ is lower semicontinuous and quasiconvex on $Y$.
\end{enumerate}

Then
\[
\sup_{x\in X}\inf_{y\in Y} f(x,y)
=
\inf_{y\in Y}\sup_{x\in X} f(x,y).
\]
\end{theorem}
Applying Theorem~\ref{thm:sion} to our target inequality, we can swap the inner minimization and maximization operators. Then, the ratio $\beta$ is a lower bound if and only if the following expression holds:
\begin{align*}
    \inf_{\alpha\ge0}
    \ \inf_{W\ge0}
    \ \inf_{\substack{
    \alpha\text{-weakly regular }F_S\\
    F_S\in\Delta(\mathbb Z_{W})
    }}
    \ \sup_{x\in\mathcal{R}_2(W)}
    \ \inf_{ F_B\in\Delta(\mathbb{Z}_{W})}
    \ \big[\mathrm{GFT}(x)+\alpha\cdot\Lambda_D(x)-\beta\cdot \mathrm{FB}\big]\ge 0.
\end{align*}
The preceding inequality is formulated for a fixed finite truncation
\(\mathbb Z_W\). We now remove this truncation by proving a stronger
statement. For any fixed \(W\) and any
\(F_S\in\Delta(\mathbb Z_W)\), an allocation rule
\(x\in\mathcal R_2\) defined on
\(\mathrm{Supp}(F_S)\times\mathbb Z_{\ge0}\) restricts to an element of
\(\mathcal R_2(W)\) on \(\mathrm{Supp}(F_S)\times\mathbb Z_W\), because
point-wise monotonicity is preserved under restriction. Moreover,
\(\Delta(\mathbb Z_W)\subseteq \Delta(\mathbb Z_{\ge0})\), so any guarantee
that holds for all buyer distributions on \(\mathbb Z_{\ge0}\) also holds for
all buyer distributions supported on \(\mathbb Z_W\). Hence it suffices to
prove
\begin{align}\label{eq:sion}
    \inf_{\alpha\ge0}
    \ \inf_{
    \alpha\text{-weakly regular }F_S}\ \sup_{x\in\mathcal{R}_2}
    \ \inf_{ F_B\in\Delta(\mathbb{Z}_{\ge0})}
    \ \big[\mathrm{GFT}(x)+\alpha\cdot\Lambda_D(x)-\beta\cdot \mathrm{FB}\big]
    \ge 0.
\end{align}

This reformulation shifts our perspective to a mechanism design game: for any fixed $\alpha$ and $\alpha$-weakly regular seller distribution $F_S$, we must construct a single allocation rule $x\in{\cal R}_2$ that satisfies the $\beta$-efficiency bound for all possible buyer distributions $F_B$. If such a robust allocation can be constructed, then the ratio $\beta$ is established as a universal lower bound on the second-best benchmark.

\subsection{The Greedy Algorithm}
\label{sec:canonical-recursion}

Having fixed the multiplier $\alpha \ge 0$ and an $\alpha$-weakly regular seller distribution $F_S$, we now direct our attention to the inner minimax problem:
\[
\sup_{x\in\mathcal R_2}\;
\inf_{F_B\in\Delta(\mathbb Z_{\ge 0})}
\Bigl[\mathrm{GFT}(x)+\alpha\cdot \Lambda_D(x)-\beta\cdot \mathrm{FB}\Bigr].
\]

At first sight, this is a global optimization problem, since the allocation rule \(x\) must simultaneously guarantee the desired lower bound for every buyer distribution \(F_B\). The key observation, however, is that once \(F_S\) is fixed, the bracketed objective is linear with respect to the buyer's discrete probability masses $f_B$. Therefore, the adversarial choice of \(F_B\) can be reduced to a family of local constraints indexed by the buyer value \(b\). In this way, the universal requirement over all buyer distributions is transformed into a buyer-wise feasibility condition.

The formal argument proceeds in three steps. First, we decompose the global
objective into buyer-wise coefficients. Second, we use these coefficients to
derive a sequential construction of the allocation rule. Finally, we prove
that this greedy construction is equivalent to the original maximin problem:
the algorithm succeeds if and only if there exists an allocation rule that certifies the desired lower bound.

\paragraph{Buyer-Wise Decomposition.}
We begin with the first step. Observe that, for a fixed allocation \(x\), the quantity
\(\mathrm{GFT}(x)+\alpha\Lambda_D(x)-\beta\,\mathrm{FB}\)
is linear in the probabilities \(f_B\). We may therefore represent it in the following form
\begin{equation}
\mathrm{GFT}(x)+\alpha\Lambda_D(x)-\beta\,\mathrm{FB}
=
\sum_{b\ge 0} f_B(b)\,H(b;x),
\label{eq:buyer-wise-decomp}
\end{equation}

\vspace{-0.5cm}
where
\begin{equation}
H(b;x)
:=
\sum_{0\le s<b} f_S(s)\Biggl[
(b-s)\left((1+\alpha)x(s,b)-\beta\right)
-\alpha\sum_{s<s'<b}x(s',b)
-\alpha\sum_{s<b'<b}x(s,b')
\Biggr].
\label{eq:Hb-def}
\end{equation}

Because the probability masses $f_B(b)$ are non-negative and sum to one, the global objective is non-negative for every valid buyer distribution if and only if every term in the summation is individually non-negative. This yields the fundamental equivalence:

\begin{equation}
\mathrm{GFT}(x)+\alpha\Lambda_D(x)-\beta\,\mathrm{FB}\ge 0,
\; \forall F_B\in \Delta(\mathbb Z_{\ge 0})
\quad\Longleftrightarrow\quad
H(b;x)\ge 0,\; \forall b\in \mathbb Z_{\ge 0}.
\label{eq:decomp-iff}
\end{equation}

Equation~\eqref{eq:decomp-iff} converts the original distributionally robust
constraint into a sequence of local feasibility requirements. It remains to
construct a monotone allocation rule $x\in\mathcal R_2$ satisfying all of
them. This suggests a natural greedy procedure. We process buyer values in
ascending order. When reaching value $b$, the profiles $x(\cdot,b')$ for
all $b'<b$ have already been fixed, so we need to determine the current profile
$x(\cdot,b)$. We first copy the previous profile
$x(\cdot,b-1)$ to preserve monotonicity. If this already gives
$H(b;x)\ge 0$, no change is needed. Otherwise, we increase the allocation at
the current buyer value by the smallest amount that makes the local
constraint tight.

The $\alpha$-weak regularity of $F_S$ ensures that this minimal correction
can be searched for in threshold form. This leads to the following greedy
construction.

\begin{algorithm}[H]
\LinesNumbered
\caption{Greedy Construction}
\label{alg:canonical-recursion}
\KwIn{The seller distribution $F_S$; parameters $\alpha\ge 0$ and $\beta\in(0,1)$.}
\KwOut{The final allocation rule $x$ or {\sc fail}}
\SetKwFunction{FMain}{Greedy Construction}
\SetKwProg{Fn}{Function}{:}{}
\Fn{\FMain}{
    Set \(x(\cdot,0) \equiv 0\) \;
    Set \(b \gets 1\), $t_0 \gets 0$ \;
    \While{true}{
        Set \(x(\cdot,b) \gets x(\cdot,b-1)\) and compute \(H(b;x)\) \;
        \If{\(H(b;x) \ge 0\)}{
            Set \(b \gets b+1\) \;
        }
        \Else{
            Set found \(\gets\) false \;
            \For{\(t_b = t_{b-1}, t_{b-1}+1, \dots, b-1\)}{
                \If{\((1+\alpha) b \le \phi_S(t_{b})\)}{
                \Return {\sc fail} \;
                }
                Choose \(x(t_b,b)\) such that \(H(b;x)=0\) \;
                \If{\(x(t_b,b) \in (0,1]\)}{
                    Set found \(\gets\) true \;
                    \textbf{break} \;
                }
                Set \(x(t_b,b) \gets 1\) \;
            }
            \If{found = true\nllabel{line:found-true}}{
                Set \(b \gets b + 1\) \;
            }
            \Else{
                \Return {\sc fail} \;
            }
        }
    }
}
\end{algorithm}

Algorithm~\ref{alg:canonical-recursion} is designed so that, whenever it successfully passes buyer
value $b$, the local constraint $H(b;x)\ge 0$ is satisfied while monotonicity
of the allocation is preserved. Therefore, if Algorithm~\ref{alg:canonical-recursion} never returns
\textsc{fail}, it directly constructs an allocation rule satisfying all
buyer-wise inequalities in \eqref{eq:decomp-iff}. The converse direction is
less immediate: we also need to show that if the greedy construction fails,
then no feasible monotone allocation rule can satisfy the remaining local
constraints. The next theorem formalizes this exact equivalence.

\begin{theorem}[Maximin-Greedy Equivalence]
\label{thm:minimax-recursion-equivalence}
Fix \(\alpha\ge 0\) and \(\beta\ge 0\), and let \(F_S\) be an $\alpha$-weakly regular seller distribution. Then 
    \[
    \sup_{x\in\mathcal R_2}\inf_{F_B\in\Delta(\mathbb Z_{\ge 0})}
    \left[
    \mathrm{GFT}(x)+\alpha\Lambda_D(x)-\beta\cdot \mathrm{FB}
    \right]\ge 0\quad\Longleftrightarrow \quad\text{The greedy algorithm never returns \textsc{fail}.}
    \]
\end{theorem}

\section{\texorpdfstring{Lower Bound: $1/2$ Efficiency is Achievable}{Lower Bound: half Efficiency is Achievable}}

\label{sec:optimality}
In this section, we prove $\mathrm{SB}\ge \tfrac12\,\mathrm{FB}$. Relying on our minimax formulation in Eq.~\eqref{eq:sion}, it is sufficient to show that for any $\alpha \ge 0$ and any $\alpha$-weakly regular seller distribution, the inequality
\[
\sup_{x\in\mathcal{R}_2}\inf_{F_B\in\Delta(\mathbb{Z}_{\ge 0})}
\bigl[\mathrm{GFT}(x)+\alpha\cdot\Lambda_D(x)-\beta\cdot \mathrm{FB}\bigr]\ge 0
\]
holds with $\beta=\tfrac12$. Theorem~\ref{thm:minimax-recursion-equivalence} (Minimax-Greedy Equivalence) allows us to translate this global optimization bound into an algorithmic feasibility question. Specifically, the bound holds if and only if Algorithm~\ref{alg:canonical-recursion} avoids returning \textsc{fail} when initialized with $\beta = \tfrac12$.

We introduce in Section~\ref{sec:state} several important concepts used to analyze the greedy algorithm, including a state representation and a state-based reformulation of Algorithm~\ref{alg:canonical-recursion}, which we call Algorithm~\ref{alg:state-recursion}.
Section~\ref{sec:t1} provide tools, primarily establishing several reduction results. In Section~\ref{sec:general}, we prove that Algorithm~\ref{alg:state-recursion} does not fail. The proof proceeds by contradiction: we assume that a failure scenario exists, and then show that from such a scenario we can identify a ``shorter'' counterexample. Additional proofs and supplementary analysis can be found in Appendix~\ref{app:optimality}.

\subsection{State and Dominance Order}
\label{sec:state}

We first introduce a compact \emph{state} representation showing that Algorithm~\ref{alg:canonical-recursion}'s progress at each \emph{clear step} (a stage where the constraint binds tightly, $H(b;x) = 0$) is completely characterized by a finite-dimensional vector. We then equip the state space with a \emph{dominance order}. This allows us to establish a monotonicity result (Proposition~\ref{prop:mono}) that helps us determine the exact boundary of the set of all achievable states.

\paragraph{State Representation.}

By Algorithm~\ref{alg:canonical-recursion}, whenever a buyer value $b$ is successfully processed (i.e., the condition {\em found = true} is met in Line~\ref{line:found-true}), the resulting allocation vector $x(\cdot,b)$ exhibits a strict threshold structure. Specifically, there is a pair \((t_b,y)\) with $t_b \in \{0, \dots, b-1\}$ and $y \in (0, 1]$ such that:
\[
x(s,b) =
\begin{cases}
1, & s < t_b,\\
y\in(0,1], & s = t_b,\\
0, & s > t_b.
\end{cases}
\]
To analyze whether Algorithm~\ref{alg:canonical-recursion} can execute perpetually over the unbounded domain (since $b$ can reach $\infty$) without returning \textsc{fail}, we require a compact representation of the information carried from one step to the next. While Algorithm~\ref{alg:canonical-recursion}'s decisions formally depend on the full history of the constructed allocation rule, the threshold structure of $x(\cdot,b)$ 
ensures that it is enough to carry forward the allocation probability at the
threshold. This structural insight motivates the introduction of a \emph{state} variable that succinctly encodes this critical value.

For a seller distribution, we use zero-based seller-side ratio indices throughout:
\(r_i:=\frac{F_S(i-1)}{f_S(i)},\  i\ge0,\)
with the convention \(F_S(-1)=0\). Hence \(r_0=0\).  In particular, the
ratio suffix appearing in a state always starts from the current threshold
index \(t_b\), not from a locally renumbered coordinate.

\begin{definition}[State and Clear State]
\label{def:state}
We let
\(
\overline{\mathbb R}_{\ge0}:=\mathbb R_{\ge0}\cup\{\infty\}
\)
and
\(
\overline{\mathbb R}_{>0}:=\mathbb R_{>0}\cup\{\infty\}.
\)
An indexed \emph{state} is an expression of the form
\(\omega=(y\,;\,r_t,r_{t+1},\ldots,r_m), \  y\in(0,1],\)
where \(0\le t\le m\), \(r_t\in\overline{\mathbb R}_{\ge0}\), and
\(r_i\in\overline{\mathbb R}_{>0}\) for all \(i>t\).

We say a state \(\omega\) is \emph{achievable} if there exist an
\(\alpha\)-weakly regular seller distribution \(F_S\) with mass sequence \(f_S\)
and a buyer value \(b\) such that Algorithm~\ref{alg:canonical-recursion}
successfully processes \(b\) without returning \textsc{fail}
(i.e., \emph{found = true} holds in Line~18) and, with \(m=b-1\) and
\(t=t_b\), the components of \(\omega\) satisfy
\(y=x(t_b,b), \qquad r_i=\frac{F_S(i-1)}{f_S(i)} \quad\text{for }i=t_b,\ldots,b-1.\)
Furthermore, if a state is achieved with condition \(H(b;x)=0\), we refer to
\(\omega\) as a \emph{clear state}, and the algorithmic transition yielding it
as a \emph{clear step}.
\end{definition}

For any clear state reached at buyer value \(b\) with threshold index \(t_b\),
the vector \((r_{t_b},r_{t_b+1},\ldots,r_{b-1})\) represents the seller-side
inverse hazard ratios of the seller types at or above the threshold index. In
a clear state, these quantities together with the threshold allocation \(y\)
form a complete state that uniquely determines all future behavior of
Algorithm~\ref{alg:canonical-recursion}.

\paragraph{Base Clear State.}
At \(b=1\), the only seller type satisfying \(s<b\) is \(s=0\), so the
clear-state condition \(H(1;x)=0\) forces
\(x(0,1)=\frac{\beta}{1+\alpha}.\)
Since \(r_0=F_S(-1)/f_S(0)=0\), this determines the distinguished base clear
state
\(\omega_0:=\left(\frac{\beta}{1+\alpha}\,;\,r_0\right), \qquad r_0=0.\)
We refer to \(\omega_0\) as the \emph{base clear state}.

In particular, suppose that Algorithm~\ref{alg:canonical-recursion} achieves
a clear state at buyer value \(b_0\), and let \(t_0:=t_{b_0}\) denote its
threshold index.  As Algorithm~\ref{alg:canonical-recursion} proceeds to
larger buyer values \(b>b_0\), the allocation rule changes only at the current
buyer value, while all previously constructed vectors \(x(\cdot,b')\),
\(b'\le b_0\), are left unchanged.  Hence \(H(b_0;x)=0\) remains true at all
later steps.  For any buyer value \(b>b_0\), it suffices to study the
difference \(H(b;x)-H(b_0;x)\).

After normalizing by \(F_S(t_0)\), define the normalized discriminant
\(\hat H_{t_0,b_0}(b;x)\) as
\begin{equation}
\label{eq:Hhat}
\begin{split}
    &\hat H_{t_0,b_0}(b;x)
    :={}
    \Bigg[(b-b_0)(1-\beta)-\alpha\sum_{t_0\le t<b}x(t,b)
    +\alpha x(t_0,b_0)\Bigg]\frac{r_{t_0}}{1+r_{t_0}}
    \\
    &+\Bigg[(1+\alpha)\Big((b-t_0)x(t_0,b)-(b_0-t_0)x(t_0,b_0)\Big)
    -(b-b_0)\beta
    -\alpha\sum_{t_0<s'<b}x(s',b)
    -\alpha\sum_{b_0\le b'<b}x(t_0,b')
    \Bigg]\frac{1}{1+r_{t_0}}
    \\
    &+\sum_{t_0<s<b_0}\Bigg[(1+\alpha)(b-s)x(s,b)-(b-b_0)\beta
    -\alpha\sum_{s<s'<b}x(s',b)
    -\alpha\sum_{b_0\le b'<b}x(s,b')
    \Bigg]
    \prod_{t_0<i<s}\left(1+\frac{1}{r_i}\right)\frac{1}{r_s}
    \\
    &+\sum_{b_0\le s<b}\Bigg[(b-s)\Big((1+\alpha)x(s,b)-\beta\Big)
    -\alpha\sum_{s<s'<b}x(s',b)
    -\alpha\sum_{b_0\le b'<b}x(s,b')
    \Bigg]
    \prod_{t_0<i<s}\left(1+\frac{1}{r_i}\right)\frac{1}{r_s}.
\end{split}
\end{equation}
Here an empty product is interpreted as \(1\).  The important point is that
\(r_i\) is indexed by the seller value \(i\) itself.  Thus, when the latest
clear state has threshold \(t_0\), the active ratio suffix is
\((r_{t_0},r_{t_0+1},\ldots)\).

\begin{claim}[State Sufficiency of Clear States]
\label{clm:state-sufficiency}
Let \(\omega=(y;r_{t_0},\ldots,r_{b_0-1})\) be a clear state reached at buyer
value \(b_0\) with threshold index \(t_0\). Then the execution of
Algorithm~\ref{alg:canonical-recursion} after step \(b_0\) is determined
entirely by \(\omega\) and the future ratio tail
\((r_{b_0},r_{b_0+1},\ldots)\), and is independent of the allocation history
before step \(b_0\).
\end{claim}

By Claim~\ref{clm:state-sufficiency}, for the purpose of analyzing whether the greedy
construction can continue without returning \textsc{fail}, it is convenient
to extract this useful information from the full informational
description of Algorithm~\ref{alg:canonical-recursion}.

The following procedure is precisely this compressed version of the greedy
construction.  Its input consists of an initial clear state
\(\omega^{(0)}=(y^{(0)};\mathbf r^{(0)})\) and an additional ratio sequence
\(\mathbf r^{(1)}\) to be processed.  Algorithm~\ref{alg:state-recursion} concatenates these two
ratio sequences and simulates the same greedy update rule as
Algorithm~\ref{alg:canonical-recursion}, but it uses 
the normalized function \(\hat H(b;x)\) instead of \(H(b;x)\) and updates $b_0$ which represents the buyer value corresponding to the latest clear state.
We use $\mathrm{ALG}_2$ to denote the procedure implemented by Algorithm~\ref{alg:state-recursion}.

\begin{algorithm}[H]
\caption{State Transition}
\LinesNumbered
\label{alg:state-recursion}
\KwIn{Initial clear state \( \omega_{\mathrm{init}} = \bigl(y^{(0)}; \mathbf r^{(0)}\bigr) \); input sequence \( \mathbf r^{(1)} \)}
\KwOut{Output state \( \omega_{\mathrm{output}} = \bigl(x(t_{|\mathbf R|+1}, |\mathbf R|+1); \mathbf R[t_{|\mathbf R|+1}:]\bigr) \) or \textsc{fail}}
\SetKwFunction{FMain}{State Transition}
\SetKwProg{Fn}{Function}{:}{}
\Fn{\FMain}{
    Set \( \mathbf R \gets \mathbf r^{(0)} \| \mathbf r^{(1)} \) \tcp*{global ratio sequence}
    Set \( b \gets |\mathbf r^{(0)}|+2 \), \( b_0\gets |\mathbf r^{(0)}|+1 \), and \( \ell\gets 0 \) \tcp*{\(\ell=t_{b_0}\)}
    Set \( x(0,b{-}1) \gets 1 \), \( x(1,b{-}1) \gets y^{(0)} \), and \( t_{b-1}\gets 1 \)\;
    \While{\( b \le |\mathbf R|+1 \)}{
        Set \( x(\cdot,b) \gets x(\cdot,b-1) \) and \( t_b\gets t_{b-1} \)\;
        Compute \( \hat H_{\ell,b_0}(b;x) \) using the active ratio sequence \( \mathbf R[\ell:] \)\;
        \If{\( \hat H_{\ell,b_0}(b;x) \ge 0 \)}{
            \If{\( \hat H_{\ell,b_0}(b;x)=0 \)}{
                Set \( b_0\gets b \) and \( \ell\gets t_b \)\;
            }
            Set \( b\gets b+1 \)\;
        }
        \Else{
            Set \( \mathrm{found} \gets \mathrm{false} \)\; \nllabel{line:state-transition-init-found}
            \For{\( t_b = t_{b-1}, t_{b-1}+1, \dots, b-1 \)}{
            \nllabel{line:state-transition-threshold-search}
                \If{\( (1+\alpha)b \le \phi^\alpha_{\mathbf R}(t_b) \)}{
                    \nllabel{line:state-transition-virtual-test}
                    \Return \textsc{fail}\;\nllabel{line:state-transition-return-fail-virtual}
                }
                Choose \( x(t_b,b) \) such that
                \(\hat H_{\ell,b_0}(b;x)=0\)
                using the active ratio sequence \( \mathbf R[\ell:] \)\;
                \If{\( x(t_b,b) \in (0,1] \)}{
                    Set \( \mathrm{found} \gets \mathrm{true} \)\;
                    \textbf{break}\;
                }
                Set \( x(t_b,b) \gets 1 \)\;\nllabel{line:state-transition-update-kb}
            }
            \If{\( \mathrm{found}=\mathrm{true} \)}{
                Set \( b_0\gets b \), \( \ell\gets t_b \), and \( b\gets b+1 \)\;
            }
            \Else{
                \Return \textsc{fail}\;\nllabel{line:state-transition-no-found-fail}
            }
        }
    }
    \Return \( \bigl(x(t_{|\mathbf R|+1}, |\mathbf R|+1); \mathbf R[t_{|\mathbf R|+1}:]\bigr) \)\;
}
\end{algorithm}

For clarity, we first introduce some basic operations on finite sequences.
For a sequence \(\mathbf a=(a_0,\ldots,a_m)\), we use zero-based indexing. We denote its length by \(|\mathbf a|\), its slice from position \(i\) to position \(j\) by \(\mathbf a[i:j]\), and its suffix starting from position \(i\) by \(\mathbf a[i:]\). Formally,
\[
|\mathbf a|=m,\qquad
\mathbf a[i:j]=(a_i,a_{i+1},\ldots,a_j),\qquad
\mathbf a[i:]=(a_i,a_{i+1},\ldots,a_m).
\]
We adopt the convention that a slice is empty whenever \(i>j\). For two sequences \(\mathbf a\) and \(\mathbf b\), we write
\(\mathbf a|\mathbf b\) for their concatenation, i.e., the sequence obtained by appending \(\mathbf b\) after \(\mathbf a\).

\paragraph{Dominance Order and Monotonicity.}

To compare states according to how close they bring Algorithm~\ref{alg:state-recursion} to failure, we introduce the following dominance order.

\begin{definition}[Dominance Order $\prec$]
\label{def:order}

Let
\(\omega=(y\,;\mathbf{r})\) and \(\omega'=(y';\mathbf{r}')\)
be two states with ratio sequences \(\mathbf{r}\) and \(\mathbf{r}'\), respectively. The two states are comparable if one of the two ratio sequences is a suffix of the other. We say that \(\omega\) is dominated by \(\omega'\) (\(\omega\prec\omega'\)) if either
(i) \(\mathbf{r}'\) is a proper suffix of \(\mathbf{r}\), or
(ii) \(\mathbf{r}=\mathbf{r}'\) and \(y<y'\). We call a state undominated if no other achievable state dominates it.

\end{definition}
Intuitively, $\omega\prec \omega'$ means that $\omega'$ is more
demanding: starting from $\omega'$ leaves Algorithm~\ref{alg:state-recursion} closer to failure
than starting from $\omega$. The following proposition shows that this
order is preserved throughout Algorithm~\ref{alg:state-recursion}'s execution.

\begin{proposition}[Monotonicity]
\label{prop:mono}
Let \(\omega\prec \omega'\) be two clear states, and fix the same ratio suffix \(\mathbf{r}=(r_t,\ldots,r_m)\). If \(\mathrm{ALG}_2(\omega';\mathbf{r})\neq\textsc{fail}\), then \(\mathrm{ALG}_2(\omega;\mathbf{r})\neq \textsc{fail}\) and
\(\mathrm{ALG}_2(\omega;\mathbf{r})\prec \mathrm{ALG}_2(\omega';\mathbf{r}).\)
In particular, if \(\mathrm{ALG}_2(\omega;\mathbf{r})= \textsc{fail}\), then
\(\mathrm{ALG}_2(\omega';\mathbf{r})=\textsc{fail}\).
\end{proposition}

Therefore, it suffices to focus on undominated states. For a ratio suffix
\(\mathbf r=(r_t,\ldots,r_m)\), let \(\mathcal S_{\mathbf r}\) denote the set
of achievable states whose ratio components are \(\mathbf r\), that is,
states of the form \((y\,;r_t,\ldots,r_m)\). The collection
\(\{\mathcal S_{\mathbf r}\}_{\mathbf r}\) partitions the set of achievable
states. Moreover, any two states in the same set \(\mathcal S_{\mathbf r}\)
are comparable under the dominance order. Hence, selecting the maximum state
from each nonempty \(\mathcal S_{\mathbf r}\) gives a collection that contains
all undominated states.
While proving that Algorithm~\ref{alg:state-recursion} does not return \textsc{fail}, we can also determine all undominated states.

\paragraph{Threshold-Stationary Property.}
For convenience of exposition, we introduce an important class of state transitions in this section.
Let $N:=\lvert\operatorname{supp}(F_S)\rvert\ge 1$ denote the support size of the seller's distribution $F_S$. 

We say that a sequence of state transitions starting from the base clear state $\omega_0$ in Algorithm~\ref{alg:state-recursion} is {\em threshold-stationary} if, at every buyer value \(b<N\),  Algorithm~\ref{alg:state-recursion} either returns \textsc{fail} at $t_b=1$, or finds an achievable state at \(t_b=0\). We impose no restriction on how Algorithm~\ref{alg:state-recursion} behaves at $b = N$.

In the threshold-stationary case, we are interested in the allocation probabilities \(x(0,b)\) for all \(b<N\). For notational simplicity, we write \(x_b:=x(0,b)\). 
Taking successive differences of the equations \(\hat H(b;x)=0\) for \(b=1,2,\ldots,N-1\) yields

\begin{equation}
\label{eq:big-jump-rec}
x_b
=
\Bigl(1-\frac{1}{b(1+\alpha)}\Bigr)x_{b-1}
+
\frac{\beta}{b(1+\alpha)}\frac{F_S(b-1)}{F_S(0)},
\qquad
x_1=\frac{\beta}{1+\alpha}.
\end{equation}

From Equation~\eqref{eq:big-jump-rec}, $x_b,\forall b$ is linear in $\beta$, we let $x_b=\beta G_b$. Unrolling the recursion \eqref{eq:big-jump-rec} gives
\begin{equation}
\label{eq:Gm-def}
G_b
=
\sum_{t=1}^{b}
\frac{F_S(t-1)}{t(1+\alpha)F_S(0)}
\prod_{j=t+1}^{b}\left(1-\frac{1}{j(1+\alpha)}\right),
\qquad b\ge 1.
\end{equation}

The next two lemmas record two basic consequences of this one-dimensional description. Lemma~\ref{lem:mono-G} shows that the normalized allocation level \(G_b\) is monotone in the buyer value. This monotonicity will be used later to compare states and to rule out certain shortest counterexamples. Lemma~\ref{lem:alpha-zero-never-fail} treats the boundary case \(\alpha=0\). In this case the recursion above is sufficiently explicit to show directly that the greedy construction cannot fail. Thus, in the subsequent analysis, we only consider the case \(\alpha>0\).

\begin{lemma}
\label{lem:mono-G}
The sequence $(G_b)_{b\ge 1}$ is strictly increasing.
\end{lemma}



\begin{lemma}[The case \(\alpha=0\)]
\label{lem:alpha-zero-never-fail}
Let \(\alpha=0\) and \(\beta=\tfrac12\). Starting from the base clear state,
Algorithm~\ref{alg:state-recursion} never returns \textsc{fail}.
\end{lemma}

\subsection{Threshold-stationary Case}
\label{sec:t1}

This subsection proves the single-ratio truncation step under the
threshold-stationary assumption. For an \(\alpha\)-weakly regular seller
distribution, the \(\alpha\)-virtual value \(\phi^\alpha_S\) is weakly increasing. We
first prove the comparison in the constant-\(\phi^\alpha_S\) case, where adjacent
support sizes can be compared explicitly, and then use it as the local
reduction for the weakly increasing case.
The non-threshold-stationary
case is handled separately in Section~\ref{sec:general}.

We call an execution feasible if Algorithm~\ref{alg:state-recursion} does not return \textsc{fail}. Although Algorithm~\ref{alg:state-recursion} could in principle update the threshold in many different ways,  Proposition~\ref{prop:regular-never-fail} shows that, for 
\(\beta=\tfrac12\), any failure can occur only at \(b=N\) and threshold index \(t_b>1\) 

\begin{proposition}[Feasibility at $t_b=1$]
    \label{prop:regular-never-fail}
    Let $\alpha > 0$ and $\beta = \tfrac12$. For the input base clear state $\omega_0$, Algorithm \ref{alg:state-recursion} never returns \textsc{fail} at threshold index $t_b = 1$.
\end{proposition}

The obstruction to failing at threshold index \(t_b=1\) is the
threshold-zero candidate. If the algorithm were to fail at \(t_b=1\), then
the first positive seller threshold would already have too large an $\alpha$-virtual
value, namely \(\phi^\alpha_S(1)\ge (1+\alpha)b\). By weak regularity, this forces
all earlier positive seller values \(1,\ldots,b-1\) to have virtual value at
least \((1+\alpha)b\). 
Lemma~\ref{lem:zero-ratio-case} shows that even this
extremal case keeps the threshold-zero allocation feasible for
\(\beta=\tfrac12\). Therefore the failure at \(t_b=1\) can never occur.

\begin{lemma}
\label{lem:zero-ratio-case}
For the constant-\(\phi^\alpha_S\) case with ratio parameter \(r=0\) and
\(\alpha>0\), the sequence
\(\{\widetilde G_{N-1}^{(N)}\}_{N\ge2}\) is strictly increasing with respect
to \(N\). Moreover,
\[
  \lim_{N\to\infty}\widetilde G_{N-1}^{(N)}
  \;=\;
  \frac{1}{1+\alpha}\,
  \frac{\Gamma\!\bigl(\tfrac{1}{1+\alpha}\bigr)^{2}}
       {\Gamma\!\bigl(\tfrac{2}{1+\alpha}\bigr)},
\qquad
  \sup_{\alpha>0}\;\frac{1}{1+\alpha}\,
  \frac{\Gamma\!\bigl(\tfrac{1}{1+\alpha}\bigr)^{2}}
       {\Gamma\!\bigl(\tfrac{2}{1+\alpha}\bigr)}
  =2,
\]
where the supremum is attained only in the limit \(\alpha\to\infty\).
\end{lemma}

The intuition is as follows.  To compare
\(\widetilde G_N^{(N+1)}\) with \(\widetilde G_{N-1}^{(N)}\), we express both
quantities with compatible product weights.  Although the common summands are
not monotone term by term, their changes can be combined with the additional final-step contribution in the \((N+1)\)-case, and the resulting coefficients are
strictly positive.  This proves strict monotonicity in \(N\).  The limit is
then obtained by converting the products into Gamma ratios, which converge to
a Beta integral; the resulting Gamma expression is increasing in \(\alpha\)
and approaches \(2\) as \(\alpha\to\infty\).

By Proposition~\ref{prop:regular-never-fail}, the threshold-stationary
execution cannot fail at threshold index \(t_b=1\). We now compare two adjacent
executions by deleting the first positive seller-side ratio.

\paragraph{Index Convention for Single-Ratio Truncation.}

Throughout the single-ratio truncation argument, the long execution has
seller-support length $k+2$ and the short execution has seller-support length
$k+1$.  We mark the ratio coordinates by the corresponding support
length:
\(r_i^{(k+2)}, i=0,1,\ldots,k+1\)
for the $k+2$ case, and
\(r_i^{(k+1)}, i=0,1,\ldots,k\)
for the $k+1$ case.  The base ratio is zero in both cases:
\(r_0^{(k+2)}=r_0^{(k+1)}=0.\)
The short case is obtained by deleting the first positive ratio of the long
case.  Thus
\begin{equation}
\label{eq:single-ratio-truncation-zero-based}
    r_i^{(k+1)}=r_{i+1}^{(k+2)},
    \qquad i=1,\ldots,k.
\end{equation}
Equivalently, the inputs after the base clear state are
\( r_1^{(k+2)},r_2^{(k+2)},\ldots,r_{k+1}^{(k+2)}\)
and \\
\( r_1^{(k+1)},r_2^{(k+1)},\ldots,r_k^{(k+1)}.\)
The zero ratio $r_0^{(k+2)}=r_0^{(k+1)}=0$ is already included in
$\omega_0$.

If the $k+2$ case reaches buyer value $b=k+2$ with threshold index $t_b=t$,
then its state is written as
\(\omega=(y;r_t^{(k+2)},r_{t+1}^{(k+2)},\ldots,r_{k+1}^{(k+2)}).\)
The corresponding artificial state in the $k+1$ case has threshold index
$t-1$ and ratio suffix
\((r_{t-1}^{(k+1)},r_t^{(k+1)},\ldots,r_k^{(k+1)}).\)
By Equation \eqref{eq:single-ratio-truncation-zero-based}, these are the same positive
seller-side ratios whenever $t\ge2$:
\(r_{t-1+i}^{(k+1)}=r_{t+i}^{(k+2)}, \ i=0,1,\ldots,k+1-t.\)
Thus the comparison is a suffix comparison, while the state notation
still starts from the actual threshold index in each case.

\paragraph{Artificial Allocation for Single-Ratio Truncation Arguments.}

Let
\(\omega=(y;r_t^{(k+2)},r_{t+1}^{(k+2)},\ldots,r_{k+1}^{(k+2)})\)
be a candidate state generated by the $k+2$ case at $b=k+2$, with threshold
index $t\ge1$.  The artificial allocation \(x_\omega^{(k+1)}\) in the $k+1$ case is obtained as
follows.  For buyer values $b<k+1$, use the allocation generated from the
short ratio sequence
$(r_1^{(k+1)},\ldots,r_k^{(k+1)})$.  At buyer value $b=k+1$, set
\[
x_\omega^{(k+1)}(s,k+1)
=
\begin{cases}
1, & s<t-1,\\
y, & s=t-1,\\
0, & s>t-1.
\end{cases}
\]
This artificial allocation need not satisfy
$\hat H(k+1;x_\omega^{(k+1)})=0$.

\paragraph{Inequality for the Threshold-Stationary Case.}

For $i\ge1$, define the virtual-value functions associated with the two ratio
sequences by
\[
    \phi^{\alpha(k+2)}(i):=(1+\alpha)i+\alpha r_i^{(k+2)},
    \qquad
    \phi^{\alpha(k+1)}(i):=(1+\alpha)i+\alpha r_i^{(k+1)}.
\]
Because of Equation \eqref{eq:single-ratio-truncation-zero-based}, for every
$\ell=1,\ldots,k$,
\begin{equation}
\label{eq:phi-shift-zero-based}
    (1+\alpha)(k+2)-\phi^{\alpha(k+2)}(\ell+1)
    =
    (1+\alpha)(k+1)-\phi^{\alpha(k+1)}(\ell).
\end{equation}

Let the $k+2$ case generate, at $b=k+2$, the state
\(\omega=(y;r_t^{(k+2)},\ldots,r_{k+1}^{(k+2)}), \  t\ge2.\)
The corresponding artificial threshold in the $k+1$ case is $t-1$. For convenience, define
for $\ell=1,\ldots,t-1$,
\begin{equation}
\label{eq:delta-zero-based}
\delta_\ell
:=
\Big((1+\alpha)(k+1)-\phi^{\alpha(k+1)}(\ell)\Big)
\prod_{i=1}^{\ell-1}\left(1+\frac1{r_i^{(k+1)}}\right)
\frac1{r_\ell^{(k+1)}}.
\end{equation}
The coefficient in $\delta_\ell$ is positive for each relevant $\ell$, since
otherwise the $k+2$ execution would have returned \textsc{fail} at the
corresponding threshold by \eqref{eq:phi-shift-zero-based}.

Let
\(p_{k+2}:=x^{(k+2)}(0,k+1), \  p_{k+1}:=x^{(k+1)}(0,k),\)
where $p_{k+2}$ is the threshold-zero allocation inherited by the $k+2$ case
from buyer value $k+1$, and $p_{k+1}$ is the analogous quantity in the
$k+1$ case from buyer value $k$.  A direct expansion of the clear-step
expression gives
\begin{align}
\hat H(k+2;x^{(k+2)})
={}&
 p_{k+2}
 -\beta\prod_{i=1}^{k+1}\left(1+\frac1{r_i^{(k+2)}}\right)
 +(1-p_{k+2})(k+2)(1+\alpha)
 +\left((1+\alpha)(k+1)-\alpha r_1^{(k+2)}\right)
 \frac1{r_1^{(k+2)}}
\nonumber\\
&+
\underbrace{\left(1+\frac1{r_1^{(k+2)}}\right)
\sum_{\ell=1}^{t-2}\delta_\ell
+
\left(1+\frac1{r_1^{(k+2)}}\right)\delta_{t-1}y}
_{\text{candidate-dependent terms}},
\label{eq:common-suffix-kplus2-zero-based}
\\[0.3em]
\hat H(k+1;x_\omega^{(k+1)})
={}&
 p_{k+1}
 -\beta\prod_{i=1}^{k}\left(1+\frac1{r_i^{(k+1)}}\right)
 +(1-p_{k+1})(k+1)(1+\alpha)
+
\underbrace{\sum_{\ell=1}^{t-2}\delta_\ell+
\delta_{t-1}y}_{\text{candidate-dependent terms}}.
\label{eq:common-suffix-kplus1-zero-based}
\end{align}
Since $\hat H(k+2;x^{(k+2)})=0$, to prove
$\hat H(k+1;x_\omega^{(k+1)})<0$, it suffices to show
\begin{equation}
\label{eq:threshold-stationary-same-state-positive-zero-based}
\hat H(k+2;x^{(k+2)})
-
\left(1+\frac1{r_1^{(k+2)}}\right)
\hat H(k+1;x_\omega^{(k+1)})
\ge0.
\end{equation}
Using \eqref{eq:common-suffix-kplus2-zero-based} and
\eqref{eq:common-suffix-kplus1-zero-based}, the candidate-dependent terms
cancel.  Therefore \eqref{eq:threshold-stationary-same-state-positive-zero-based}
is equivalent to $\mathcal D\ge0$, where
\begin{equation}
\label{eq:D-zero-based}
\mathcal D
:=
\left(1+\frac1{r_1^{(k+2)}}\right)
\Big((k+1)(1+\alpha)-1\Big)p_{k+1}
+
1
-
\Big((k+2)(1+\alpha)-1\Big)p_{k+2}.
\end{equation}
If the $k+1$ case later outputs a state with the same ratio suffix, the
positive coefficient of the boundary variable implies that its boundary
allocation is strictly larger than $y$.

\begin{proposition}[Dominance under Threshold-Stationarity]
\label{prop:success-stationary-phi}
Assume $\alpha>0$.  Let the $k+2$ and $k+1$ ratio sequences satisfy the
zero-based truncation convention
\eqref{eq:single-ratio-truncation-zero-based}.  Suppose the execution
\[
    \mathrm{ALG}_2\bigl(\omega_0;
    r_1^{(k+2)},\ldots,r_{k+1}^{(k+2)}\bigr)
\]
is threshold-stationary.  Then the following hold.
\begin{enumerate}
    \item If the $k+2$ execution outputs a state
    \(\omega=(y;r_t^{(k+2)},\ldots,r_{k+1}^{(k+2)})\)
    with $t\ge2$, then either the $k+1$ execution returns \textsc{fail}, or
    the output of the $k+1$ execution dominates $\omega$ after identifying
    $r_i^{(k+1)}$ with $r_{i+1}^{(k+2)}$ for $i=1,\ldots,k$.

    \item If the $k+2$ execution returns \textsc{fail}, then the $k+1$
    execution also returns \textsc{fail}.
\end{enumerate}
\end{proposition}

The proof proceeds by minimizing \(\mathcal{D}\) over the admissible choices of weakly increasing \((\phi^\alpha_0,\ldots,\phi^\alpha_k)\). Suppose that, at a minimizer, the \(\phi_i\)'s are not all equal. Then the first-order coefficients of the perturbation are not balanced across coordinates. We can therefore choose an admissible perturbation that shifts mass from a coordinate with a larger marginal contribution to one with a smaller marginal contribution, keeping the required constraints unchanged while strictly decreasing \(\mathcal{D}\). This contradicts the choice of the minimizer. Thus the minimum is attained only in the constant-\(\phi_S\) case. The desired bound \(\mathcal{D}\ge 0\) then follows from the constant-\(\phi_S\) result in Lemma~\ref{lem:extremal-t1}.

For convenience, given $r,\alpha>0$, we denote
\(a:=\frac{\alpha}{1+\alpha}\in(0,1),\;\lambda:=ar.\)
\begin{lemma}
\label{lem:extremal-t1}
Let $\alpha>0$ and $r\ge0$.
Define the constant-$\phi^\alpha$ ratio sequences by
\( r_i^{(k+2)}=\frac{k+1-i+\lambda}{a}, \  i=1,\ldots,k+1,\)
and
\( r_i^{(k+1)}=\frac{k-i+\lambda}{a}, \  i=1,\ldots,k.\)
If the $k+2$ execution outputs a proper suffix state
$(y;r_t^{(k+2)},\ldots,r_{k+1}^{(k+2)})$ with $t\ge2$, then
$\mathcal D\ge0$.
\end{lemma}

Suppose the execution in \(k+2\) case produces a proper suffix state \(\omega\).
We compare it with the \(k+1\) case by placing the corresponding artificial
allocation. By Lemma~\ref{lem:1}, we derive an upper bound on the key quantity
\(\lambda\), which allows us to simplify the expression.
We then show that \(\mathcal{D}\ge0\) using the inequality of the Euler beta function.

\begin{lemma}
\label{lem:1}
Let $\alpha>0$, $r\ge0$. 
Let the $k+2$ constant-$\phi^\alpha$ ratio sequence be
\(r_i^{(k+2)}=\frac{k+1-i+\lambda}{a}, \ i=1,\ldots,k+1.\)
If the $k+2$ execution outputs a proper suffix state, then
$\lambda<1-a$.
\end{lemma}

The intuition is as follows. 
Suppose that Algorithm~\ref{alg:state-recursion} reaches a state whose ratio vector is a proper suffix. Then the following expression must satisfy a necessary lower bound.
This expression decreases as \(\lambda\) increases. To
establish \(\lambda < 1-a\), it suffices to substitute \(\lambda = 1-a\) into
the expression. At this cutoff, the product terms simplify considerably. A
Gamma--Beta estimate then shows that the expression falls below the required
lower bound. This contradiction implies that \(\lambda\) must be strictly less
than \(1-a\).

\subsection{The General Case and Main Theorem}
\label{sec:general}

Section~\ref{sec:t1} proves the single-ratio truncation argument when the
\(k+2\) case is threshold-stationary. In this subsection, we first generalize the argument by relaxing the threshold-stationary assumption. Then, we prove the never-fail result by contradiction. In addition, we provide the set of all undominated states.

\begin{proposition}[Single-ratio truncation: non-threshold-stationary with $t_b\le1$ case]
\label{prop:strict-decrease-non-stationary}
We assume \(\alpha>0\) and \(\beta=\tfrac12\). Let the \(k+2\) and \(k+1\)
ratio sequences satisfy the zero-based truncation convention
\eqref{eq:single-ratio-truncation-zero-based}. Suppose that the long execution
\(\mathrm{ALG}_2\bigl(\omega_0; r_1^{(k+2)},\ldots,r_{k+1}^{(k+2)}\bigr)\)
satisfies \(t_b\le1\) for every \(b<k+2\). Then the following hold.

\begin{enumerate}
    \item If the long execution outputs a proper suffix state, then either the
    short execution
    \(\mathrm{ALG}_2\bigl(\omega_0; r_1^{(k+1)},\ldots,r_k^{(k+1)}\bigr)\)
    returns \textsc{fail}, or the output of the long execution is dominated by
    the output of the short execution after identifying
    \(r_i^{(k+1)}=r_{i+1}^{(k+2)}\) for \(i=1,\ldots,k\).

    \item If the long execution returns \textsc{fail}, then the short execution
    also returns \textsc{fail}.
\end{enumerate}
\end{proposition}

In the threshold-stationary case, the allocation function has a relatively clean expression. But in the non-threshold-stationary case, the expression is involved. To overcome the difficulty, we construct the artificial \(k+1\) allocation used in Section~\ref{sec:t1}. We call buyer value \(b^\star\) a turning point
if the threshold index first moves away from \(0\) at that step; equivalently,
\(t_b=0\) for \(b<b^\star\) and \(t_{b^\star}>0\). Proposition~\ref{prop:success-stationary-phi}
supplies positivity of the residual measuring the gap between the normalized
 \(\hat H\) in the \(k+2\) case and the artificial  \(\hat H\) in the \(k+1\) case at that turning point. Proposition ~\ref{prop:strict-decrease-non-stationary} shows
that the recursion preserves this positivity.  Consequently, the same dominance and failure-propagation
conclusions continue to hold.

We now collect the local truncation results proved in the preceding
lemmas and propositions. Lemma~\ref{lem:extremal-t1} establishes the result in the
constant-\(\phi^\alpha_S\) case. Proposition~\ref{prop:success-stationary-phi}
then extends it to threshold-stationary executions, and
Proposition~\ref{prop:strict-decrease-non-stationary} further generalizes
it to a non-threshold-stationary setting. They imply that deleting a leading ratio cannot turn a
failing execution into a successful one; moreover, when both executions
are successful and the resulting states are comparable, the execution
with the shorter ratio sequence produces a more dominant state.

We use these two implications in two different ways.  Theorem~\ref{thm:never-fail}
uses the failure implication in a minimal-counterexample argument to rule out
the possibility that Algorithm~\ref{alg:state-recursion} returns
\textsc{fail}.  Theorem~\ref{thm:dominance} then iterates the dominance
implication to identify the undominated states reachable from the base clear
state.

\begin{theorem}[Algorithm never fails for $\beta=\tfrac{1}{2}$]
\label{thm:never-fail}
For any $\alpha$-weakly regular $F_S$ and $\beta=\tfrac{1}{2}$,
Algorithm~\ref{alg:state-recursion} never returns {\sc fail}.
\end{theorem}

\begin{proof}[Proof of Theorem~\ref{thm:never-fail}]
Fix \(\beta=\tfrac12\). If \(\alpha=0\), the claim follows from
Lemma~\ref{lem:alpha-zero-never-fail}. Hence assume \(\alpha>0\).

Suppose, toward a contradiction, that Algorithm~\ref{alg:state-recursion}
returns \textsc{fail} for some \(\alpha\)-weakly regular ratio sequence. Choose
a finite failing global ratio vector \((r_0,r_1,\ldots,r_m)\) of minimal
length, with \(r_0=0\), so that the actual input after the base state is
\(\mathrm{ALG}_2(\omega_0;r_1,\ldots,r_m)=\textsc{fail}.\)
By minimality, the failure occurs at \(b=m+1\) after this input tail is fully
processed; otherwise a proper prefix would already be a shorter counterexample.

First suppose that, before this failure, there is no successful clear step
with positive threshold index. Then the execution on \((r_0,r_1,\ldots,r_m)\)
is threshold-stationary. Applying Proposition~\ref{prop:success-stationary-phi}
to the first positive ratio gives a shorter truncated global vector
\((\bar r_0,\bar r_1,\ldots,\bar r_{m-1})\), where \(\bar r_0=0\) and
\(\bar r_i=r_{i+1}\) for \(i\ge1\), such that
\(\mathrm{ALG}_2(\omega_0;\bar r_1,\ldots,\bar r_{m-1})=\textsc{fail},\)
contradicting minimality.

Now suppose that a first successful clear step with positive threshold index
exists. If its threshold index is larger than \(1\), denote that step by
\(b=k+2\). Up to this step, the execution on the long block
\((r_0,r_1,\ldots,r_{k+1})\) is threshold-stationary, and its output has a
proper suffix after deleting at least one positive ratio. Proposition~\ref{prop:success-stationary-phi}
therefore implies either that the truncated block already fails, which is a
shorter counterexample, or that the truncated block produces a dominating
state. In the dominance case, we continue both states with the same remaining
ratios \(r_{k+2},\ldots,r_m\). By Claim~\ref{clm:state-sufficiency}, the run
from the original state is the original failing execution; by
Proposition~\ref{prop:mono}, the run from the dominating truncated state also
fails. This again gives a shorter counterexample.

It remains to consider the case in which the first positive threshold index is
\(1\). Following Proposition~\ref{prop:strict-decrease-non-stationary}, let
\(k+2\) be the first later step at which the execution either returns
\textsc{fail} or makes a successful clear step with threshold index larger
than \(1\). Such a step exists because the full execution fails. By
construction, the execution on the corresponding long block is
non-threshold-stationary and satisfies \(t_b\le1\) for every \(b<k+2\).

If this block already returns \textsc{fail}, Proposition~\ref{prop:strict-decrease-non-stationary}
passes the failure to the single-ratio truncation, producing a shorter
counterexample. Otherwise, the block outputs a proper suffix state. The same
proposition then implies either failure of the truncated block, or dominance by
the truncated block's output. In the dominance case, continuing both states
with the same remaining ratio tail and applying Proposition~\ref{prop:mono}
again yields failure of the shorter truncated input. This contradicts the
minimality of the original failing global ratio vector.

All cases lead to a shorter failing input, which is impossible. Therefore no
finite \(\alpha\)-weakly regular ratio sequence can make
Algorithm~\ref{alg:state-recursion} return \textsc{fail}. Hence, for every
\(\alpha\)-weakly regular \(F_S\) and \(\beta=\tfrac12\),
Algorithm~\ref{alg:state-recursion} never returns \textsc{fail}.
\end{proof}

\medskip
Theorem~\ref{thm:never-fail} proves the main conclusion needed for the
lower-bound argument: when \(\beta=\tfrac12\), Algorithm~\ref{alg:state-recursion} never returns \textsc{fail} on
every finite input. Next, we not only rule out failure, but also identify the
boundary cases in which the canonical output is undominated.  This gives the
following dominance characterization.

\begin{theorem}[Dominance]
\label{thm:dominance}
Fix $\beta=\tfrac{1}{2}$ and a zero-based global ratio vector $\mathbf r=(r_0,r_1,\ldots,r_k)$ with $r_0=0$. Let
\(\omega^\star := \mathrm{ALG}_2(\omega_0; r_1,\ldots,r_k)\) if \(r_1>0\), and \(\omega^\star := \mathrm{ALG}_2(\omega_0; r_2,\ldots,r_k)\) if \(r_1=0\). Then
\[
\omega^\star\text{ is undominated}\quad\Leftrightarrow\quad
\omega^\star\in \mathcal S_{(r_1,\ldots,r_k)}\;\land\; \forall \mathbf{r}'\prec_{\mathrm{suf}}(r_1,\ldots,r_k),\; \mathrm{ALG}_2(\omega_0;\mathbf r')\in\mathcal S_{(0,\mathbf r')}
\]
where \(\mathbf r'\prec_{\mathrm{suf}}(r_1,\ldots,r_k)\) means that \(\mathbf r'\) is a proper suffix of the positive ratio tail.
\end{theorem}

\begin{figure}[htbp]
    \centering
    \captionsetup[subfigure]{
        justification=centering,
        singlelinecheck=false,
        skip=2pt
    }
    \captionsetup[figure]{skip=6pt}

    \begin{subfigure}[t]{0.29\textwidth}
        \centering
        \vspace{0pt}
        \begin{tikzpicture}[scale=0.42]
            \def\n{11}

            \begin{scope}
                \clip (0,0) -- (0,\n) -- (\n,\n) -- cycle;

                \fill[gray, fill opacity=0.28] (0,1) rectangle (1,\n);
                \fill[gray, fill opacity=0.28] (1,5) rectangle (2,\n);
                \fill[gray, fill opacity=0.28] (2,7) rectangle (4,\n);
                \fill[gray, fill opacity=0.28] (4,\n-1) rectangle (8,\n);

                \draw[step=1cm, draw=black!75] (0,0) grid (\n,\n);
            \end{scope}

            \draw[thick] (0,0) -- (0,\n) -- (\n,\n) -- cycle;

            \draw[dashed] (4.5,0) -- (4.5,5.5);
            \draw[dashed] (7.5,0) -- (7.5,\n-1);
            \draw[dashed] (9.5,0) -- (9.5,\n-0.5);

            \draw[dashed,red,thick] (1,2) rectangle (2,5);
            \draw[dashed,blue,thick] (0,1) rectangle (1,6);
            \draw[purple,thick] (1,5) rectangle (2,6);

            \node at (1.5,-0.6) {$\gamma_1$};
            \node at (3,-0.6) {$\dots$};
            \node at (4.5,-0.6) {$\gamma_i$};
            \node at (5.5,-0.6) {$\dots$};
            \node at (6.5,-0.6) {$\gamma_t$};

            \node at (7.5,-0.6) {$r_1$};
            \node at (8.5,-0.6) {$\dots$};
            \node at (9.5,-0.6) {$r_k$};

            \node at (7.5,\n-0.5) {$\tilde y$};
            \node at (0.5,1.5) {$\scriptstyle \frac{\beta}{1+\alpha}$};

            \node at (-0.7,\n+0.5) {$b$};
            \node at (-0.7,1.5) {$1$};
            \node at (-1,5.5) {$i+1$};
            \node at (-1.3,\n-0.5) {$t{+}k{+}1$};
        \end{tikzpicture}
        \caption{The achievable state\\
        $\tilde\omega=\text{ALG}_2(\omega_0,\gamma_1,\ldots,\gamma_t,\mathbf{r})$}
        \label{fig:canonical-tilde}
    \end{subfigure}
    \hfill
    \begin{minipage}[t]{0.04\textwidth}
        \centering
        \vspace{0pt}
        \raisebox{-2.6cm}[0pt][0pt]{%
            \tikz[baseline=-0.6ex]{
                \draw[->, line width=0.9pt] (0.5,0) -- (1.0,0);
            }
        }
    \end{minipage}
    \hfill
    \begin{subfigure}[t]{0.29\textwidth}
        \centering
        \vspace{0pt}
        \begin{minipage}[t][5.6cm][b]{\linewidth}
        \begin{tikzpicture}[scale=0.42]
            \def\n{10}

            \begin{scope}
                \clip (0,0) -- (0,\n) -- (\n,\n) -- cycle;

                \fill[gray, fill opacity=0.28] (0,1) rectangle (1,\n);
                \fill[gray, fill opacity=0.28] (1,6) rectangle (3,\n);
                \fill[gray, fill opacity=0.28] (3,\n-1) rectangle (7,\n);

                \draw[step=1cm, draw=black!75] (0,0) grid (\n,\n);
            \end{scope}

            \draw[thick] (0,0) -- (0,\n) -- (\n,\n) -- cycle;

            \draw[dashed] (6.5,0) -- (6.5,\n-1);
            \draw[dashed] (8.5,0) -- (8.5,\n-0.5);

            \node at (1.5,-0.6) {$\gamma_2'$};
            \node at (2.5,-0.6) {$\dots$};
            \node at (3.5,-0.6) {$\gamma_i'$};
            \node at (4.5,-0.6) {$\dots$};
            \node at (5.5,-0.6) {$\gamma_t$};

            \node at (6.5,-0.6) {$r_1$};
            \node at (7.5,-0.6) {$\dots$};
            \node at (8.5,-0.6) {$r_k$};

            \node at (6.5,\n-0.5) {$y'$};
            \node at (0.5,1.5) {$\scriptstyle \frac{\beta}{1+\alpha}$};

            \node at (-0.7,\n+0.5) {$b$};
            \node at (-0.7,1.5) {$1$};
            \node at (-1,\n-0.5) {$t+k$};
        \end{tikzpicture}
        \end{minipage}
        \caption{Deleting the first prefix term gives
        $\omega'=\text{ALG}_2(\omega_0,\gamma_2',\ldots,\gamma_i',\gamma_{i+1},\ldots,\gamma_t,\mathbf{r})$}
        \label{fig:canonical-middle}
    \end{subfigure}
    \hfill
    \begin{minipage}[t]{0.04\textwidth}
        \centering
        \vspace{0pt}
        \raisebox{-2.6cm}[0pt][0pt]{%
            \tikz[baseline=-0.6ex]{
                \draw[->, line width=0.9pt] (0.5,0) -- (1.0,0);
            }
        }
    \end{minipage}
    \hfill
    \begin{subfigure}[t]{0.22\textwidth}
        \centering
        \vspace{0pt}
        \begin{minipage}[t][5.6cm][b]{\linewidth}
            \centering
            \begin{tikzpicture}[scale=0.5]
                \def\n{6}
    
                \fill[gray, fill opacity=0.28] (0,1) rectangle (1,\n);
                \fill[gray, fill opacity=0.28] (1,\n-1) rectangle (2,\n);
    
                \begin{scope}
                    \clip (0,0) -- (0,\n) -- (\n,\n) -- cycle;
                    \draw[step=1cm, draw=black!75] (0,0) grid (\n,\n);
                \end{scope}
    
                \draw[thick] (0,0) -- (0,\n) -- (\n,\n) -- cycle;
    
                \draw[dashed] (1.5,0) -- (1.5,\n-1);
                \draw[dashed] (4.5,0) -- (4.5,\n-0.5);
    
                \node at (1.5,-0.55) {$r_1$};
                \node at (3,-0.55) {$\dots$};
                \node at (4.5,-0.55) {$r_k$};
    
                \node at (0.5,1.5) {$\frac{\beta}{1+\alpha}$};
                \node at (1.5,\n-0.5) {$y^\star$};

                \node at (-0.7,\n+0.5) {$b$};
                \node at (-0.7,1.5) {$1$};
                \node at (-1,\n-0.5) {$k+1$};
            \end{tikzpicture}
        \end{minipage}
        \caption{The threshold-stationary state
        $\omega^\star=\mathrm{ALG}_2(\omega_0,\mathbf{r})$}
        \label{fig:canonical-star}
    \end{subfigure}
    
    \caption{
    Schematic illustration of the proof of Theorem~\ref{thm:dominance}. Each unit square represents an allocation entry \(x(s,b)\); gray squares indicate \(x(s,b)>0\); the horizontal axis records the seller value \(s\), and the vertical axis records the buyer value \(b\). Panel~(a) shows an achievable state
    \(\tilde\omega=\text{ALG}_2(\omega_0,\gamma_1,\ldots,\gamma_t,\mathbf{r})\).
    The blue and red boxes highlight the local modification in one reduction step: the blue box marks the initial block that is removed, while the red box marks the region to be reshaped, corresponding to the update from \(\gamma_2,\ldots,\gamma_i\) to \(\gamma_2',\ldots,\gamma_i'\) up to the first turning point, so that the resulting state satisfies \(\omega' \prec \tilde\omega\). Panel~(b) shows the resulting state
    \(\omega'=\text{ALG}_2(\omega_0,\gamma_2',\ldots,\gamma_i',\gamma_{i+1},\ldots,\gamma_t,\mathbf{r})\).
    Repeating this reduction removes the entire prefix and yields the threshold-stationary state
    \(\omega^\star=\mathrm{ALG}_2(\omega_0,\mathbf{r})\) in panel~(c).
    }
    \label{fig:canonical-schematic}
\end{figure}
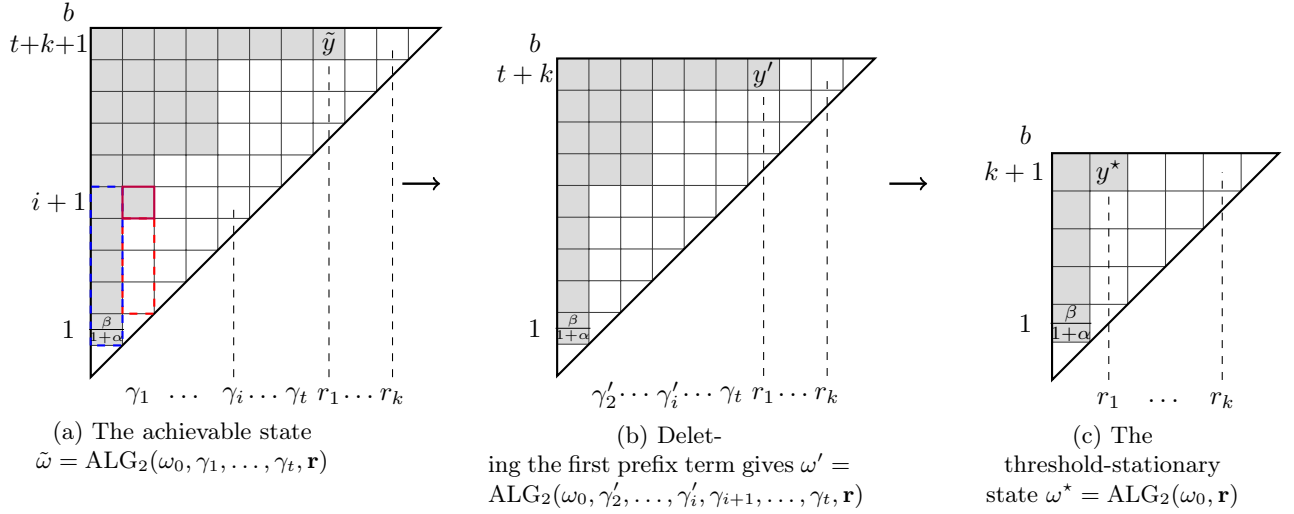

\begin{proof}[Proof of Theorem~\ref{thm:dominance}]
Let \(\mathbf{r}=(r_1,\ldots,r_k)\) denote the positive ratio tail. We first consider if \(\omega^\star=\mathrm{ALG}_2(\omega_0;\mathbf{r})=(y^\star;\mathbf{r})\in \mathcal S_{\mathbf r}\).
Let
\(\widetilde\omega=\text{ALG}_2(\omega_0;\gamma_1,\ldots,\gamma_t,\mathbf{r})=(\widetilde y\,; \mathbf{r})\)
be any achievable state in \(\mathcal S_{\mathbf r}\) with \(t\ge 1\) shown in
Figure~\ref{fig:canonical-schematic}(a), where $\gamma_1,\ldots,\gamma_t$ is the prefix ratio sequence.

Let \(i\) denote the first turning-point index in the prefix part, which means that the threshold of the allocation $x$ moves forward from $0$ when $b=i+1$.
We apply the truncation result from Section~\ref{sec:t1} to the initial block
\((\gamma_1,\ldots,\gamma_i)\), while keeping the remaining ratio sequence
\((\gamma_{i+1},\ldots,\gamma_t,\mathbf{r})\)
fixed. This argument step deletes \(\gamma_1\) and modifies
\(\gamma_2,\ldots,\gamma_i\) into
\(\gamma_2',\ldots,\gamma_i'\), obtaining a new achievable state
\(\omega'=\text{ALG}_2(\omega_0;\gamma_2',\ldots,\gamma_i',\gamma_{i+1},\ldots,\gamma_t,\mathbf{r})=(y';\mathbf{r})\)
such that
\(\widetilde\omega \prec \omega'\), corresponding to the transition from
Figure~\ref{fig:canonical-schematic}(a) to
Figure~\ref{fig:canonical-schematic}(b).
Iterating this operation finitely many times removes the whole prefix and
reduces to the threshold-stationary case state \(\omega^\star\) shown in Figure~\ref{fig:canonical-schematic}(c).

Next, we consider if \(\omega^\star\notin \mathcal S_{\mathbf r}\). If \(\omega^\star\in \mathcal S_{(0,\mathbf r)}\), then \(\mathcal S_{\mathbf r}=\emptyset\). Otherwise, \(\omega^\star\in\mathcal S_{\mathbf r'}\) where \(\mathbf{r}'\) is the proper suffix of \(\mathbf{r}\). Since any states in \(\mathcal S_{\mathbf r'}\) dominate states in \(\mathcal S_{\mathbf r}\), 
this proves the claim.
\end{proof}

\medskip
Combining Theorem~\ref{thm:never-fail} with the ironing and discretization reductions established earlier, we obtain the main theorem for arbitrary independent buyer and seller distributions.

\begin{theorem}[Half-approximation guarantee]
\label{cor:main}
For every pair of independent buyer and seller distributions,
\(\mathrm{SB} \;\ge\; \frac{1}{2}\,\mathrm{FB}.\)
\end{theorem}

\section{Upper Bound: Tightness via an Example}
\label{sec:ub}

Section~5 showed that the bottleneck in the lower-bound analysis occurs only in the limiting regime \(\alpha\to\infty\) when $\phi_S$ is constant. For ease of verification, this section presents an explicit continuous analogue
of the same regime and includes several additional discussions showing that the
approximation ratio indeed converges to \(\tfrac12\). This establishes the tightness of the half-approximation bound.

Let \(\alpha>0\) and \(c>0\) be parameters, and define
\begin{align*}
    F_S(s) =
    \left(
    \frac{\alpha c}{\alpha c+(1+\alpha)(1-s)}
    \right)^{\frac{\alpha}{1+\alpha}},
    \quad
    F_B(b) =
    1-
    \left(
    \frac{\alpha c}{\alpha c+(1+\alpha)b}
    \right)^{\frac{\alpha}{1+\alpha}},
\end{align*}
for $s\in(0,1]$ and $b\in[0,1)$. This is a symmetric family of distributions on \([0,1]\), with point masses at \(s=0\) and \(b=1\), both equal to
\[
P=
\left(
\frac{\alpha c}{\alpha c+1+\alpha}
\right)^{\frac{\alpha}{1+\alpha}}.
\]

The family above is a particular instance of a general class of
mixed distributions we use throughout this section: those whose
seller distribution has a single atom at $s=0$ and whose buyer
distribution has a single atom at $b=1$, with absolutely
continuous mass elsewhere. Formally, we can represent
\[
\mathrm{d}F_S(s)=p_S\,\delta_0(\mathrm{d}s)+f_S(s)\,\mathrm{d}s,
\qquad
\mathrm{d}F_B(b)=f_B(b)\,\mathrm{d}b+p_B\,\delta_1(\mathrm{d}b),
\]
where $\delta_{x_0}$ denotes the Dirac measure concentrated at $x_0$. For interior points define
\[
\phi_S(s):=s+\frac{F_S(s)}{f_S(s)},
\qquad
\phi_B(b):=b-\frac{1-F_B(b)}{f_B(b)},
\qquad s,b\in(0,1),
\]
with no $\alpha$ and use the endpoint conventions
\[
\widehat\phi_S(0):=0,\qquad
\widehat\phi_S(s):=\phi_S(s)\ (s\in(0,1)),
\]
\[
\widehat\phi_B(b):=\phi_B(b)\ (b\in(0,1)),
\qquad
\widehat\phi_B(1):=1.
\]

\begin{proposition}[Endpoint-mass absolutely continuous case]
\label{prop:endpoint_mass_phi}
Under the endpoint-mass notation above, the mixed-distribution problem in
Proposition~\ref{prop:general_mixed} is equivalent to
\begin{align*}
\sup_x\quad
& \mathrm{GFT}(x)=
\int_{[0,1]^2}(b-s)x(s,b)\,\mathrm{d}F_S(s)\,\mathrm{d}F_B(b)\\
\mathrm{s.t.}\quad
& \Lambda_{\text{gen}}(x)=
\int_{[0,1]^2}
\bigl(\widehat\phi_B(b)-\widehat\phi_S(s)\bigr)
x(s,b)\,\mathrm{d}F_S(s)\,\mathrm{d}F_B(b)\ge 0,\\
& x_B(b_1)\le x_B(b_2), \qquad \forall\, b_1<b_2,\\
& x_S(s_1)\ge x_S(s_2), \qquad \forall\, s_1<s_2,
\end{align*}
Equivalently, decomposing the product measure into the endpoint atom, the two
one-dimensional boundaries, and the interior gives
\begin{align*}
\Lambda_{\text{gen}}(x)
&=
p_Sp_B\,x(0,1)
+p_S\int_0^1 \phi_B(b)x(0,b)f_B(b)\,\mathrm{d}b
+p_B\int_0^1 \bigl(1-\phi_S(s)\bigr)x(s,1)f_S(s)\,\mathrm{d}s \\
&\qquad
+\int_0^1\int_0^1
\bigl(\phi_B(b)-\phi_S(s)\bigr)x(s,b)f_S(s)f_B(b)\,\mathrm{d}b\,\mathrm{d}s.
\end{align*}
\end{proposition}

We compute \(\mathrm{FB}\) and \(\mathrm{SB}\) in closed form. The most critical technical step is showing that the SB mechanism is boundary-restricted for this family. By comparing the marginal gain in \(\mathrm{GFT}\) to the marginal loss in the feasibility constraint $\Lambda$, we demonstrate that boundary trades, where $s=0$ or $b=1$, strictly dominate any interior trades $(s,b)$ in terms of their contribution to the objective per unit of feasibility slack consumed. Finally, we compute the ratio $\mathrm{SB}/\mathrm{FB}$ and show that as $c \to 0^+$ and \(\alpha\to\infty\), the ratio \(\mathrm{SB}/\mathrm{FB}\) tends to \(\tfrac12\). Therefore, no mechanism can guarantee a ratio strictly larger than \(\tfrac12\).

\begin{theorem}\label{thm:ub}
No mechanism that is BIC, IR, and WBB can guarantee a worst-case approximation ratio to the first-best GFT exceeding $1/2$.
\end{theorem}

\paragraph{Why we do not impose constant virtual values on the original scale.}
A natural alternative would be to require
\[
b-\frac{1-F_B(b)}{f_B(b)}\equiv -c,
\qquad
s+\frac{F_S(s)}{f_S(s)}\equiv 1+c
\]
directly on \([0,1]\). This yields
\[
F_B(b)=\frac{b}{c+b},
\qquad
F_S(s)=\frac{c}{c+1-s},
\]
with atoms \(c/(1+c)\) at \(b=1\) and \(s=0\). For this family,
\[
\mathrm{FB}(c)
=\int_0^1 \frac{c^2}{(c+b)(c+1-b)}\,\mathrm{d}b
=\frac{2c^2}{1+2c}\log\!\left(1+\frac1c\right).
\]
Moreover, one can verify that the optimal second-best allocation is again supported only on the two boundary lines; writing \(u=b\) on \(\{s=0\}\) and \(u=1-s\) on \(\{b=1\}\), the WBB constraint truncates both lines at the common threshold \(u\ge t\), where \(t=\frac{c}{1+2c}\), and therefore
\[
\mathrm{SB}(c)
=\left(\frac{c}{1+c}\right)^2
+2\frac{c}{1+c}\int_t^1 u\,\frac{c}{(c+u)^2}\,\mathrm{d}u
=\frac{2c^2}{1+c}\log\!\left(1+\frac1{2c}\right).
\]
Hence both \(\mathrm{FB}\) and \(\mathrm{SB}\) are governed by the same logarithmic term:
\[
\mathrm{FB}(c)=2c^2\log(1/c)+O(c^2),
\qquad
\mathrm{SB}(c)=2c^2\log(1/c)+O(c^2).
\]
Equivalently, the WBB constraint only changes the lower cutoff of the same logarithmic integral, from \(u=0\) to \(u=t=\Theta(c)\), and therefore does not alter its leading coefficient. In particular,
\[
\frac{\mathrm{SB}(c)}{\mathrm{FB}(c)}\to 1
\qquad\text{as } c\to0^+.
\]
Thus the direct equal-virtual-value construction does not realize the bottleneck identified in Section~5. The family above is introduced precisely because it approaches the critical regime only through the limit
\[
\frac{\alpha}{1+\alpha}\to 1,
\qquad\text{equivalently,}\qquad
\alpha\to\infty.
\]

\paragraph{Order of limits.}
The distinction from the exactly constant-virtual-value construction is therefore
not the limiting shape of the distributions alone, but the order in which that
shape is approached. In the family used in this section, one first fixes
\(\alpha<\infty\) and lets \(c\to0^+\), and only then sends
\(\alpha\to\infty\). This is the regime that exposes the bottleneck:
\[
\lim_{\alpha\to\infty}\lim_{c\to0^+}
\frac{\mathrm{SB}_{\alpha,c}}{\mathrm{FB}_{\alpha,c}}
=\frac12.
\]
By contrast, imposing constant virtual values directly on the original scale is
equivalent to taking the limit \(\alpha\to\infty\) first, since for fixed \(c\),
\[
\left(
\frac{\alpha c}{\alpha c+(1+\alpha)u}
\right)^{\frac{\alpha}{1+\alpha}}
\longrightarrow
\frac{c}{c+u},
\qquad u\in[0,1].
\]
After this limit has been taken, the subsequent small-\(c\) limit gives the
constant-virtual-value family discussed above, for which
\[
\lim_{c\to0^+}\lim_{\alpha\to\infty}
\frac{\mathrm{SB}_{\alpha,c}}{\mathrm{FB}_{\alpha,c}}
=1.
\]
Thus the two constructions correspond to non-commuting iterated limits:
\[
\lim_{\alpha\to\infty}\lim_{c\to0^+}
\frac{\mathrm{SB}_{\alpha,c}}{\mathrm{FB}_{\alpha,c}}
=\frac12
\neq
1
=
\lim_{c\to0^+}\lim_{\alpha\to\infty}
\frac{\mathrm{SB}_{\alpha,c}}{\mathrm{FB}_{\alpha,c}}.
\]
The half-approximation bottleneck is visible only when the small-\(c\) limit is
taken before the exact constant-virtual-value limit.
\newpage{}
\appendix
\section{Omitted Proofs in Section~\ref{sec:dis}}
\label{app:general_mixed}
\begin{proof}[Proof of Proposition~\ref{prop:general_mixed}]
Let \(p_B(s,b)\) denote the payment made by the buyer and let
\(p_S(s,b)\) denote the payment received by the seller. Define the interim
expected payments by
\[
P_B(b)=\int_{[0,1]}p_B(s,b)\,\mathrm dF_S(s),
\qquad
P_S(s)=\int_{[0,1]}p_S(s,b)\,\mathrm dF_B(b).
\]
The interim utilities are
\[
u_B(b)=b\,x_B(b)-P_B(b),
\qquad
u_S(s)=P_S(s)-s\,x_S(s).
\]

\medskip
\noindent
\((1)\) Necessity. Suppose that \(x\) is implementable by a BIC, IR, and
WBB mechanism. By the standard one-dimensional envelope
characterization, BIC implies that \(x_B\) is weakly increasing and \(x_S\)
is weakly decreasing. Moreover,
\[
u_B(b)=u_B(0)+\int_0^b x_B(t)\,\mathrm dt,
\qquad
u_S(s)=u_S(1)+\int_s^1 x_S(t)\,\mathrm dt .
\]
Equivalently,
\[
P_B(b)
=
b\,x_B(b)-\int_0^b x_B(t)\,\mathrm dt-u_B(0),
\]
and
\[
P_S(s)
=
s\,x_S(s)+\int_s^1 x_S(t)\,\mathrm dt+u_S(1).
\]
Therefore,
\begin{align*}
\Lambda_{\mathrm{gen}}(x)
&=
\int_{[0,1]}
\left(
b\,x_B(b)-\int_0^b x_B(t)\,\mathrm dt
\right)\mathrm dF_B(b) -
\int_{[0,1]}
\left(
s\,x_S(s)+\int_s^1 x_S(t)\,\mathrm dt
\right)\mathrm dF_S(s) \\
&=
\int_{[0,1]}P_B(b)\,\mathrm dF_B(b)
-
\int_{[0,1]}P_S(s)\,\mathrm dF_S(s)
+
u_B(0)+u_S(1).
\end{align*}
Weak budget balance gives
\[
\int_{[0,1]}P_B(b)\,\mathrm dF_B(b)
-
\int_{[0,1]}P_S(s)\,\mathrm dF_S(s)
\ge 0,
\]
and IR gives \(u_B(0)\ge 0\) and \(u_S(1)\ge 0\). Hence
\[
\Lambda_{\mathrm{gen}}(x)\ge 0.
\]

\medskip
\noindent
\((2)\) Sufficiency. Conversely, suppose that \(x_B\) is weakly increasing,
\(x_S\) is weakly decreasing, and
\(\Lambda_{\mathrm{gen}}(x)\ge 0\).
Define the report-dependent payments by
\begin{align*}
p_B(s,b)=p_S(s,b)
:={}
b\,x_B(b)-\int_0^b x_B(t)\,\mathrm dt -
\left(
\int_0^s x_S(t)\,\mathrm dt-s\,x_S(s)
\right) +
\int_{[0,1]}
\left(
\int_0^r x_S(t)\,\mathrm dt-r\,x_S(r)
\right)\mathrm dF_S(r).
\end{align*}
Then \(p_B(s,b)=p_S(s,b)\) for every \((s,b)\), so the mechanism is strongly
 budget-balanced.

The buyer's interim payment is
\[
P_B(b)
=
\int_{[0,1]}p_B(s,b)\,\mathrm dF_S(s)
=
b\,x_B(b)-\int_0^b x_B(t)\,\mathrm dt.
\]
The seller's interim payment is
\begin{align*}
P_S(s)
=
\int_{[0,1]}p_S(s,b)\,\mathrm dF_B(b) =
s\,x_S(s)+\int_s^1 x_S(t)\,\mathrm dt
+
\Lambda_{\mathrm{gen}}(x).
\end{align*}
Hence the induced interim utilities are
\[
u_B(b)=\int_0^b x_B(t)\,\mathrm dt,
\qquad
u_S(s)=\int_s^1 x_S(t)\,\mathrm dt+\Lambda_{\mathrm{gen}}(x).
\]
Since \(x_B,x_S\ge 0\) and \(\Lambda_{\mathrm{gen}}(x)\ge 0\), interim IR holds.

Finally, BIC follows from the envelope characterization together with the
monotonicity of \(x_B\) and \(x_S\): \(x_B\) is weakly increasing for the
buyer, and \(x_S\) is weakly decreasing for the seller. Thus the constructed
mechanism is BIC, IR, and SBB.

Combining \((1)\) and \((2)\), the feasible allocation rules are exactly those
satisfying the two monotonicity constraints and
\(\Lambda_{\mathrm{gen}}(x)\ge 0\). Therefore the original second-best problem
is equivalent to the stated formulation.
\end{proof}

\label{app:dis}
\begin{proof}[Proof of Proposition~\ref{prop:discrete}]
Since \(F_S\) and \(F_B\) are discrete, the objective in
Proposition~\ref{prop:general_mixed} becomes
\[
\sum_{i,j=0}^N x(s_i,b_j)f_S(s_i)f_B(b_j)(b_j-s_i),
\]
and the interim allocation rules are
\[
x_B(b_j)=\sum_{i=0}^N x(s_i,b_j)f_S(s_i),
\qquad
x_S(s_i)=\sum_{j=0}^N x(s_i,b_j)f_B(b_j).
\]
The monotonicity constraints reduce to the adjacent inequalities because the
supports are ordered and finite.

It remains to rewrite \(\Lambda_{\mathrm{gen}}\). Since the distributions are
purely discrete, the values of \(x_B\) and \(x_S\) away from the support points
are not determined by the mechanism. Let \(g\) and \(h\) be arbitrary monotone
extensions satisfying
\[
g(b_j)=x_B(b_j),\qquad h(s_i)=x_S(s_i),
\]
where \(g\) is weakly increasing and \(h\) is weakly decreasing. Then
\[
\Lambda_{g,h}(x)
=
\sum_{j=0}^N f_B(b_j)
\left(
b_jx_B(b_j)-\int_{b_0}^{b_j}g(t)\,\mathrm{d}t
\right)
-
\sum_{i=0}^N f_S(s_i)
\left(
s_ix_S(s_i)+\int_{s_i}^{s_N}h(t)\,\mathrm{d}t
\right).
\]
For \(t\in[b_\ell,b_{\ell+1}]\), monotonicity gives
\(g(t)\ge x_B(b_\ell)\). Hence
\[
\int_{b_0}^{b_j}g(t)\,\mathrm{d}t
\ge
\sum_{\ell<j}(b_{\ell+1}-b_\ell)x_B(b_\ell).
\]
Similarly, for \(t\in[s_{\ell-1},s_\ell]\),
\(h(t)\ge x_S(s_\ell)\), and therefore
\[
\int_{s_i}^{s_N}h(t)\,\mathrm{d}t
\ge
\sum_{\ell>i}(s_\ell-s_{\ell-1})x_S(s_\ell).
\]
Both lower bounds are attained by the canonical step extensions
\[
g^*(t)=x_B(b_\ell)\quad \text{for }t\in[b_\ell,b_{\ell+1}),
\qquad
h^*(t)=x_S(s_\ell)\quad \text{for }t\in(s_{\ell-1},s_\ell].
\]
Thus the sharp discrete budget term is
\[
\begin{aligned}
\Lambda_D(x)
=\;&
\sum_{j=0}^N f_B(b_j)
\left(
b_jx_B(b_j)
-
\sum_{\ell<j}(b_{\ell+1}-b_\ell)x_B(b_\ell)
\right)\\
&-
\sum_{i=0}^N f_S(s_i)
\left(
s_ix_S(s_i)
+
\sum_{\ell>i}(s_\ell-s_{\ell-1})x_S(s_\ell)
\right).
\end{aligned}
\]

We now put this expression into pointwise form. Exchanging the order of
summation in the buyer part gives
\[
\sum_{j=0}^N
x_B(b_j)f_B(b_j)
\left(
b_j+(b_j-b_{j+1})\frac{1-F_B(b_j)}{f_B(b_j)}
\right),
\]
where \(b_{N+1}:=b_N\), so the boundary term vanishes because
\(1-F_B(b_N)=0\). Similarly, the seller part equals
\[
\sum_{i=0}^N
x_S(s_i)f_S(s_i)
\left(
s_i+(s_i-s_{i-1})\frac{F_S(s_{i-1})}{f_S(s_i)}
\right),
\]
where \(s_{-1}:=s_0\), so the boundary term at \(i=0\) vanishes.

Substituting the interim rules into \(\Lambda_D(x)\), we obtain
\[
\Lambda_D(x)
=
\sum_{i,j=0}^N
x(s_i,b_j)f_S(s_i)f_B(b_j)
\Bigg(
b_j+(b_j-b_{j+1})\frac{1-F_B(b_j)}{f_B(b_j)}-s_i-(s_i-s_{i-1})\frac{F_S(s_{i-1})}{f_S(s_i)}
\Bigg).
\]
Therefore \(\Lambda_D(x)\ge 0\) is exactly the displayed discrete
budget-balance constraint. Combining this with the discrete objective and the
adjacent monotonicity constraints proves the proposition.
\end{proof}

\begin{proof}[Proof of Lemma~\ref{lem:feasibility-preservation}]
By assumption, the discrete allocation rule $x^{(k+1)}$ is weakly increasing in the buyer's value and weakly decreasing in the seller's value. Because $\widetilde{x}(s,b)$ is constructed as a flat step-function that inherits these exact values, $\widetilde{x}(s,b)$ is globally monotonic over $(s,b) \in [0,1]^2$. 

Then, it remains to show that 
\[\Lambda(\widetilde{x})=\int_0^1\int_0^1\left(b-\frac{1-F_B(b)}{f_B(b)}-s-\frac{F_S(s)}{f_S(s)}\right)\widetilde{x}(s,b)f_B(b)f_S(s)\mathrm{d}s\mathrm{d}b \ge 0.\]
The objective $\Lambda(\widetilde{x})$ can be decomposed as a sum of integrals over the grid cells:
\begin{equation}\label{eq:Lambda}
    \Lambda(\widetilde{x})=\sum_{j=1}^N \sum_{i=1}^N x^{(k+1)}(s_i, b_{j-1}) \left[ \int_{s_{i-1}}^{s_i} \int_{b_{j-1}}^{b_j} \big(\psi_B(b) - \psi_S(s)\big) f_B(b) f_S(s) \mathrm{d}b \mathrm{d}s \right],
\end{equation}
where the virtual values are defined as $\psi_B(b) = b - \frac{1-F_B(b)}{f_B(b)}$ and $\psi_S(s) = s + \frac{F_S(s)}{f_S(s)}$.

Because the integrals are separable, we can evaluate the buyer's continuous virtual value within a single cell:
\begin{align*}
    \int_{b_{j-1}}^{b_j} \psi_B(b) f_B(b) \mathrm{d}b &= \int_{b_{j-1}}^{b_j} \big[ b f_B(b) - (1-F_B(b)) \big] \mathrm{d}b= \Big[ b(F_B(b)-1) \Big]_{b_{j-1}}^{b_j}\\
    &= b_j(F_B(b_j)-1)-b_{j-1}(F_B(b_{j-1})-1).
\end{align*}
Let $\hat{\psi}_B(b_{j-1}) = b_{j-1} - (b_j - b_{j-1}) \frac{1 - \hat{F}_B(b_j)}{\hat{f}_B(b_{j-1})}$. Multiplying this discrete virtual value by its probability mass $\hat{f}_B(b_{j-1}) = \hat{F}_B(b_j) - \hat{F}_B(b_{j-1})$ gives:
\begin{align*}
    \hat{\psi}_B(b_{j-1}) \hat{f}_B(b_{j-1}) &=b_{j-1} \big(\hat{F}_B(b_j) - \hat{F}_B(b_{j-1})\big) - (b_j - b_{j-1}) \big(1 - \hat{F}_B(b_j)\big)\\
    &= b_j(\hat{F}_B(b_j)-1)-b_{j-1}(\hat{F}_B(b_{j-1})-1).
\end{align*}
This yields the exact identity:$$\hat{\psi}_B(b_{j-1}) \hat{f}_B(b_{j-1}) = \int_{b_{j-1}}^{b_j} \psi_B(b) f_B(b) \mathrm{d}b.$$

By defining $\hat{\psi}_S(s_i) = s_i + (s_i - s_{i-1}) \frac{\hat{F}_S(s_{i-1})}{\hat{f}_S(s_i)}$ and applying the symmetric integration for the seller's cell $(s_{i-1}, s_i]$, we obtain:
$$\hat{\psi}_S(s_i) \hat{f}_S(s_i) = \int_{s_{i-1}}^{s_i} \psi_S(s) f_S(s) \mathrm{d}s.$$

Substituting the two identities above back into Eq.~\eqref{eq:Lambda}, we have:
\[ \Lambda(\widetilde{x})= \sum_{j=1}^N \sum_{i=1}^N x^{(k+1)}(s_i, b_{j-1}) \Big[ \hat{\psi}_B(b_{j-1}) \hat{f}_B(b_{j-1}) \hat{f}_S(s_i) - \hat{\psi}_S(s_i) \hat{f}_S(s_i) \hat{f}_B(b_{j-1}) \Big].\]
This shows that $\Lambda(\widetilde{x})$ is exactly the discrete discriminant under $x^{(k+1)}$ for discrete distributions $\hat{F}_S$ and $\hat{F}_B$. Because the discrete allocation rule $x^{(k+1)}$ is BIC, IR and WBB by assumption, the allocation rule $\widetilde{x}(s,b)$ strictly satisfies $\Lambda(\widetilde{x}) \ge 0$. This completes the proof.
\end{proof}

\begin{proof}[Proof of Lemma~\ref{lem:convergence}]
We first compare \(\widehat{\mathrm{SB}}_{M+1}\) with
\(\widehat{\mathrm{SB}}_M\). Consider a rectangle of the \(2^{-M}\)-partition,
whose seller and buyer representatives are \(s_k^{(M)}=k2^{-M}\) and
\(b_{\ell-1}^{(M)}=(\ell-1)2^{-M}\). For every finer-grid pair
\((s_i^{(M+1)},b_{j-1}^{(M+1)})\) contained in this rectangle, we have
\(0\le s_k^{(M)}-s_i^{(M+1)}\le 2^{-(M+1)}\) and
\(0\le b_{j-1}^{(M+1)}-b_{\ell-1}^{(M)}\le 2^{-(M+1)}\). Hence
\begin{equation}\label{eq:dyadic-local-surplus-comparison}
0
\le
\bigl(b_{j-1}^{(M+1)}-s_i^{(M+1)}\bigr)
-
\bigl(b_{\ell-1}^{(M)}-s_k^{(M)}\bigr)
\le
2^{-M}.
\end{equation}

Let \(x^M\) be optimal for \(\widehat{\mathrm{SB}}_M\). Lift it to the
\(2^{-(M+1)}\)-grid by assigning the same value to all finer rectangles
contained in each rectangle of the \(2^{-M}\)-partition. By grid consistency,
the lifted allocation is feasible on the \(2^{-(M+1)}\)-grid, and by
\eqref{eq:dyadic-local-surplus-comparison}, its objective value is at least
\(\widehat{\mathrm{SB}}_M\). Therefore
\(\widehat{\mathrm{SB}}_{M+1}\ge \widehat{\mathrm{SB}}_M\).

Conversely, let \(x^{M+1}\) be optimal for
\(\widehat{\mathrm{SB}}_{M+1}\). Pool it to the \(2^{-M}\)-grid by
probability-mass averaging over the finer rectangles contained in each
rectangle of the \(2^{-M}\)-partition. The pooled allocation is feasible on
the \(2^{-M}\)-grid. By \eqref{eq:dyadic-local-surplus-comparison}, each
finer-grid surplus is at most the corresponding \(2^{-M}\)-grid surplus plus
\(2^{-M}\). Since \(0\le x^{M+1}\le1\) and the total probability mass is one,
we obtain
\begin{equation}\label{eq:dyadic-sb-comparison}
\widehat{\mathrm{SB}}_M
\le
\widehat{\mathrm{SB}}_{M+1}
\le
\widehat{\mathrm{SB}}_M+2^{-M}.
\end{equation}
Thus \(\{\widehat{\mathrm{SB}}_M\}_{M\ge1}\) is nondecreasing. Moreover,
\(\widehat{\mathrm{SB}}_M\ge0\) because the zero allocation is feasible, and
\(\widehat{\mathrm{SB}}_M\le1\) because \(0\le x\le1\) and values lie in
\([0,1]\). Hence \(\lim_{M\to\infty}\widehat{\mathrm{SB}}_M\) exists.

Next, we show that \(\widehat{\mathrm{FB}}_M\to \mathrm{FB}_{F_S,F_B}\). On
each rectangle \(I_i^{(M)}\times J_j^{(M)}\), one has
\(0\le (b-s)-(b_{j-1}^{(M)}-s_i^{(M)})\le 2^{1-M}\). Since the positive-part
map is increasing and \(1\)-Lipschitz, integrating over all rectangles gives
\begin{equation}\label{eq:dyadic-fb-convergence}
0
\le
\mathrm{FB}_{F_S,F_B}
-
\widehat{\mathrm{FB}}_M
\le
2^{1-M}.
\end{equation}
Therefore \(\widehat{\mathrm{FB}}_M\to \mathrm{FB}_{F_S,F_B}\).

We now compare the continuous GFT of the step-function extension with the
discrete value. On each rectangle \(I_i^{(M)}\times J_j^{(M)}\), we have
\(b-s\ge b_{j-1}^{(M)}-s_i^{(M)}\). Since
\(\widetilde x_M(s,b)=\widehat x_M(s_i^{(M)},b_{j-1}^{(M)})\) on this
rectangle and \(\widehat x_M\ge0\), integrating gives
\(\mathrm{GFT}_{F_S,F_B}(\widetilde x_M)\ge \widehat{\mathrm{SB}}_M\).
Moreover, by feasibility of \(\widetilde x_M\) for the continuous second-best
problem,
\begin{equation}\label{eq:dyadic-gft-sandwich}
\mathrm{SB}_{F_S,F_B}
\ge
\mathrm{GFT}_{F_S,F_B}(\widetilde x_M)
\ge
\widehat{\mathrm{SB}}_M .
\end{equation}

The sequence
\(\{\mathrm{GFT}_{F_S,F_B}(\widetilde x_M)\}_{M\ge1}\) is bounded by
\eqref{eq:dyadic-gft-sandwich}. Hence it admits a convergent subsequence, say
along \(M_m\to\infty\). Passing to the limit in
\eqref{eq:dyadic-gft-sandwich} along this subsequence gives
\begin{equation}\label{eq:dyadic-subsequence-limit}
\lim_{m\to\infty}
\mathrm{GFT}_{F_S,F_B}(\widetilde x_{M_m})
\ge
\lim_{M\to\infty}
\widehat{\mathrm{SB}}_M .
\end{equation}

Finally, since \(\mathrm{FB}_{F_S,F_B}>0\) and
\(\widehat{\mathrm{FB}}_M\to\mathrm{FB}_{F_S,F_B}\), we have
\(\widehat{\mathrm{FB}}_M>0\) for all sufficiently large \(M\). Combining
\eqref{eq:dyadic-fb-convergence} with the existence of
\(\lim_{M\to\infty}\widehat{\mathrm{SB}}_M\), we obtain
\begin{equation}\label{eq:dyadic-ratio-limit}
\lim_{M\to\infty}
\frac{\widehat{\mathrm{SB}}_M}{\widehat{\mathrm{FB}}_M}
=
\frac{
\lim_{M\to\infty}\widehat{\mathrm{SB}}_M
}{
\mathrm{FB}_{F_S,F_B}
}.
\end{equation}
Dividing \eqref{eq:dyadic-subsequence-limit} by
\(\mathrm{FB}_{F_S,F_B}\) and using \eqref{eq:dyadic-ratio-limit}, we conclude
that
\[
\lim_{m\to\infty}
\frac{
\mathrm{GFT}_{F_S,F_B}(\widetilde x_{M_m})
}{
\mathrm{FB}_{F_S,F_B}
}
\ge
\lim_{M\to\infty}
\frac{
\widehat{\mathrm{SB}}_M
}{
\widehat{\mathrm{FB}}_M
}.
\]
This proves the claim.
\end{proof}

\section{Omitted Proofs in Section~\ref{sec:new-problem}}
\label{app:new-problem}
\begin{proof}[Proof of Lemma~\ref{lem:Lagrangian}]
Considering the primal problem
\[
\sup_{x\in\mathcal R_1}\ \mathrm{GFT}(x)
\qquad\text{s.t.}\qquad
\Lambda_D(x)\ge 0,
\]
its optimal value is
\[
P^*:=\sup\left\{\mathrm{GFT}(x):x\in\mathcal R_1,\ \Lambda_D(x)\ge 0\right\}.
\]

To put the problem into the standard Lagrangian form, we multiply both the objective and the constraint by $-1$. Then the problem is equivalently written as
\[
-P^*
=
\inf_{x\in\mathcal R_1}\ \left(-\mathrm{GFT}(x)\right)
\qquad\text{s.t.}\qquad
-\Lambda_D(x)\le 0.
\]

Now we define the standard Lagrangian for this minimization problem:
\[
\widetilde L(x,\alpha)
:=
-\mathrm{GFT}(x)+\alpha\left(-\Lambda_D(x)\right),
\qquad \alpha\ge 0.
\]
That is,
\[
\widetilde L(x,\alpha)
=
-\mathrm{GFT}(x)-\alpha\,\Lambda_D(x).
\]
Since both the objective function and the constraints are linear with respect to the allocation rule $x$ (for any fixed $F_S$ and $F_B$), the problem constitutes a linear program in $x$. It follows from strong duality for linear programming that
\[
-P^*
=
\sup_{\alpha\ge 0}\ \min_{x\in\mathcal R_1}\ \widetilde L(x,\alpha).
\]
Substituting the definition of $\widetilde L$ gives
\[
-P^*
=
\sup_{\alpha\ge 0}\ \min_{x\in\mathcal R_1}
\left(-\mathrm{GFT}(x)-\alpha\,\Lambda_D(x)\right).
\]

Finally, multiply both sides by $-1$. Using
\[
-\sup_{\alpha\ge0}\min_{x\in\mathcal R_1} a(x,\alpha)
=
\inf_{\alpha\ge0}\max_{x\in\mathcal R_1}\left(-a(x,\alpha)\right),
\]
we obtain
\[
P^*
=
\inf_{\alpha\ge 0}\ \max_{x\in\mathcal R_1}
\left(\mathrm{GFT}(x)+\alpha\,\Lambda_D(x)\right).
\]
This is exactly the claimed representation.
\end{proof}

\begin{proof}[Proof of Lemma~\ref{lem:wlog-fS0}]
Let
\[
m:=\min\{s\in\mathbb Z_+ : f_S(s)>0\}.
\]
If $m=0$, there is nothing to prove.

We first claim that, in maximizing
\[
\mathrm{GFT}(x)+\alpha\,\Lambda_D(x),
\]
we may without loss of generality restrict attention to allocation rules satisfying
\[
x(s,b)=0\qquad \forall\, b<m.
\]
Indeed, using the virtual-value representation,
\[
\mathrm{GFT}(x)+\alpha\,\Lambda_D(x)
=
\sum_{s,b} x(s,b)\bigl(\phi_B(b)-\phi_S(s)\bigr)f_S(s)f_B(b).
\]
For every $s$ with $f_S(s)>0$, we have $s\ge m$, hence
\[
\phi_S(s)\ge (1+\alpha)s\ge (1+\alpha)m.
\]
For every $b<m$,
\[
\phi_B(b)\le (1+\alpha)b < (1+\alpha)m.
\]
Therefore
\[
\phi_B(b)-\phi_S(s)<0
\qquad\text{whenever } b<m\le s.
\]
Thus replacing $x(s,b)$ by $0$ on the region $b<m$ weakly increases
$\mathrm{GFT}(x)+\alpha\,\Lambda_D(x)$.
Since this modification only decreases lower-$b$ entries, it preserves monotonicity, so the modified rule still belongs to $\mathcal R_2$.

Hence we may restrict to rules with $x(s,b)=0$ for all $b<m$.
For such rules, all terms with $b<m$ vanish in the mechanism objective; and they also vanish in the first-best benchmark, because $b<m\le s$ implies $(b-s)_+=0$.

Now define the shifted seller distribution
\[
\widetilde f_S(s):=f_S(s+m),\qquad s\ge 0,
\]
so that $\widetilde f_S(0)=f_S(m)>0$. Moreover,
\[
\widetilde\phi_S(s)=\phi_S(s+m)-m(1+\alpha),
\]
hence $\widetilde F_S$ remains $\alpha$-weakly regular.

Fix any buyer distribution $F_B$, and let
\[
p:=\Pr_{b\sim F_B}[b\ge m].
\]
If $p=0$, then $b<m\le s$ almost surely, so both $\mathrm{FB}$ and $\mathrm{SB}$ are zero, and the claim is trivial.

Assume $p>0$, and define the conditional shifted buyer distribution
\[
\widetilde f_B(b):=\frac{f_B(b+m)}{p},\qquad b\ge 0.
\]
Its virtual value satisfies
\[
\widetilde\phi_B(b)=\phi_B(b+m)-m(1+\alpha).
\]

Given any $\widetilde x\in\mathcal R_2$, define
\[
x(s,b):=
\begin{cases}
0, & b<m,\\[3pt]
\widetilde x(s-m,b-m), & s\ge m,\ b\ge m.
\end{cases}
\]
(Values for $s<m$ are irrelevant since $f_S(s)=0$ there, and may be filled in arbitrarily so as to preserve monotonicity.)
Then $x\in\mathcal R_2$.

Using the change of variables $s'=s-m$, $b'=b-m$ on the region $s\ge m$, $b\ge m$, we obtain
\[
\mathrm{GFT}(x)+\alpha\,\Lambda_D(x)-\beta\,\mathrm{FB}
=
p\Bigl[
\mathrm{GFT}(\widetilde x)+\alpha\,\Lambda_D(\widetilde x)-\beta\,\widetilde{\mathrm{FB}}
\Bigr],
\]
where the quantities on the right are evaluated under $(\widetilde F_S,\widetilde F_B)$.

Therefore the desired inequality for $F_S$ holds for all buyer distributions $F_B$
if it holds for the shifted seller distribution $\widetilde F_S$ for all buyer distributions on $\mathbb Z_+$.
Since $\widetilde f_S(0)>0$, this proves the claim.
\end{proof}

\begin{proof}[Proof of Lemma~\ref{lem:position-wise-monotone}]
For fixed \(\alpha\), the Lagrangian objective can be written as
\[
\mathrm{GFT}(x)+\alpha\,\Lambda_D(x)
=
\sum_{s,b}x(s,b)\left(\phi_B(b)-\phi_S(s)\right)f_S(s)f_B(b),
\]
where \(\phi_B(b):=(1+\alpha)b - \alpha\frac{1-F_B(b)}{f_B(b)}\) is the buyer-side \(\alpha\)-virtual value. By construction, the ironed term
\(
\phi_B^{\text{ir}}(b)-\phi_S^{\text{ir}}(s)
\)
is weakly increasing in \(b\) and weakly decreasing in \(s\). 
Since the objective is linear in $x$ and the constraints defining ${\cal R}_1$ (Bayesian monotonicity) are weaker than the point-wise monotonicity constraints of ${\cal R}_2$, any optimal $x\in{\cal R}_1$ can be transformed into a point-wise monotone allocation in ${\cal R}_2$ via ``ironing'' (averaging $x$ over the intervals where the virtual values are constant). This transformation preserves feasibility and does not decrease the objective value.
\end{proof}

\begin{proof}[Proof of Lemma~\ref{lem:ironing-ratio}]
By construction, the integrated virtual value $\Phi_S(s):=\sum_{t=0}^s \phi_S(t)f_S(t)$ satisfies
\[
\Phi_S^{\text{ir}}(s)\ge \Phi_S(s)
\qquad\text{for all }s.
\]
Also, we have that 
\[
\phi_S(s)f_S(s)=(1+\alpha)s\,f_S(s)+\alpha F_S(s-1),
\]
so discrete summation by parts gives
\[
\Phi_S(s)
=
(1+\alpha)sF_S(s)-\sum_{t=0}^{s-1}F_S(t).
\]
Hence
\begin{equation}\label{eq:Phi-order}
\Phi_S(s)\le \Phi_S^{\text{ir}}(s)
\qquad\Longleftrightarrow\qquad
\sum_{t=0}^s tF_S(t)\ge \sum_{t=0}^s tF_S^{\text{ir}}(t).
\end{equation}

For any seller-side monotone interim allocation \(x_S\),
\[
\sum_{s=0}^N \phi_S(s)x_S(s)f_S(s)
=
\Phi_S(N)x_S(N)-\sum_{s=0}^{N-1}\Phi_S(s)\left(x_S(s+1)-x_S(s)\right),
\]
and similarly,
\[
\sum_{s=0}^N \phi_S^{\text{ir}}(s)x_S(s)f_S^{\text{ir}}(s)
=
\Phi_S^{\text{ir}}(N)x_S(N)-\sum_{s=0}^{N-1}\Phi_S^{\text{ir}}(s)\left(x_S(s+1)-x_S(s)\right).
\]
Since \(x_S(s+1)-x_S(s)\le 0\) and \(\Phi_S^{\text{ir}}\ge \Phi_S\), it follows that
\[
\sum_{s=0}^N \phi_S(s)x_S(s)f_S(s)
\le
\sum_{s=0}^N \phi_S^{\text{ir}}(s)x_S(s)f_S^{\text{ir}}(s).
\]
Conversely, because \(\phi_S^{\text{ir}}\) is constant on each ironed interval, an optimal allocation can always be chosen to be constant on these intervals (via averaging), preserving the objective value. Thus, $\mathrm{SB}$ is invariant under ironing,
\begin{equation}
    \label{eq:iron-sb}
\mathrm{SB}_{F_S^{\text{ir}},F_B}=\mathrm{SB}_{F_S,F_B}.
\end{equation}

For the first-best, 
\[
\mathrm{FB}_{F_S,F_B}
=
\sum_{s,b}(b-s)_+ f_S(s)f_B(b),
\]
and for each fixed \(b\), the function \(s\mapsto (b-s)_+\) is weakly decreasing. The order in \eqref{eq:Phi-order} therefore implies that
\begin{equation}
    \label{eq:iron-fb}
\mathrm{FB}_{F_S^{\text{ir}},F_B}\ge \mathrm{FB}_{F_S,F_B}.
\end{equation}
Combining Eqs.~\eqref{eq:iron-sb} and~\eqref{eq:iron-fb} yields the claimed ratio inequality.
\end{proof}

\begin{proof}[Proof of Theorem~\ref{thm:minimax-recursion-equivalence}]
In order to prove this result, we proceed in three main steps.
\paragraph{Step 1: Local Constancy of the Allocation.} 
If the left-hand side holds, there exists a feasible allocation $x$ such that 
$$\inf_{F_B\in\Delta(\mathbb Z_{\ge 0})}
    \left[
    \mathrm{GFT}(x)+\alpha\Lambda_D(x)-\beta\cdot \mathrm{FB}
    \right]\ge 0.$$
We first show that we can assume $x$ is ``locally constant'' where the feasibility constraint is not binding.

We consider a modified allocation rule $x^{(k+1)}$. Suppose there exists a value $b$ such that by setting $x^{(k+1)}(\cdot,k)=x(\cdot,k)$ for all $k<b$ and $x^{(k+1)}(\cdot, b)=x(\cdot,b-1)$, the condition ${H}(b;x^{(k+1)})\ge 0$ remains satisfied. 
If we define $x^{(k+1)}(\cdot,k)=x(\cdot,k)$ for all $k>b$, it follows that ${H}(t;x^{(k+1)})\ge 0$ holds for all $t$. 

To see this, we only need to verify the case $t>b$. In this regime, the coefficient associated with $x^{(k+1)}(\cdot,b)$ in the expression of $H(t;x^{(k+1)})$ is negative. Since $x^{(k+1)}(\cdot,b)=x(\cdot,b-1)\le x(\cdot,b)$ by the monotonicity of $x$, it follows that $H(t;x^{(k+1)})\ge H(t;x)\ge0$. Consequently, the modified rule $x^{(k+1)}$ also satisfies the lower bound:
$$
\inf_{F_B\in\Delta(\mathbb Z_{\ge 0})}
    \left[
    \mathrm{GFT}(x^{(k+1)})+\alpha\Lambda_D(x^{(k+1)})-\beta\cdot \mathrm{FB}
    \right]\ge 0.
$$
Therefore, without loss of generality, we may restrict our attention to allocation rules that are locally constant whenever possible; specifically, if the feasibility condition $H(b;x)\ge 0$ allows for the binding constraint $x(\cdot,b)=x(\cdot,b-1)$, we can assume this equality holds.

\begin{remark}
Before Step~2, we explain why the check
$$
\textbf{if } (1+\alpha)b \le \phi_S(k_{b-1}) \textbf{ then return {\sc fail}}
$$
in Algorithm~\ref{alg:canonical-recursion} enforces the structural constraint established in Step~2. When the algorithm reaches this line, the current state satisfies $H(b;x)<0$, implying a deficit that must be corrected by increasing $x(\cdot,b)$ for some seller type $s\in\{k_{b-1},\dots,b-1\}$. However, by the $\alpha$-weak regularity of $F_S$, the virtual value $\phi_S$ is weakly increasing, hence
$$
(1+\alpha)b \;\le\; \phi_S(k_{b-1})
\;\Longrightarrow\;
(1+\alpha)b \;\le\; \phi_S(s)
\quad \forall\, s \ge k_{b-1}.
$$
Eq.~\eqref{eq:virtual-value-cutoff} (established in Step~2) requires $x(s,b)=0$ for every such $s$. Since the algorithm searches exclusively over $s\in\{k_{b-1},\dots,b-1\}$, \emph{every} candidate seller type is forbidden by Eq.~\eqref{eq:virtual-value-cutoff}, and the deficit $H(b;x)<0$ is irremediable within the feasible domain. The algorithm therefore correctly certifies failure immediately, without entering the loop.
\end{remark}

\paragraph{Step 2: Virtual Value Cutoff.}
We now show that, without loss of generality, the allocation satisfies
\begin{equation}
(1+\alpha)\,b \;\le\; \phi_S(s) \;\;\Longrightarrow\;\; x(s,b)=0.
\label{eq:virtual-value-cutoff}
\end{equation}
Suppose otherwise. Among all pairs \((s,b)\) violating \eqref{eq:virtual-value-cutoff}, we choose the one with the \emph{smallest} \(b\), and the \emph{smallest} \(s\). Since \(F_S\) is \(\alpha\)-weakly regular, the virtual value \(\phi_S\) is weakly increasing, so
\[
\phi_S(s') \;\ge\; \phi_S(s) \;\ge\; (1+\alpha)\,b
\qquad \forall\; s'\ge s.
\]
We define the modified allocation \(x^{(k+1)}\) by
\[
x^{(k+1)}(s',b') \;:=\;
\begin{cases}
x(s',b'), & \text{if } b'\neq b \text{ or } s'< s,\\[4pt]
0, & \text{if } b'= b \text{ and } s'\ge s.
\end{cases}
\]

Buyer-side monotonicity is preserved because \(b\) is the smallest buyer type violating the condition and for $b'<b$, the values $x(s',b')$ already satisfy the cutoff.
Thus, \(x^{(k+1)}\in \mathcal R_2\).
Now we will verify \(H(t; x^{(k+1)})\ge H(t;x)\) for every \(t\in\mathbb Z_{\ge 0}\):
\begin{itemize}
\item \emph{Case \(t<b\).} The allocation \(x(\cdot,b)\) does not appear in the expression for \(H(t;x)\), so \(H(t;x^{(k+1)})=H(t;x)\).

\item \emph{Case \(t=b\).} From Eq.~\eqref{eq:Hb-def}, the coefficient of \(x(r,b)\) in \(H(b;x)\) is proportional to \(\left((1+\alpha)b-\phi_S(r)\right)f_S(r)\). For every \(r\ge s\) this coefficient is non-positive, since \(\phi_S(r)\ge(1+\alpha)b\). Setting these entries to zero can therefore only increase \(H(b;x)\), whence \(H(b;x^{(k+1)})\ge H(b;x)\ge 0\).

\item \emph{Case \(t>b\).} Value \(b\) enters \(H(t;x)\) solely through the cross-term \(-\alpha\sum_{u<t}f_S(u)\,x(u,b)\), which has a non-positive sign. Reducing entries of \(x(\cdot,b)\) to zero weakly increases this term, so \(H(t;x^{(k+1)})\ge H(t;x)\ge 0\).
\end{itemize}
Since \(H(t;x^{(k+1)})\ge 0\) for every \(t\), the allocation \(x^{(k+1)}\) still guarantees
\[
\inf_{F_B\in\Delta(\mathbb Z_{\ge 0})}
\left[
\mathrm{GFT}(x^{(k+1)})+\alpha\,\Lambda_D(x^{(k+1)})-\beta\cdot\mathrm{FB}
\right]\ge 0.
\]
Iterating this argument over successive buyer types, we may assume that~\eqref{eq:virtual-value-cutoff} holds everywhere without loss of generality.

\paragraph{Step 3: Optimality of Threshold-Form Allocations.}
Take any \(x\in\mathcal R_2\). For any \(b\ge 1\), if \(x(s-1,b)<1\) and \(x(s,b)>0\), we move mass from $s$ to $s-1$ while keeping $H(b;x)$ constant. Let \(\delta>0\) be the reduction at \(s\), and \(\varepsilon>0\) be the increase at \(s-1\). Since the contribution of the value \(b\) to \(H(b;x)\) is
\[
\sum_{r<b}\left((1+\alpha)b-\phi_S(r)\right)f_S(r)x(r,b),
\]
preserving \(H(b;x)\) requires
\[
\left((1+\alpha)b-\phi_S(s-1)\right)f_S(s-1)\varepsilon
=
\left((1+\alpha)b-\phi_S(s)\right)f_S(s)\delta.
\]
Equivalently,
\[
\frac{f_S(s-1)\varepsilon}{f_S(s)\delta}
=
\frac{(1+\alpha)b-\phi_S(s)}{(1+\alpha)b-\phi_S(s-1)}.
\]
Since value \(b\) enters \(H(t;x)\) only through
\[
-\alpha\sum_{u<b} f_S(u)x(u,b),
\]
and \(\phi_S(s)\) is weakly increasing in \(s\),
for every \(t>b\),
\[
\frac{(1+\alpha)b-\phi_S(s)}{(1+\alpha)b-\phi_S(s-1)}
\le
1.
\]
Then this allocation shift implies
\[
H(t;x^{(k+1)})\ge H(t;x)\qquad\text{for every }t>b.
\]

Consequently, without changing \(H(b;x)\), we may replace \(x(\cdot,b)\) by a smallest value threshold, and this weakly increases every later \(H(t;x)\). Applying this transformation iteratively, we may restrict attention threshold-form allocations of the type
\[
\left(1,\ldots,1,x^*,0,\ldots,0\right).
\]

Note that this leftward shifting preserves property~\eqref{eq:virtual-value-cutoff}: all non-zero entries remain in the region where \((1+\alpha)b > \phi_S(s)\), so the coefficients \(\left((1+\alpha)b-\phi_S(s)\right)\) appearing in the shift ratio are strictly positive and the transformation is well-defined.

In summary, Steps~1–3 show that the minimax condition is equivalent to the existence of a threshold-form allocation satisfying $H(b;x)\ge 0$ for any $b\in\mathbb{Z}_+$. The greedy Algorithm~\ref{alg:canonical-recursion} constructs exactly such an allocation or correctly identifies its non-existence.
\end{proof}

\section{Omitted Proofs in Section~\ref{sec:optimality}}
\label{app:optimality}

\begin{proof}[Proof of Proposition~\ref{prop:mono}]
Run Algorithm~\ref{alg:state-recursion} from \(\omega\) and \(\omega'\) with the same
ratio sequence \(\mathbf r\). We first assume that both \(\mathrm{ALG}_2(\omega;\mathbf r)\) and \(\mathrm{ALG}_2(\omega';\mathbf r)\) are not \textsc{fail}. Let \(x\) and \(x'\) denote the two resulting
allocation rules. Since \(\omega\) and \(\omega'\) are comparable and both runs use the same appended ratio sequence, the output states \(\mathrm{ALG}_2(\omega;\mathbf r)\) and \(\mathrm{ALG}_2(\omega';\mathbf r)\) are also comparable.

We argue by contradiction. Suppose that the monotonicity statement fails.
Let \(b\) be the first step at which the threshold index of \(x\) is no longer less than that of \(x'\). Then for every earlier step \(t<b\),
\[
x(s,t)\le x'(s,t)
\qquad\text{for } s<t.
\]
Since \(b\) is the first violating step, \(\mathrm{ALG}_2(\omega;\mathbf r)\) must output exactly at a clear state; otherwise the allocation row at step \(b\) would be copied from the previous step.
Hence the criterion of \(x\) at step \(b\) is exactly \(\hat H(b; x)=0.\)
Therefore,
every term that comes from a past step and is multiplied by
\(-\alpha\) is smaller for \(x\) than for \(x'\). 

On the other hand, at step \(b\), we compare the contribution to
\(\hat H(b;\cdot)\) in \eqref{eq:Hhat} from the allocation variables
\(x(s,b)\) associated with buyer value \(b\). By the choice of \(b\), the threshold order is preserved for every earlier
buyer value \(t<b\). At buyer value \(b\), this argument changes direction:
the threshold chosen in the run from \(\omega\) becomes weakly larger than the
threshold chosen in the run from \(\omega'\).
Thus, after comparing
the two \(\hat H\)-formulas term by term, we obtain
\[
\hat H(b;x')<\hat H(b;x)=0,
\]
a contradiction!
Therefore the threshold index of \(x\) is never larger than \(x'\),
i.e. 
\(\mathrm{ALG}_2(\omega;\mathbf r)\prec \mathrm{ALG}_2(\omega';\mathbf r).\)

Next, we prove that if \(\mathrm{ALG}_2(\omega;\mathbf r)=\textsc{fail}\), then \(\mathrm{ALG}_2(\omega';\mathbf r)=\textsc{fail}\) and if \(\mathrm{ALG}_2(\omega';\mathbf r)\neq\textsc{fail}\), then \(\mathrm{ALG}_2(\omega;\mathbf r)\neq\textsc{fail}\). We prove this by
contradiction. Suppose that
\(\mathrm{ALG}_2(\omega;\mathbf r)=\textsc{fail},
\mathrm{ALG}_2(\omega';\mathbf r)\neq\textsc{fail}.\)
Let \(b\) be the first step at which the run from \(\omega\) returns
\textsc{fail}. Applying
the ordering argument above to step \(b-1\),
the partially constructed states remain ordered. 

We first show that the failure at step \(b\) cannot be caused by the
condition
\((1+\alpha)b \le \phi
_S(k_{b-1}).\)
Let \(b'\le b\) be the corresponding buyer index in the run
from \(\omega'\). Since the threshold index for \(x'\)
is weakly larger and \(\phi_S\) is nondecreasing, we have
\(\phi_S(k_{b-1})\le \phi_S(k'_{b'-1}).\)
Hence
\[
(1+\alpha)b-\phi_S(k_{b-1})
\ge
(1+\alpha)b'-\phi_S(k'_{b'-1})
\ge 0.
\]

Therefore the failure of the run from \(\omega\) must be caused by the case that
the algorithm cannot find any value \(x(k_b,b)\in(0,1]\) satisfying
\(\hat H(b;x)=0\). By the ordering argument above, if the more
demanding run \(\mathrm{ALG}_2(\omega';\mathbf r)\) can choose a threshold and a fractional
value in \((0,1]\) that makes \(\hat H(b';x')=0\), then the run \(\mathrm{ALG}_2(\omega;\mathbf r)\)
must also be able to find such an achievable state. Hence, in the run from \(\omega\), the candidate-threshold search in
Line~\ref{line:state-transition-threshold-search} cannot fall through to the failure branch
with \(\mathrm{found}=\mathrm{false}\) in Line~\ref{line:state-transition-no-found-fail}.

Thus \(\mathrm{ALG}_2(\omega;\mathbf r)\) cannot fail whenever
\(\mathrm{ALG}_2(\omega';\mathbf r)\neq\textsc{fail}\). Combining this
with the non-failure ordering argument proved above gives the final statement.
\end{proof}

\begin{proof}[Proof of Lemma~\ref{lem:mono-G}]
From \eqref{eq:big-jump-rec} and $x_b=\beta G_b$, we have, for every $b\ge 1$,
\[
G_{b+1}
=
\left(1-\frac{1}{(b+1)(1+\alpha)}\right)G_b
+
\frac{1}{(b+1)(1+\alpha)}\frac{F_S(b)}{F_S(0)} .
\]
Hence
\[
G_{b+1}-G_b
=
\frac{1}{(b+1)(1+\alpha)}
\left(\frac{F_S(b)}{F_S(0)}-G_b\right).
\]
It remains to show that \(G_b\le F_S(b)/F_S(0)\). By \eqref{eq:Gm-def} and the monotonicity of \(F_S\),
\[
\begin{aligned}
G_b
&=
\sum_{t=1}^{b}
\frac{F_S(t-1)}{t(1+\alpha)F_S(0)}
\prod_{j=t+1}^{b}\left(1-\frac{1}{j(1+\alpha)}\right)  \\
&<
\frac{F_S(b)}{F_S(0)}
\sum_{t=1}^{b}
\frac{1}{t(1+\alpha)}
\prod_{j=t+1}^{b}\left(1-\frac{1}{j(1+\alpha)}\right).
\end{aligned}
\]
Moreover, the last sum telescopes:
\[
\sum_{t=1}^{b}
\frac{1}{t(1+\alpha)}
\prod_{j=t+1}^{b}\left(1-\frac{1}{j(1+\alpha)}\right)
=
1-\prod_{j=1}^{b}\left(1-\frac{1}{j(1+\alpha)}\right)
\le 1 .
\]
Therefore \(G_b< F_S(b)/F_S(0)\), and consequently
\[
G_{b+1}-G_b> 0 .
\]
Thus \(G_b\) is strictly increasing in \(b\).
\end{proof}

\begin{proof}[Proof of Lemma~\ref{lem:alpha-zero-never-fail}]
When \(\alpha=0\), the seller virtual value is \(\phi_S(j)=j\). Hence the
returning \textsc{fail} in Line~\ref{line:state-transition-return-fail-virtual}
would require
\[
    b\le \phi_S(k_b)=k_b .
\]
This is impossible, because Algorithm~\ref{alg:state-recursion} only considers
threshold indices \(k_b\le b-1\). Thus the only remaining possible way for the
algorithm to return \textsc{fail} is that the candidate-threshold search in
Line~\ref{line:state-transition-threshold-search} exhausts all candidate threshold
indices and then reaches Line~\ref{line:state-transition-no-found-fail} without
finding a threshold allocation in \((0,1]\) satisfying the clear-step equation.

We rule this out. Fix a buyer value \(b\), and let \(b_0<b\) be the buyer value
of the clear state from which the normalized equation
\(\hat H(b;\cdot)\) in \eqref{eq:Hhat} is computed. Suppose that, after copying
the allocation from the preceding buyer value, i.e., $x(\cdot, b)=x(\cdot, b-1)$, the algorithm enters the
candidate-threshold search in Line~\ref{line:state-transition-threshold-search} at
step \(b\), and we have
\[
    \hat H(b;x)<0 .
\]

For every seller value \(s\) appearing in the summation over \(1\le s<b_0\)
or in the summation over \(b_0\le s<b\) in \eqref{eq:Hhat}, define
\[
    w_s
    :=
    \prod_{i=1}^{s-k_{b_0}}
    \left(1+\frac{1}{r_i}\right)
    \frac{1}{r_s},
    \qquad 1\le s<b,
\]

where the product is interpreted as \(1\) whenever the index set is empty.
This quantity is exactly the multiplicative factor attached in
\eqref{eq:Hhat} to the summand associated with seller value \(s\). In
particular, \(w_s>0\) for every \(1\le s<b\).

When \(\alpha=0\), the coefficient of \(x(s,b)\) in \(\hat H(b;x)\) is
\(
    (b-s)w_s>0 .
\)
Therefore \(\hat H(b;\cdot)\) is strictly increasing in every variable
\(x(s,b)\) that Line~\ref{line:state-transition-update-kb} sets to \(1\)
at buyer value \(b\).

Let \(\bar x\) be obtained from the allocation 
by setting
\[
    \bar x(s,b)=1,\qquad s=0,\ldots,b-1,
\]
while leaving all allocation vectors at buyer values different from \(b\)
unchanged, i.e., $\bar x(\cdot, <b)=x(\cdot,<b)$. We claim that \(\hat H(b;\bar x)>0\). Substituting
\(\alpha=0\), \(\beta=\tfrac12\), and \(\bar x(s,b)=1\) into
\eqref{eq:Hhat} gives
\[
\begin{aligned}
\hat H(b;\bar x)
={}&
\frac{b-b_0}{2}
+
\sum_{0< s<b_0}
 w_s
\left(
    (b-s)-(b_0-s)x(s,b_0)-\frac{b-b_0}{2}
\right) +
\sum_{b_0\le s<b}
 w_s\frac{b-s}{2},
\end{aligned} 
\]

where \(b_0\) is the buyer value from the last clear state. 
For every \(1\le s<b_0\), since \(x(s,b_0)\le 1\),
\[
    (b-s)-(b_0-s)x(s,b_0)-\frac{b-b_0}{2}
    \ge
    (b-s)-(b_0-s)-\frac{b-b_0}{2}
    =
    \frac{b-b_0}{2}
    >0.
\]
For every \(b_0\le s<b\), we also have
\(
    \frac{b-s}{2}>0.
\)
Together with \(w_s>0\) and \((b-b_0)/2>0\), this proves
\[
    \hat H(b;\bar x)>0 .
\]

Starting from the allocation immediately before the candidate-threshold search
in Line~\ref{line:state-transition-threshold-search}, for which
\(\hat H(b;x)<0\), consider the process of setting the variables
\(x(s,b)\), \(s=0,\ldots,b-1\), to \(1\) one by one in the same order as
Line~\ref{line:state-transition-threshold-search}. At each trial \(k_b\), all
variables except \(x(k_b,b)\) are fixed.
Since the coefficient of \(x(k_b,b)\) in \(\hat H(b;x)\) is strictly positive,
\(\hat H(b;x)\), viewed only as a function of \(x(k_b,b)\), is a strictly
increasing linear function.

After all these variables are set to \(1\), the resulting allocation is
\(\bar x\), and we have proved that \(\hat H(b;\bar x)>0\). Therefore, during
this sequential update process, there must be a first index \(k_b\) such that
increasing \(x(k_b,b)\) to \(1\) makes \(\hat H(b;x)\) weakly positive.
For this index, before increasing \(x(k_b,b)\) the value of \(\hat H(b;x)\) is
negative, while after setting \(x(k_b,b)=1\) it is nonnegative. Since
\(\hat H(b;x)\) is a strictly increasing linear function of \(x(k_b,b)\), there
exists a value
\(
    x(k_b,b)\in(0,1]
\)
such that
\[
    \hat H(b;x)=0 .
\]
Thus the algorithm finds a feasible  allocation at step \(b\), and
cannot return \textsc{fail}.
\end{proof}

\begin{proof}[Proof of Proposition~\ref{prop:regular-never-fail}]
Normalize \(F_S(0)=1\), and write
\(\kappa:=\frac{1+\alpha}{\alpha}.\)
First consider \(b=1\). Then
\[
    G_1=\frac{F_S(0)}{1+\alpha}=\frac{1}{1+\alpha},
    \qquad
    x_1=\beta G_1=\frac{1}{2(1+\alpha)}<1.
\]
Moreover \(\phi^\alpha_S(0)=0<(1+\alpha)\). Hence the threshold-zero candidate is
feasible, and the algorithm cannot reach, and therefore cannot fail at,
threshold index \(t_b=1\).

Now fix \(b\ge 2\). Suppose toward a contradiction that
Algorithm~\ref{alg:state-recursion} returns \textsc{fail} at threshold index
\(t_b=1\). Then the virtual-value test at \(t_b=1\) gives
\((1+\alpha)b\le \phi^\alpha_S(1).\)
Since the ratio sequence is \(\alpha\)-weakly regular, \(\phi^\alpha_S\) is weakly
increasing. Thus, for every \(1\le m\le b-1\),
\(\phi^\alpha_S(m)\ge (1+\alpha)b.\)
Using
\(\phi^\alpha_S(m) = (1+\alpha)m+ \alpha\frac{F_S(m{-}1)}{f_S(m)},\)
we obtain
\(\frac{F_S(m{-}1)}{f_S(m)} \ge \kappa(b-m), 1\le m\le b-1 .\)
Therefore, for every \(0\le i\le b-1\),
\[
\begin{aligned}
    F_S(i)
    &=
    \prod_{m=1}^{i}
    \left(
        1+\frac{1}{F_S(m{-}1)/f_S(m)}
    \right) \le
    \prod_{m=1}^{i}
    \left(
        1+\frac{1}{\kappa(b-m)}
    \right)
    =
    \widetilde F_S^{(b+1)}(i),
\end{aligned}
\]
where \(\widetilde F_S^{(b+1)}\) denotes the normalized CDF of the \(\phi^\alpha_S=(1+\alpha)(N-1)\) instance with support size \(b+1\).

By the nonnegative-coefficient representation of \(G_b\) in
\eqref{eq:Gm-def}, this pointwise CDF domination implies
\(G_b\le \widetilde G_b^{(b+1)}.\)
The right-hand side is exactly the extremal constant-\(\phi^\alpha_S\) case.
By Lemma~\ref{lem:zero-ratio-case},
\(\widetilde G_b^{(b+1)}<2.\)
Hence
\(0<x_b=\beta G_b=\frac12G_b<1.\)
In addition, the virtual-value test at threshold zero is satisfied:
\(\phi^\alpha_S(0)=0<(1+\alpha)b.\)
Thus the threshold-zero candidate is feasible and is considered before
threshold index \(t_b=1\). This contradicts the assumption that the algorithm
returns \textsc{fail} at \(t_b=1\).

Therefore Algorithm~\ref{alg:state-recursion} never returns \textsc{fail} at
threshold index \(t_b=1\).
\end{proof}

\begin{proof}[Proof of Lemma~\ref{lem:zero-ratio-case}]
Throughout the proof, empty products are understood as \(1\).  Let
\(a:=\alpha/(1+\alpha)\).  Since \(\alpha>0\), we have \(a\in(0,1)\) and
\(1-a=1/(1+\alpha)\in(0,1)\).  

\textbf{1. Monotonicity Proof}

Recall that, in the case \(r=0\),
\[
\widetilde F_S^{(N)}(t-1)
=
\prod_{j=N-t}^{N-2}\left(1+\frac{a}{j}\right),
\qquad 1\le t\le N-1,
\]
and
\[
\widetilde G_{N-1}^{(N)}
=
\sum_{t=1}^{N-1}
\frac{\widetilde F_S^{(N)}(t-1)}{t(1+\alpha)}
\prod_{j=t+1}^{N-1}
\left(1-\frac{1}{j(1+\alpha)}\right).
\]

We first prove strict monotonicity in \(N\).  Fix \(N\ge2\) and
\(1\le t\le N-1\).  Since
\(\widetilde F_S^{(N+1)}(t) = \widetilde F_S^{(N+1)}(t-1) \left(1+\frac{a}{N-t}\right),\)
the increment
\(\widetilde F_S^{(N+1)}(t)-\widetilde F_S^{(N+1)}(t-1)
=
\frac{a}{N-t}\widetilde F_S^{(N+1)}(t-1)\) is strictly positive.
Moreover,
\[
\widetilde F_S^{(N)}(t-1)
=
\frac{1+\frac{a}{N-t}}{1+\frac{a}{N-1}}\,
\widetilde F_S^{(N+1)}(t-1).
\]
Hence
\[
\begin{aligned}
\widetilde F_S^{(N)}(t-1)
-
\widetilde F_S^{(N+1)}(t-1)
&=
\left[
\frac{1+\frac{a}{N-t}}{1+\frac{a}{N-1}}-1
\right]
\widetilde F_S^{(N+1)}(t-1)  \\
&=
\frac{a(t-1)}
{(N-t)(N-1+a)}
\widetilde F_S^{(N+1)}(t-1)  \\
&=
\frac{t-1}{N-1+a}
\left[
\widetilde F_S^{(N+1)}(t)
-
\widetilde F_S^{(N+1)}(t-1)
\right].
\end{aligned}
\]
It follows that
\begin{equation}
\label{eq:GNminus1-diff-r0-alpha-positive}
\begin{aligned}
\widetilde G_{N-1}^{(N)}
-
\widetilde G_{N-1}^{(N+1)}
&=
\sum_{t=1}^{N-1}
\frac{t-1}{t(1+\alpha)(N-1+a)}
\left[
\prod_{j=t+1}^{N-1}
\left(1-\frac{1}{j(1+\alpha)}\right)
\right] \cdot
\left[
\widetilde F_S^{(N+1)}(t)
-
\widetilde F_S^{(N+1)}(t-1)
\right].
\end{aligned}
\end{equation}

On the other hand, by the recursive form of \(\widetilde G^{(N+1)}\),
\begin{equation}
\label{eq:GN-step-r0-alpha-positive}
\widetilde G_N^{(N+1)}
-
\widetilde G_{N-1}^{(N+1)}
=
\frac{
\widetilde F_S^{(N+1)}(N-1)
-
\widetilde G_{N-1}^{(N+1)}
}
{N(1+\alpha)}.
\end{equation}
For each \(1\le t\le N-1\), since
\(\frac{1}{t(1+\alpha)}
=
1-\left(1-\frac{1}{t(1+\alpha)}\right)\), we have
\[
\begin{aligned}
\frac{1}{t(1+\alpha)}
\prod_{j=t+1}^{N-1}
\left(1-\frac{1}{j(1+\alpha)}\right)
&=
\prod_{j=t+1}^{N-1}
\left(1-\frac{1}{j(1+\alpha)}\right)  -
\prod_{j=t}^{N-1}
\left(1-\frac{1}{j(1+\alpha)}\right).
\end{aligned}
\]
Thus, for every \(1\le s\le N-1\),
\begin{equation}
\label{eq:partial-sum-r0-alpha-positive}
\sum_{t=s}^{N-1}
\frac{1}{t(1+\alpha)}
\prod_{j=t+1}^{N-1}
\left(1-\frac{1}{j(1+\alpha)}\right)
=
1-
\prod_{j=s}^{N-1}
\left(1-\frac{1}{j(1+\alpha)}\right).
\end{equation}

Using
\(\widetilde F_S^{(N+1)}(t-1)
=
1+\sum_{s=1}^{t-1}
[
\widetilde F_S^{(N+1)}(s)
-
\widetilde F_S^{(N+1)}(s-1)
]\), we obtain
\[
\begin{aligned}
\widetilde G_{N-1}^{(N+1)}
&=
\sum_{t=1}^{N-1}
\frac{1}{t(1+\alpha)}
\left[
\prod_{j=t+1}^{N-1}
\left(1-\frac{1}{j(1+\alpha)}\right)
\right] \cdot
\left[
1+\sum_{s=1}^{t-1}
\left(
\widetilde F_S^{(N+1)}(s)
-
\widetilde F_S^{(N+1)}(s-1)
\right)
\right]  \\
&=
1-
\prod_{j=1}^{N-1}
\left(1-\frac{1}{j(1+\alpha)}\right) +
\sum_{s=1}^{N-1}
\left[
1-
\prod_{j=s+1}^{N-1}
\left(1-\frac{1}{j(1+\alpha)}\right)
\right] \cdot
\left[
\widetilde F_S^{(N+1)}(s)
-
\widetilde F_S^{(N+1)}(s-1)
\right],
\end{aligned}
\]
where \eqref{eq:partial-sum-r0-alpha-positive} was used in the second
equality.  Since
\[
\widetilde F_S^{(N+1)}(N-1)
=
1+
\sum_{s=1}^{N-1}
\left[
\widetilde F_S^{(N+1)}(s)
-
\widetilde F_S^{(N+1)}(s-1)
\right],
\]
we get
\begin{equation}
\label{eq:FminusG-r0-alpha-positive}
\begin{aligned}
&\qquad\widetilde F_S^{(N+1)}(N-1)
-
\widetilde G_{N-1}^{(N+1)}  \\
& =
\prod_{j=1}^{N-1}
\left(1-\frac{1}{j(1+\alpha)}\right) +
\sum_{s=1}^{N-1}
\left[
\prod_{j=s+1}^{N-1}
\left(1-\frac{1}{j(1+\alpha)}\right)
\right]
\left[
\widetilde F_S^{(N+1)}(s)
-
\widetilde F_S^{(N+1)}(s-1)
\right].
\end{aligned}
\end{equation}
Substituting \eqref{eq:FminusG-r0-alpha-positive} into
\eqref{eq:GN-step-r0-alpha-positive} gives
\begin{equation}
\label{eq:GN-step-expanded-r0-alpha-positive}
\begin{aligned}
&\qquad\widetilde G_N^{(N+1)}
-
\widetilde G_{N-1}^{(N+1)}\\
&=
\frac{1}{N(1+\alpha)}
\prod_{j=1}^{N-1}
\left(1-\frac{1}{j(1+\alpha)}\right) +
\sum_{t=1}^{N-1}
\frac{1}{N(1+\alpha)}
\left[
\prod_{j=t+1}^{N-1}
\left(1-\frac{1}{j(1+\alpha)}\right)
\right] \cdot
\left[
\widetilde F_S^{(N+1)}(t)
-
\widetilde F_S^{(N+1)}(t-1)
\right].
\end{aligned}
\end{equation}
Subtracting \eqref{eq:GNminus1-diff-r0-alpha-positive} from
\eqref{eq:GN-step-expanded-r0-alpha-positive}, we obtain
\[
\begin{aligned}
&\qquad\widetilde G_N^{(N+1)}
-
\widetilde G_{N-1}^{(N)}\\
&=
\frac{1}{N(1+\alpha)}
\prod_{j=1}^{N-1}
\left(1-\frac{1}{j(1+\alpha)}\right) +
\sum_{t=1}^{N-1}
\left[
\frac{1}{N(1+\alpha)}
-
\frac{t-1}{t(1+\alpha)(N-1+a)}
\right]  \\
&\qquad\qquad\qquad\qquad\qquad\qquad\qquad\qquad \cdot
\left[
\prod_{j=t+1}^{N-1}
\left(1-\frac{1}{j(1+\alpha)}\right)
\right]
\left[
\widetilde F_S^{(N+1)}(t)
-
\widetilde F_S^{(N+1)}(t-1)
\right].
\end{aligned}
\]
The coefficient in brackets equals
\(\frac{N-t+at}{Nt(1+\alpha)(N-1+a)}\), which is strictly positive for
\(1\le t\le N-1\).  Moreover,
\(1-\frac{1}{j(1+\alpha)}>0\) for every \(j\ge1\), and the increments
\(\widetilde F_S^{(N+1)}(t)-\widetilde F_S^{(N+1)}(t-1)\) are strictly
positive.  Hence
\(\widetilde G_N^{(N+1)}-\widetilde G_{N-1}^{(N)}>0\).  Thus
\(\{\widetilde G_{N-1}^{(N)}\}_{N\ge2}\) is strictly increasing.

\vspace{0.5cm}

\textbf{2. Limitation Proof}

It remains to compute the limit as \(N\to\infty\).  We first rewrite each
summand in Gamma-function form.  Since
\[
\widetilde F_S^{(N)}(t-1)
=
\prod_{j=N-t}^{N-2}
\left(1+\frac{a}{j}\right),
\]
we have
\begin{equation}
\label{eq:F-gamma-alpha-positive}
\widetilde F_S^{(N)}(t-1)
=
\frac{\Gamma(N-1+a)}{\Gamma(N-t+a)}
\frac{\Gamma(N-t)}{\Gamma(N-1)}.
\end{equation}
Also, using \(1-a=1/(1+\alpha)\), we have
\(1-\frac{1}{j(1+\alpha)}=1-\frac{1-a}{j}=\frac{j-1+a}{j}\). Therefore
\begin{equation}
\label{eq:prod-gamma-alpha-positive}
\begin{aligned}
\prod_{j=t+1}^{N-1}
\left(1-\frac{1}{j(1+\alpha)}\right)
&=
\prod_{j=t+1}^{N-1}
\frac{j-1+a}{j} =
\frac{\Gamma(N-1+a)}{\Gamma(t+a)}
\frac{\Gamma(t+1)}{\Gamma(N)}.
\end{aligned}
\end{equation}
Combining \eqref{eq:F-gamma-alpha-positive} and
\eqref{eq:prod-gamma-alpha-positive}, and using
\(1/(1+\alpha)=1-a\), \(\Gamma(t+1)=t\Gamma(t)\), and
\(\Gamma(N)=(N-1)\Gamma(N-1)\), the \(t\)-th summand satisfies
\begin{equation}
\label{eq:summand-gamma-alpha-positive}
\begin{aligned}
&\frac{\widetilde F_S^{(N)}(t-1)}{t(1+\alpha)}
\prod_{j=t+1}^{N-1}
\left(1-\frac{1}{j(1+\alpha)}\right) =
\frac{1-a}{N-1}
\left(
\frac{\Gamma(N-1+a)}{\Gamma(N-1)}
\right)^2
\frac{\Gamma(t)}{\Gamma(t+a)}
\frac{\Gamma(N-t)}{\Gamma(N-t+a)}.
\end{aligned}
\end{equation}

By the standard Gamma-ratio asymptotic formula
\[
\frac{\Gamma(x+b)}{\Gamma(x+c)}
=
x^{b-c}
\left(1+O\left(\frac1x\right)\right),
\qquad x\to\infty,
\]
we obtain, whenever \(t/(N-1)\to x\in(0,1)\),
\[
\frac{\widetilde F_S^{(N)}(t-1)}{t(1+\alpha)}
\prod_{j=t+1}^{N-1}
\left(1-\frac{1}{j(1+\alpha)}\right)
=
\frac{1-a}{N-1}
x^{-a}(1-x)^{-a}
(1+o(1)).
\]
Since \(a<1\), the function \(x^{-a}(1-x)^{-a}\) is integrable on
\((0,1)\).  Moreover, the two boundary sums with
\(t\le \delta(N-1)\) and \(t\ge (1-\delta)(N-1)\) are bounded, uniformly in
\(N\), by a constant depending only on \(a\) times \(\delta^{1-a}\).  Therefore
the preceding pointwise asymptotics yield the Riemann-sum limit
\[
\lim_{N\to\infty}
\widetilde G_{N-1}^{(N)}
=
(1-a)
\int_0^1 x^{-a}(1-x)^{-a}\,dx .
\]
Using the Beta-function identity,
\[
\int_0^1 x^{-a}(1-x)^{-a}\,dx
=
B(1-a,1-a)
=
\frac{\Gamma(1-a)^2}{\Gamma(2(1-a))},
\]
we obtain
\(\lim_{N\to\infty} \widetilde G_{N-1}^{(N)} = (1-a) \frac{\Gamma(1-a)^2}{\Gamma(2(1-a))}.\)
Equivalently, since \(1-a=1/(1+\alpha)\),
\[
\lim_{N\to\infty}
\widetilde G_{N-1}^{(N)}
=
\frac{1}{1+\alpha}
\frac{
\Gamma\!\bigl(\tfrac{1}{1+\alpha}\bigr)^2
}{
\Gamma\!\bigl(\tfrac{2}{1+\alpha}\bigr)
}.
\]

Finally, we analyze the supremum over \(\alpha>0\).  By the Weierstrass product
formula for the Gamma function, for \(a\in(0,1)\),
\[
(1-a)
\frac{\Gamma(1-a)^2}{\Gamma(2(1-a))}
=
2\prod_{n=1}^{\infty}
\frac{1+2(1-a)/n}{(1+(1-a)/n)^2}.
\]
Taking logarithms and differentiating term by term on compact subintervals of
\((0,1)\), we get
\[
\begin{aligned}
\frac{d}{da}
\log\left[
(1-a)
\frac{\Gamma(1-a)^2}{\Gamma(2(1-a))}
\right] & =
2\sum_{n=1}^{\infty}
\left(
\frac{1}{n+1-a}
-
\frac{1}{n+2(1-a)}
\right)  \\
& =
2(1-a)
\sum_{n=1}^{\infty}
\frac{1}{(n+1-a)(n+2(1-a))}
>0.
\end{aligned}
\]
Thus \((1-a)\Gamma(1-a)^2/\Gamma(2(1-a))\) is strictly increasing in
\(a\in(0,1)\).  Since \(a=\alpha/(1+\alpha)\) is strictly increasing in
\(\alpha>0\), the limit is strictly increasing in \(\alpha>0\).

As \(\alpha\to\infty\), we have \(a\to1^-\) and \(1-a\to0^+\).  Using
\(\Gamma(1-a)\sim 1/(1-a)\) and
\(\Gamma(2(1-a))\sim 1/[2(1-a)]\), we obtain
\((1-a) \frac{\Gamma(1-a)^2}{\Gamma(2(1-a))} \longrightarrow 2.\)
Therefore,
\[
\sup_{\alpha>0}
\frac{1}{1+\alpha}
\frac{
\Gamma\!\bigl(\tfrac{1}{1+\alpha}\bigr)^2
}{
\Gamma\!\bigl(\tfrac{2}{1+\alpha}\bigr)
}
=
2.
\]
The supremum is not attained at any finite \(\alpha>0\), and is approached
only in the limit \(\alpha\to\infty\).  This completes the proof.
\end{proof}

\begin{proof}[Proof of Proposition~\ref{prop:success-stationary-phi}]
The only nontrivial step is to prove $\mathcal D\ge0$.  We minimize
$\mathcal D$ over the admissible weakly increasing virtual values.  Introduce
\(\phi^\alpha_i := (i+1)(1+\alpha)+\alpha r_{i+1}^{(k+2)}, \  i=0,1,\ldots,k.\)
Thus
\(1+\frac1{r_{i+1}^{(k+2)}} = 1+\frac{\alpha}{\phi^\alpha_i-(i+1)(1+\alpha)}.\)
The weak regularity constraints are
\(\phi^\alpha_0\le\phi^\alpha_1\le\cdots\le\phi^\alpha_k.\)

Let $F_{k+2}$ and $F_{k+1}$ denote the normalized CDFs generated by the
$k+2$ and $k+1$ ratio sequences, respectively.  Define
\[
K(\phi^\alpha_0,\ldots,\phi^\alpha_k)
:=
\sum_{m=0}^{k+1}
\frac{F_{k+2}(m)}{2(1+\alpha)(m+1)}
\prod_{j=m+2}^{k+2}
\left(1-\frac1{(1+\alpha)j}\right),
\]
and
\[
J(\phi^\alpha_0,
\ldots,\phi^\alpha_k)
:=
\sum_{m=1}^{k+1}
\frac{\alpha(k+1-m)F_{k+2}(m)}{2(1+\alpha)m(m+1)}
\prod_{j=m+2}^{k+1}
\left(1-\frac1{(1+\alpha)j}\right).
\]
If $K\le1$, then the threshold-zero candidate at $b=k+2$ is feasible and the
$k+2$ output does not have a proper suffix.  Hence, in the proper-suffix case,
we must have $K\ge1$.

Moreover, for $m=0,\ldots,k$,
\(F_{k+2}(m+1) = \left(1+\frac1{r_1^{(k+2)}}\right)F_{k+1}(m).\)
Substituting the threshold-stationary expansions of $p_{k+1}$ and $p_{k+2}$
into \eqref{eq:D-zero-based} gives
\[
    \mathcal D
    =
    J(\phi^\alpha_0,\ldots,\phi^\alpha_k)
    +1
    -\frac{k+2}{2}
    \prod_{j=2}^{k+2}\left(1-\frac1{(1+\alpha)j}\right).
\]
Thus minimizing $\mathcal D$ is equivalent to minimizing $J$ while keeping
$K$ fixed and preserving
$\phi^\alpha_0\le\cdots\le\phi^\alpha_k$.

First consider $\phi^\alpha_0$.  The function $K$ does not depend on the leading
factor except through the common multiplicative term that is fixed on the
level set, while increasing $\phi^\alpha_0$ decreases
$1+1/r_1^{(k+2)}$.  Hence, if $\phi^\alpha_0<\phi^\alpha_1$, we can increase
$\phi^\alpha_0$ slightly, preserve feasibility, and strictly decrease $J$.
Therefore every minimizer satisfies $\phi^\alpha_0=\phi^\alpha_1$.

Now fix $i\in\{0,\ldots,k-1\}$ and suppose $\phi^\alpha_i<\phi^\alpha_{i+1}$.  For
$m\ge i$, differentiating the normalized CDF gives
\[
\frac{\partial F_{k+2}(m)}{\partial \phi^\alpha_i}
=
-
\frac{\alpha}{
(\phi^\alpha_i-(i+1)(1+\alpha))
(\phi^\alpha_i-(i+1)(1+\alpha)+\alpha)}
F_{k+2}(m),
\]
whereas the derivative is zero for $m<i$.  Hence
$\partial K/\partial\phi^\alpha_i<0$ and
$\partial K/\partial\phi^\alpha_{i+1}<0$.

For $\ell=0,\ldots,k-1$, after cancelling the common negative factors in
$\partial_{\phi^\alpha_\ell}J$ and $\partial_{\phi^\alpha_\ell}K$, we obtain
\[
\frac{\partial J/\partial\phi^\alpha_\ell}
{\partial K/\partial\phi^\alpha_\ell}
=
\frac{
\displaystyle
\sum_{m=\ell+1}^{k+1}
\frac{\alpha(k+1-m)F_{k+2}(m)}
{2(1+\alpha)m(m+1)}
\prod_{j=m+2}^{k+1}
\left(1-\frac1{(1+\alpha)j}\right)
}
{
\displaystyle
\sum_{m=\ell+1}^{k+1}
\frac{F_{k+2}(m)}{2(1+\alpha)(m+1)}
\prod_{j=m+2}^{k+2}
\left(1-\frac1{(1+\alpha)j}\right)
}.
\]
For each $m$, the quotient of the coefficient of $F_{k+2}(m)$ in the numerator
by the corresponding coefficient in the denominator equals
\(\frac{\alpha(1+\alpha)(k+2)(k+1-m)} {m((1+\alpha)(k+2)-1)},\)
which is strictly decreasing in $m$.  Therefore the above weighted average is
strictly larger at $\ell=i$ than at $\ell=i+1$:
\[
\frac{\partial J/\partial\phi^\alpha_i}
{\partial K/\partial\phi^\alpha_i}
>
\frac{\partial J/\partial\phi^\alpha_{i+1}}
{\partial K/\partial\phi^\alpha_{i+1}}.
\]
Along the local level set of $K$, the implicit function theorem gives
\[
\frac{d\phi^\alpha_{i+1}}{d\phi^\alpha_i}
=
-
\frac{\partial K/\partial\phi^\alpha_i}
{\partial K/\partial\phi^\alpha_{i+1}}
<0,
\]
and therefore
\[
\frac{dJ}{d\phi^\alpha_i}
=
\frac{\partial K}{\partial\phi^\alpha_i}
\left(
\frac{\partial J/\partial\phi^\alpha_i}{\partial K/\partial\phi^\alpha_i}
-
\frac{\partial J/\partial\phi^\alpha_{i+1}}{\partial K/\partial\phi^\alpha_{i+1}}
\right)<0.
\]
Thus we can increase the smaller value $\phi^\alpha_i$ and decrease the larger
value $\phi^\alpha_{i+1}$ while preserving $K$, preserving weak monotonicity for a
small perturbation, and strictly decreasing $J$.  This contradicts minimality.
Hence no minimizer can satisfy $\phi^\alpha_i<\phi^\alpha_{i+1}$, and every minimizer has
\(\phi^\alpha_0=\phi^\alpha_1=\cdots=\phi^\alpha_k.\)
Equivalently,
\[
    (1+\alpha)+\alpha r_1^{(k+2)}
    =2(1+\alpha)+\alpha r_2^{(k+2)}
    =\cdots=
    (k+1)(1+\alpha)+\alpha r_{k+1}^{(k+2)}.
\]
By Lemma~\ref{lem:extremal-t1}, $\mathcal D\ge0$.

Now suppose the $k+2$ execution outputs
$\omega=(y;r_t^{(k+2)},\ldots,r_{k+1}^{(k+2)})$ with $t\ge2$.  Construct
$x_\omega$ in the $k+1$ case.  The inequality
\eqref{eq:threshold-stationary-same-state-positive-zero-based} implies that
$\hat H(k+1;x_\omega)<0$ whenever
$\hat H(k+2;x_{k+2})=0$.  Hence, if the $k+1$ execution succeeds with the same
suffix, its boundary allocation must be strictly larger than $y$; if it does
not succeed, it returns \textsc{fail}.  This proves the first assertion.

For the failure implication, suppose the $k+2$ execution returns
\textsc{fail}.  By Proposition~\ref{prop:regular-never-fail}, such a failure
cannot occur before $b=k+2$.  If it occurs at the virtual-value test for
threshold $t\ge2$, then \eqref{eq:phi-shift-zero-based} implies that the
$k+1$ execution also fails at the corresponding threshold $t-1$.  Otherwise,
the $k+2$ execution reaches the end of the candidate search without finding a
feasible boundary value.  For every remaining candidate state, the same
comparison inequality gives a negative artificial $k+1$ expression.  Hence the
$k+1$ execution also reaches the end of the candidate search without a
feasible allocation.  Therefore the $k+1$ execution returns \textsc{fail}.
\end{proof}

\begin{proof}[Proof of Lemma~\ref{lem:extremal-t1}]

If Algorithm~\ref{alg:state-recursion} cannot find an achievable state at
\(t_b=0\), then \(t_b\) is set to \(1\) in
Line~\ref{line:state-transition-update-kb}. When
Algorithm~\ref{alg:state-recursion} reaches
Line~\ref{line:state-transition-virtual-test}, we have
\(\phi^\alpha_S(t_b)=(1+\alpha)(N-1)+\alpha r > (1+\alpha)b,\)
and therefore Algorithm~\ref{alg:state-recursion} returns \textsc{fail} in
Line~\ref{line:state-transition-return-fail-virtual}. Consequently, the transition is threshold-stationary.

The two executions are constant-$\phi^\alpha$ executions.  Indeed,
\[
    \phi^{\alpha(k+2)}(i)=(1+\alpha)i+\alpha r_i^{(k+2)}
    =(1+\alpha)(k+1)+\alpha r,
    \qquad i=1,\ldots,k+1,
\]
and
\(\phi^{\alpha(k+1)}(i)=(1+\alpha)i+\alpha r_i^{(k+1)} =(1+\alpha)k+\alpha r, \qquad i=1,\ldots,k.\)
Moreover $r_i^{(k+1)}=r_{i+1}^{(k+2)}$ for $i=1,\ldots,k$.

For $N\in\{k+1,k+2\}$, let $F_N$ be the normalized CDF generated by the
corresponding ratio sequence, and define
\[
U_N(\lambda)
:=
\sum_{t=1}^{N}
\frac{1-a}{t}\,F_N(t-1)
\prod_{j=t+1}^{N}\left(1-\frac{1-a}{j}\right).
\]
Since $\beta=1/2$ and $1+\alpha=1/(1-a)$, proving $\mathcal D\ge0$ is
equivalent to proving
\begin{equation}
\label{eq:U-difference-goal-zero-based}
2(1-a)
+
\left(1+\frac{a}{k+\lambda}\right)(k+1)U_{k+1}(\lambda)
-
(k+2)U_{k+2}(\lambda)
\ge0.
\end{equation}

Write
\(U_{k+1}(\lambda)=\sum_{t=1}^{k+1}A_t(\lambda),\)
where
\[
A_t(\lambda)
=
\frac{1-a}{t}\,F_{k+1}(t-1)
\prod_{s=t+1}^{k+1}\left(1-\frac{1-a}{s}\right),
\qquad 1\le t\le k+1.
\]
Then
\[
\frac{k+\lambda+1-t}{k+\lambda+1-t+a}A_t(\lambda)
=
\frac{t-1}{t-1+a}A_{t-1}(\lambda),
\qquad 2\le t\le k+1.
\]
Furthermore,
\(\left(1+\frac{a}{k+\lambda}\right)F_{k+1}(t-1) =F_{k+2}(t), \qquad 1\le t\le k+1.\)
Thus
\[
U_{k+2}(\lambda)
=
\frac{k+1+a}{k+2}
\left(
A_1(\lambda)
+
\frac{k+\lambda+a}{k+\lambda}
\sum_{t=1}^{k+1}\frac{t}{t+a}A_t(\lambda)
\right).
\]
Consequently,
\[
\left(1+\frac{a}{k+\lambda}\right)(k+1)U_{k+1}(\lambda)
-(k+2)U_{k+2}(\lambda)
=
\sum_{t=1}^{k+1}w_tA_t(\lambda),
\]
where
\(w_1 = \frac{ak(k+\lambda+a)}{(1+a)(k+\lambda)}-(k+1+a),\)
and, for $2\le t\le k+1$,
\(w_t = \frac{a(k+\lambda+a)(k+1-t)}{(k+\lambda)(t+a)}.\)

By Lemma~\ref{lem:1}, $\lambda<1-a$.  Since
$A_1(\lambda)=A_1(1-a)$ and, for $t\ge2$,
\[
\frac{A_t(\lambda)}{A_t(1-a)}
=
\prod_{s=2}^{t}
\frac{k+\lambda+1-s+a}{k+\lambda+1-s}
\cdot
\frac{k+2-a-s}{k+2-s}
\ge1,
\]
we have $A_t(\lambda)\ge A_t(1-a)$ for $t\ge2$.  Also,
$\lambda<1-a$ implies
\(\frac{k+\lambda+a}{k+\lambda} \ge \frac{k+1}{k+1-a}.\)
Therefore
\(w_1 \ge \frac{ak(k+1)}{(1+a)(k+1-a)}-(k+1+a),\)
and, for $2\le t\le k+1$,
\(w_t \ge \frac{a(k+1)(k+1-t)}{(k+1-a)(t+a)}.\)
Combining these inequalities gives
\begin{align*}
&2(1-a)+\sum_{t=1}^{k+1}w_tA_t(\lambda)
\\
&\ge
2(1-a)
+
\left(\frac{ak(k+1)}{(1+a)(k+1-a)}-(k+1+a)\right)A_1(1-a)
+
\sum_{t=2}^{k+1}
\frac{a(k+1)(k+1-t)}{(k+1-a)(t+a)}A_t(1-a)
\\
&=
(1-a)
\left[
2
-
\frac{(a)_{k+1}}{(k+1)!}
-
\sum_{s=0}^{k}
\frac{1-a}{s+1-a}\frac{(a)_s}{s!}
\right],
\end{align*}
where $(a)_s$ is the rising Pochhammer symbol.  Since
\[
\sum_{s=0}^{k}
\frac{1-a}{s+1-a}\frac{(a)_s}{s!}
+
\frac{(a)_{k+1}}{(k+1)!}
<
\sum_{s=0}^{\infty}
\frac{1-a}{s+1-a}\frac{(a)_s}{s!}
=
(1-a)B(1-a,1-a)<2,
\]
we obtain
\(2(1-a)+\sum_{t=1}^{k+1}w_tA_t(\lambda)>0.\)
This proves \eqref{eq:U-difference-goal-zero-based}, and hence
$\mathcal D\ge0$.
\end{proof}

\begin{proof}[Proof of Lemma~\ref{lem:1}]
If $\lambda=0$, then the claim is immediate.  Assume $\lambda>0$.
Because the $k+2$ execution outputs a proper suffix state, the threshold-zero
candidate at $b=k+2$ is not feasible.  If $x'$ denotes the allocation that
keeps the threshold at zero at this step, then
\begin{equation}
\label{eq:H-negative-main-zero-based}
\hat H(k+2;x')
=
 p_{k+2}
 -\beta\prod_{i=1}^{k+1}\left(1+\frac1{r_i^{(k+2)}}\right)
 +(1-p_{k+2})(k+2)(1+\alpha)
\le0.
\end{equation}
Define
\[
U_{k+2}(\lambda)
:=
\sum_{t=1}^{k+2}
\frac{1-a}{t}\,F_{k+2}(t-1)
\prod_{j=t+1}^{k+2}\left(1-\frac{1-a}{j}\right).
\]
Since $\beta=1/2$, inequality
\eqref{eq:H-negative-main-zero-based} implies
\begin{equation}
\label{eq:U-lambda-greater-than-2-zero-based}
    U_{k+2}(\lambda)\ge2.
\end{equation}

Let the $t$-th summand of $U_{k+2}(\lambda)$ be
\[
A_t(\lambda)
:=
\frac{1-a}{t}\,F_{k+2}(t-1)
\prod_{j=t+1}^{k+2}\left(1-\frac{1-a}{j}\right),
\qquad 1\le t\le k+2.
\]
For $1\le t\le k+1$,
\[
\frac{A_{t+1}(\lambda)}{A_t(\lambda)}
=
\frac{t}{t+a}
\cdot
\frac{k+1+\lambda-t+a}{k+1+\lambda-t}.
\]
Hence, for $t\ge2$,
\[
\frac{A_t(\lambda)}{A_1(\lambda)}
=
\prod_{j=1}^{t-1}
\frac{j}{j+a}
\cdot
\frac{k+1+\lambda-j+a}{k+1+\lambda-j}.
\]
Every factor is strictly decreasing in $\lambda$, while $A_1(\lambda)$ is
independent of $\lambda$.  Thus $U_{k+2}(\lambda)$ is strictly decreasing in
$\lambda$.

It remains to show $U_{k+2}(1-a)<2$.  When $\lambda=1-a$, telescoping gives
\[
A_t(1-a)
=
\frac{1-a}{k+2}
\cdot
\frac{\Gamma(k+2+a)}{\Gamma(k+2-a)}
\cdot
\frac{\Gamma(t)\Gamma(k+3-t-a)}
{\Gamma(t+a)\Gamma(k+3-t)}.
\]
Using the Beta-Gamma identity,
\[
U_{k+2}(1-a)
=
\frac{(1-a)(k+2-a)}{k+2}
\int_0^1
(1-x)^{-a}
\sum_{t=1}^{k+2}
\frac{\Gamma(k+2+a)}{\Gamma(t+a)\Gamma(k+3-t)}
 x^{t-1}(1-x)^{k+2-t}
\,dx.
\]
A standard induction gives
\[
\sum_{t=1}^{k+2}
\frac{\Gamma(k+2+a)}{\Gamma(t+a)\Gamma(k+3-t)}
 x^{t-1}(1-x)^{k+2-t}
=
\sum_{j=0}^{k+1}\frac{(a)_j}{j!}(1-x)^j.
\]
Therefore
\[
U_{k+2}(1-a)
=
\frac{k+2-a}{k+2}
\sum_{j=0}^{k+1}
\frac{1-a}{j+1-a}\frac{(a)_j}{j!}.
\]
Since $(k+2-a)/(k+2)<1$,
\[
U_{k+2}(1-a)
<
\sum_{j=0}^{\infty}
\frac{1-a}{j+1-a}\frac{(a)_j}{j!}
=
(1-a)B(1-a,1-a).
\]
Finally,
\((1-a)B(1-a,1-a) = 2(1-a)B(1-a,2-a) <2(1-a)\int_0^1 y^{-a}\,dy =2.\)
Thus $U_{k+2}(1-a)<2$.  Since $U_{k+2}$ is strictly decreasing, any
$\lambda\ge1-a$ would imply
$U_{k+2}(\lambda)\le U_{k+2}(1-a)<2$, contradicting
\eqref{eq:U-lambda-greater-than-2-zero-based}.  Hence $\lambda<1-a$.
\end{proof}

\begin{proof}[Proof of Proposition~\ref{prop:strict-decrease-non-stationary}]
We only need to consider the case that there is a turning point in the \(k+2\) case execution before the buyer value \(b=k+2\).
Let buyer value \(j+2\) be the turning point, so \(k_b=0\) for
\(b\le j+1\) and \(k_{j+2}=1\). For \(b=j+2,\ldots,k+2\), write
\(x_b^{(k+2)}:=x^{(k+2)}(1,b)\), and let \(x_{b-1}^{(k+1)}\) be the
corresponding value generated by
\(\mathrm{ALG}_2(\omega_0;\mathbf r)\) at buyer value \(b-1\).

If the \(k+1\) case has already returned \textsc{fail}, there is nothing to prove, so
we only need to consider the case that \(k+1\) case execution is threshold-stationary and does not fail at \(b=k+1\) step.
Set \(R_b:=(b-1)(1+\alpha)-\alpha r_0\), and define
\[ Q_b:=\frac{R_b}{r_0}(1-x_b^{(k+2)})
-\left(1+\frac1{r_0}\right)(b-1)(1+\alpha)
\bigl(1-x_{b-1}^{(k+1)}\bigr).
\]
At the first turning point, Algorithm~\ref{alg:state-recursion} has not returned
\textsc{fail} in Line~\ref{line:state-transition-return-fail-virtual} at threshold index
\(1\), so \(R_{j+2}>0\). Since
\(R_b=R_{j+2}+(b-j-2)(1+\alpha)\), we have \(R_b>0\) for all
\(b=j+2,\ldots,k+2\). We also use \(x_t^{(k+1)}\le1\) for the relevant
values.

\medskip
\noindent\textbf{1. The initial residual \(Q_{j+2}>0\).}
Replace \(r_j\) by a variable \(v>0\). The threshold-stationary recursion gives
\[
x_{j+1}^{(j+1)}(v)
=
\left(1-\frac{1}{(j+1)(1+\alpha)}\right)x_j^{(j+1)}
+
\frac{\beta}{(j+1)(1+\alpha)}
\prod_{\ell=1}^{j-1}\left(1+\frac1{r_\ell}\right)
\left(1+\frac1v\right).
\]
There is a unique \(r_j^*\in[0,r_j]\) such that the corresponding
allocation \(x_{j+1}^{(j+1)}(r_j^*)\) is exactly \(1\).
We call the corresponding allocation $\bar x^{(j+1)}$.

The two clear-step equations at \(b=j+1\), one with \(r_j^*\) and
one with \(r_j\), are
\[
\hat H(j{+}1;\bar x^{(j+1)})=\bar x_j^{(j+1)}
-
\beta\prod_{i=1}^{j-1}\left(1+\frac1{r_i}\right)\left(1+\frac1{r_j^*}\right)
+
(1-\bar x_j^{(j+1)})(1+\alpha)(j+1)
=
0
\]
and
\[
\hat H(j{+}1;x^{(j+1)})=x_j^{(j+1)}
-
\beta\prod_{i=1}^{j}\left(1+\frac1{r_i}\right)
+
\bigl(x_{j+1}^{(j+1)}-x_j^{(j+1)}\bigr)(1+\alpha)(j+1)
=
0.
\]

Subtracting the two equations gives
\begin{align}
    (1+\alpha)(j+1)\bigl(1-x_{j+1}^{(k+1)}\bigr)
=
\beta\prod_{i=1}^{j-1}\left(1+\frac1{r_i}\right)
\left(\frac1{r_j^*}-\frac1{r_j}\right).
\label{eq:j-plus-one}
\end{align}

\begin{equation}
\widetilde x^{(j+2)}(s,j+2)=
\begin{cases}
1, & s \le 1,\\
0, & s > 1,
\end{cases}
\quad\text{and}\quad
\widetilde x^{(j+2)}(\cdot,b)= x^{(j+2)}(\cdot,b), \quad b<j+2.
\end{equation}

According to the proof of Proposition~\ref{prop:success-stationary-phi}, we have
\(\hat H(j{+}2;\widetilde x^{(j+2)})> 0\). That is
\[\widetilde x_{j+1}^{(j+2)}
-
\beta\prod_{i=0}^{j-1}\left(1+\frac1{r_i}\right)\left(1+\frac1{r_j^*}\right)
+
\bigl(1-\widetilde x_{j+1}^{(j+2)}\bigr)(1+\alpha)(j+2)
+
\frac{R_{j+2}}{r_0}
>0.
\]

For the common ratio \(r_j\), we have \(\hat H(j{+}2;x^{(j+2)})=0\). That is
\[
x_{j+1}^{(j+2)}
-
\beta\prod_{i=0}^{j-1}\left(1+\frac1{r_i}\right)\left(1+\frac1{r_j}\right)
+
\bigl(1-x_{j+1}^{(j+2)}\bigr)(1+\alpha)(j+2)
+
\frac{R_{j+2}}{r_0}\,x_{j+2}^{(j+2)}
=
0.
\]

Subtracting two inequalities gives
\begin{align}
    \frac{R_{j+2}}{r_0}\bigl(1-x_{j+2}^{(k+2)}\bigr)
>
\beta\left(1+\frac1{r_0}\right)\prod_{i=1}^{j-1}\left(1+\frac1{r_i}\right)
\left(\frac1{r_j^*}-\frac1{r_j}\right).
\label{eq:j-plus-two}
\end{align}
By \eqref{eq:j-plus-one} and \eqref{eq:j-plus-two}, we obtain \(Q_{j+2}>0\).

\medskip
\noindent\textbf{2. Propagation of \(Q_b>0\).}

We prove by induction that \(Q_b>0\) for \(b=j+2,\ldots,k+2\). The base case is
\(Q_{j+2}>0\). Fix \(b\in\{j+3,\ldots,k+2\}\), and assume \(Q_{b-1}>0\).

In the \(k+1\) case execution, we have the following allocation formula
\begin{equation}
\label{eq:nonstat-kplus1-recursion-r0}
\begin{aligned}
(b-1)(1+\alpha)\bigl(1-x_{b-1}^{(k+1)}\bigr)
={}&
\bigl((b-1)(1+\alpha)-1\bigr)
\bigl(1-x_{b-2}^{(k+1)}\bigr)
+
1-\beta\prod_{\ell=1}^{b-2}\left(1+\frac1{r_\ell}\right).
\end{aligned}
\end{equation}

In the \(k+2\) case, because \(b-1<k+2\), the hypothesis gives
\(k_{b-1}\le1\). After the turning point, the allocation has
threshold index \(1\). 
If $x_b^{(k+2)}=x_{b-1}^{(k+2)}$, since $x_b^{(k+1)}\ge x_{b-1}^{(k+1)}$ by Lemma~\ref{lem:mono-G}, $Q_{b-1}>0$ implies $Q_b>0$ immediately. Otherwise, if $x_b^{(k+2)}>x_{b-1}^{(k+2)}$, we have
\begin{equation}
\label{eq:nonstat-kplus2-recursion-r0}
\hat H(b;x^{(k+2)})=
-\frac{R_b}{r_0}\bigl(1-x_b^{(k+2)}\bigr)
+
\frac{R_b-1}{r_0}\bigl(1-x_{b-1}^{(k+2)}\bigr)
+
\left(1+\frac1{r_0}\right)\left(1-\beta\prod_{\ell=1}^{b-2}\left(1+\frac1{r_\ell}
\right)\right)= 0.
\end{equation}

Combining \eqref{eq:nonstat-kplus1-recursion-r0} and
\eqref{eq:nonstat-kplus2-recursion-r0}, we obtain
\[
\begin{aligned}
Q_b
={}&
\frac{R_b-1}{r_0}\bigl(1-x_{b-1}^{(k+2)}\bigr)
-
\left(1+\frac1{r_0}\right)
\bigl((b-1)(1+\alpha)-1\bigr)
\bigl(1-x_{b-2}^{(k+1)}\bigr).
\end{aligned}
\]

Using \(R_b-1=R_{b-1}+\alpha\) and
\((b-1)(1+\alpha)-1=(b-2)(1+\alpha)+\alpha\), this implies
\begin{equation}
\label{eq:nonstat-Q-recursion-r0}
\begin{aligned}
Q_b
=
Q_{b-1}
+
\alpha
\left[
\frac{1-x_{b-1}^{(k+2)}}{r_0}
-
\left(1+\frac1{r_0}\right)
\bigl(1-x_{b-2}^{(k+1)}\bigr)
\right].
\end{aligned}
\end{equation}
By the definition of \(Q_{b-1}\),
\[
\begin{aligned}
R_{b-1}
\left[
\frac{1-x_{b-1}^{(k+2)}}{r_0}
-
\left(1+\frac1{r_0}\right)
\bigl(1-x_{b-2}^{(k+1)}\bigr)
\right]
=
Q_{b-1}
+
\alpha(r_0+1)\bigl(1-x_{b-2}^{(k+1)}\bigr).
\end{aligned}
\]
The right-hand side is positive because \(Q_{b-1}>0\) and
\(x_{b-2}^{(k+1)}\le1\). Since \(R_{b-1}>0\), the bracket in
\eqref{eq:nonstat-Q-recursion-r0} is positive. Hence \(Q_b>Q_{b-1}>0\).
Therefore \(Q_{k+2}>0\).

\medskip
\noindent\textbf{3. Conclusion at \(b=k+2\).}
At \(b=k+2\), let
\(\omega=(y;r_i,\ldots,r_k)\) denote the state generated by the \(k+2\) case. 
The same expansion at \(b=k+2\) as in the
threshold-stationary ordering argument applies: after normalizing 
\(\hat H(k+1;x^{(k+1)}[\omega])\) by \((1+\frac1{r_0})\), all candidate-dependent suffix terms
cancel. The only remaining difference is the residual \(Q_{k+2}\). Hence
\begin{equation}
\label{eq:nonstat-same-state-Hhat-bound-r0}
\hat H(k+2;x^{(k+2)})-\left(1+\frac{1}{r_0}\right)\hat H(k+1;x^{(k+1)}[\omega])
=Q_{k+2}
>0 .
\end{equation}

Following the same analysis as in Lemma~\ref{lem:extremal-t1}, we obtain the corresponding monotonicity relation between the two cases. Whenever both executions achieve states, the state achieved in the \(k+2\) case is dominated by the state achieved in the \(k+1\) case. Moreover, if the \(k+2\) case returns \textsc{fail}, then the \(k+1\) case also returns \textsc{fail}.
\end{proof}






\section{Omitted Proofs in Section~\ref{sec:ub}}
\begin{proof}[Proof of Proposition~\ref{prop:endpoint_mass_phi}]
Let
\[
x_B(b)=\int_{[0,1]}x(s,b)\,dF_S(s)
=
p_Sx(0,b)+\int_0^1 x(s,b)f_S(s)\,ds,
\]
and
\[
x_S(s)=\int_{[0,1]}x(s,b)\,dF_B(b)
=
p_Bx(s,1)+\int_0^1 x(s,b)f_B(b)\,db.
\]

Starting from Proposition~\ref{prop:general_mixed},
the buyer-side term is
\begin{align*}
\Lambda_B(x_B)
&:=
\int_{[0,1]}
\left(
b\,x_B(b)-\int_0^b x_B(t)\,dt
\right)dF_B(b)\\
&=
p_B\left(x_B(1)-\int_0^1 x_B(t)\,dt\right)
+\int_0^1
\left(
b\,x_B(b)-\int_0^b x_B(t)\,dt
\right)f_B(b)\,db.
\end{align*}
By Fubini,
\[
\int_0^1\int_0^b x_B(t)\,dt\,f_B(b)\,db
=
\int_0^1 x_B(t)\int_t^1 f_B(b)\,db\,dt
=
\int_0^1 x_B(t)\bigl(1-F_B(t)-p_B\bigr)\,dt.
\]
Hence
\begin{align*}
\Lambda_B(x_B)
&=
p_Bx_B(1)
+\int_0^1
\bigl(bf_B(b)-(1-F_B(b))\bigr)x_B(b)\,db\\
&=
p_Bx_B(1)
+\int_0^1 \phi_B(b)x_B(b)f_B(b)\,db.
\end{align*}

Similarly, the seller-side term is
\begin{align*}
\Lambda_S(x_S)
&:=
\int_{[0,1]}
\left(
s\,x_S(s)+\int_s^1 x_S(t)\,dt
\right)dF_S(s)\\
&=
p_S\int_0^1 x_S(t)\,dt
+\int_0^1
\left(
s\,x_S(s)+\int_s^1 x_S(t)\,dt
\right)f_S(s)\,ds.
\end{align*}
Again by Fubini,
\[
\int_0^1\int_s^1 x_S(t)\,dt\,f_S(s)\,ds
=
\int_0^1 x_S(t)\int_0^t f_S(s)\,ds\,dt
=
\int_0^1 x_S(t)\bigl(F_S(t)-p_S\bigr)\,dt.
\]
Therefore
\begin{align*}
\Lambda_S(x_S)
&=
\int_0^1
\bigl(sf_S(s)+F_S(s)\bigr)x_S(s)\,ds\\
&=
\int_0^1 \phi_S(s)x_S(s)f_S(s)\,ds.
\end{align*}

Combining the two parts yields
\[
\Lambda_{\text{con-mass}}(x)
=
p_Bx_B(1)
+\int_0^1 \phi_B(b)x_B(b)f_B(b)\,db
-\int_0^1 \phi_S(s)x_S(s)f_S(s)\,ds.
\]
Now substitute the formulas for \(x_B\) and \(x_S\). After expanding and rearranging,
\begin{align*}
\Lambda_{\text{con-mass}}(x)
&=
p_Sp_B\,x(0,1)
+p_S\int_0^1 \phi_B(b)x(0,b)f_B(b)\,db\\
&\qquad
+p_B\int_0^1 \bigl(1-\phi_S(s)\bigr)x(s,1)f_S(s)\,ds
+\int_0^1\int_0^1
\bigl(\phi_B(b)-\phi_S(s)\bigr)x(s,b)f_S(s)f_B(b)\,db\,ds.
\end{align*}
This is exactly
\[
\Lambda_{\text{con-mass}}(x)
=
\int_{[0,1]^2}
\bigl(\widehat\phi_B(b)-\widehat\phi_S(s)\bigr)
x(s,b)\,dF_S(s)\,dF_B(b),
\]
with the endpoint conventions \(\widehat\phi_S(0)=0\) and \(\widehat\phi_B(1)=1\).

The expression for \(\mathrm{GFT}(x)\) follows from the same decomposition of the product measure:
\[
dF_S(s)dF_B(b)
=
p_Sp_B\,\delta_{(0,1)}
+p_S\,\delta_0(ds)f_B(b)\,db
+p_B\,f_S(s)\,ds\,\delta_1(db)
+f_S(s)f_B(b)\,ds\,db.
\]
\end{proof}

\begin{proof}[Proof of Theorem~\ref{thm:ub}]
We show that for every $\epsilon>0$ there exists a pair of distributions such that $\frac{\mathrm{SB}}{\mathrm{FB}}<\frac{1}{2}+\epsilon$. The proof proceeds by deriving the closed-form analytical expressions for $\text{FB}$ and $\text{SB}$, extracting their asymptotic limits as $c \to 0^+$, and finally evaluating the ratio as $\alpha \to \infty$.

\paragraph{First-best benchmark and asymptotics.}

By definition, the ideal expected surplus FB is:
\begin{align*}
\text{FB}(\alpha,c) &= \int_0^1\int_0^b(b-s)f_B(b)f_S(s)\mathrm{d}s\mathrm{d}b\\
&=\int_0^1 F_S(b) (1 - F_B(b)) \mathrm{d}b \\
&= \int_0^1 \left( \frac{c^2}{(c + (1+\alpha)(1-b))(c + (1+\alpha)b)} \right)^a \mathrm{d}b.
\end{align*}

We analyze the asymptotic behavior of $\mathrm{FB}(\alpha,c)$ as $c \to 0^+$. Factoring out $c^{2a}$, the dominant term of the integral is governed by the $c \to 0^+$ limit of the denominator:
$$\lim_{c \to 0^+} \frac{\text{FB}(\alpha,c)}{c^{2a}} = \int_0^1 \frac{1}{((1+\alpha)^2 \cdot b(1-b))^a} \mathrm{d}b = (1+\alpha)^{-2a} \int_0^1 b^{-a} (1-b)^{-a} \mathrm{d}b.$$

\paragraph{Structure of the second-best allocation.}

We first demonstrate that the optimal allocation rule restricts trade exclusively to the boundaries
\[
{\cal T}:=\{(s,b):s=0,\ b\ge 1-y\}\cup\{(s,b):s\le y,\ b=1\},
\]
where \(y\) is the truncation point.

By Proposition~\ref{prop:endpoint_mass_phi}, for the present family the second-best problem is
\[
\max_x \ \mathrm{GFT}(x)
\qquad\text{s.t.}\qquad
\Lambda(x)\ge 0,
\]
with
\begin{align*}
\mathrm{GFT}(x)
&=
P^2x(0,1)
+P\int_0^1 b\,x(0,b)f_B(b)\,db
+P\int_0^1 (1-s)x(s,1)f_S(s)\,ds \\
&\qquad
+\int_0^1\int_0^1 (b-s)x(s,b)f_S(s)f_B(b)\,db\,ds,
\end{align*}
and
\begin{align*}
\Lambda(x)
&=
P^2x(0,1)
+P\int_0^1 \phi_B(b)x(0,b)f_B(b)\,db
+P\int_0^1 \bigl(1-\phi_S(s)\bigr)x(s,1)f_S(s)\,ds \\
&\qquad
+\int_0^1\int_0^1 \bigl(\phi_B(b)-\phi_S(s)\bigr)x(s,b)f_S(s)f_B(b)\,db\,ds.
\end{align*}
Substituting the present family gives
\[
\phi_B(b)=-\frac{c+b}{\alpha},
\qquad
1-\phi_S(s)=-\frac{c+1-s}{\alpha},
\qquad
\phi_B(b)-\phi_S(s)=-\frac{(b-s)+2c+1+\alpha}{\alpha},
\]
for \(0<b<1\), \(0<s<1\), and \(0<s<b<1\), respectively, while the atom \((0,1)\) contributes \(1\) to both \(\mathrm{GFT}\) and \(\Lambda\).

Thus, after factoring out the positive measure elements \(P^2\), \(Pf_B(b)\,db\), \(Pf_S(s)\,ds\), and \(f_S(s)f_B(b)\,db\,ds\), the pointwise coefficients in the objective and in the feasibility constraint are
\[
\begin{array}{c|c|c}
\text{location} & \text{coefficient in }\mathrm{GFT} & \text{coefficient in }\Lambda \\ \hline
(0,1) & 1 & 1 \\[1mm]
(0,b),\ 0<b<1 & b & -\dfrac{c+b}{\alpha} \\[2mm]
(s,1),\ 0<s<1 & 1-s & -\dfrac{c+1-s}{\alpha} \\[2mm]
(s,b),\ 0<s<b<1 & b-s & -\dfrac{(b-s)+2c+1+\alpha}{\alpha}
\end{array}
\]
Hence \((0,1)\) is the unique location at which increasing trade relaxes the feasibility constraint. Therefore any optimum sets \(x(0,1)=1\). After this point mass is fully used, every additional trade raises \(\mathrm{GFT}\) but lowers \(\Lambda\), so the question is where the available feasibility slack should be spent.

We first compare two points on the seller boundary \(b=1\). Let \(0<s_1<s_2<1\). Suppose we increase trade probability at \((s_1,1)\) by \(\Delta_1\) and decrease trade probability at \((s_2,1)\) by \(\Delta_2\), choosing them so that the feasibility constraint is unchanged:
\[
\frac{c+1-s_1}{\alpha}\Delta_1
=
\frac{c+1-s_2}{\alpha}\Delta_2.
\]
The corresponding change in the objective is
\begin{align*}
\Delta \mathrm{GFT}
&=(1-s_1)\Delta_1-(1-s_2)\Delta_2 \\
&=\Delta_1\left[(1-s_1)-(1-s_2)\frac{c+1-s_1}{c+1-s_2}\right] \\
&=\frac{c(s_2-s_1)}{c+1-s_2}\Delta_1
>0.
\end{align*}
Thus, for a fixed loss of feasibility slack, trading with a lower seller type strictly dominates trading with a higher seller type. By symmetry, on the buyer boundary \(s=0\), a higher buyer type strictly dominates a lower buyer type. Therefore, whenever trade occurs on the boundaries, it must take the form
\[
\{(s,1):s\le y\}\cup\{(0,b):b\ge 1-y\}
\]
for some cutoff \(y\in[0,1]\).

We next compare boundary trade with interior trade. Fix a boundary point \((s,1)\) with \(0<s<1\), and an interior point \((s',b')\) with \(0<s'<b'<1\). Write
\[
u:=b'-s'\in(0,1).
\]
Suppose we increase trade probability at \((s,1)\) by \(\Delta_B\) and decrease trade probability at \((s',b')\) by \(\Delta_I\), again choosing them so that the feasibility constraint is unchanged:
\[
\frac{c+1-s}{\alpha}\Delta_B
=
\frac{u+2c+1+\alpha}{\alpha}\Delta_I.
\]
Then the induced change in the objective is
\begin{align*}
\Delta \mathrm{GFT}
&=(1-s)\Delta_B-u\Delta_I \\
&=\Delta_B\left[(1-s)-u\frac{c+1-s}{u+2c+1+\alpha}\right] \\
&=\frac{\Delta_B}{u+2c+1+\alpha}
\Bigl[(1-s)(u+2c+1+\alpha)-u(c+1-s)\Bigr] \\
&=\frac{\Delta_B}{u+2c+1+\alpha}
\Bigl[(1-s)(2c+1+\alpha)-cu\Bigr].
\end{align*}
For fixed \(s\), the last expression is strictly decreasing in \(u\). Hence the strongest interior competitor is obtained at the maximal surplus level \(u\uparrow 1\), i.e. in the limit \(b'\to 1\), \(s'\to 0\). Therefore the boundary point \((s,1)\) dominates every interior point whenever
\[
(1-s)(2c+1+\alpha)-c>0.
\]
Let \(y_1\) be the boundary cutoff at which equality holds. Then
\[
(1-y_1)(2c+1+\alpha)=c,
\]
so
\[
y_1=\frac{c+1+\alpha}{2c+1+\alpha}.
\]
Hence every boundary point with \(s<y_1\) yields a strictly larger increase in the objective than any interior point for the same decrease in feasibility slack. It follows that an optimal allocation cannot place positive mass in the interior before exhausting the boundary segment up to \(y_1\).

\paragraph{Budget exhaustion and second-best surplus.}

We now verify the feasibility of reaching $y_1$. Let the actual trade interval length on the boundary be $t$, and let $y = 1-t$ be the length of the non-trading gap. For the seller boundary, the expected virtual surplus is: $$P \int_0^t (1 - \phi_S(s)) f_S(s) \mathrm{d}s = P \int_0^t (1-s) f_S(s) \mathrm{d}s - P \int_0^t F_S(s) \mathrm{d}s.$$ 
Applying integration by parts yields:
$$\int_0^t (1-s) f_S(s) \mathrm{d}s = (1-t)F_S(t) - P + \int_0^t F_S(s) \mathrm{d}s.$$
Thus, the boundary virtual surplus evaluates to $P [ (1-t)F_S(t) - P ]$. Equating the total virtual surplus to zero for budget balance gives $P^2 + 2P [ (1-t)F_S(t) - P ] = 0$, which reduces to $2y\cdot F_S(1-y) = P$. Substituting the respective distribution functions, we have $2y \left( \frac{c}{c+(1+\alpha)y} \right)^a = \left( \frac{c}{c+1+\alpha} \right)^a$. 

We now verify this budget exhaustion. As derived from the budget balance condition, the actual gap $y$ is the root of the continuous function $g(y) := 2y - \left( \frac{c+(1+\alpha)y}{c+1+\alpha} \right)^a = 0$. We evaluate $g(\cdot)$ at the critical threshold $y_1$:
\begin{align*}g(1-y_1) &= 2\left(\frac{c}{2c+1+\alpha}\right) - \left( \frac{c + (1+\alpha)\frac{c}{2c+1+\alpha}}{c+1+\alpha} \right)^a \\&= \frac{2c}{2c+1+\alpha} - \left( \frac{2c}{2c+1+\alpha} \right)^a.
\end{align*}
Let $z = \frac{2c}{2c+1+\alpha}$. Since $c>0$ and $\alpha>0$, we have $0<z<1$. Moreover, since
$a=\frac{\alpha}{1+\alpha}<1$, it follows that $z<z^a$. Therefore
\[
g(1-y_1)=z-z^a<0.
\]
On the other hand,
\[
g(1/2)
=
1-\left(\frac{c+\frac{1+\alpha}{2}}{c+1+\alpha}\right)^a
>0.
\]
Hence, by the Intermediate Value Theorem, the budget-balancing gap $y$ satisfies
\[
y\in (1-y_1,1/2).
\]
Equivalently, the corresponding boundary trading length $t=1-y$ satisfies
\[
1/2<t<y_1.
\]

Thus the feasibility constraint binds at a cutoff $t$ that is strictly below the critical cutoff $y_1$ at which interior points could become relevant in the marginal comparison. For every boundary point in the traded segments
\[
\{(s,1):0\le s\le t\}
\cup
\{(0,b):1-t\le b\le 1\},
\]
the gain in $\mathrm{GFT}$ per unit loss of feasibility slack is strictly larger than that of any interior point $(s,b)$ with $0<s<b<1$. Consequently, the optimal second-best allocation is
\[
x^\star(s,b)=
\begin{cases}
1, & (s,b)\in\{(0,1)\}\cup\{(s,1):0\le s\le 1-y\}\cup\{(0,b):y\le b\le 1\},\\[2pt]
0, & \text{otherwise.}
\end{cases}
\]

Therefore, the actual Second-Best expected surplus is the sum of the point mass surplus at the boundary \((0,1)\) and the actual surplus from the other boundary trades:
$$\text{SB}(\alpha,c) = P^2 + 2P \int_0^{1-y} (1-s) f_S(s) \mathrm{d}s.$$
Applying the same integration by parts, this expands to $P^2 + 2P \left[ yF_S(1-y) - P + \int_0^{1-y} F_S(s) \mathrm{d}s \right]$. By enforcing the budget balance condition $2yF_S(1-y) = P$, the discrete algebraic terms cancel ($P^2 + P^2 - 2P^2 = 0$), collapsing the surplus to a pure integral:
$$\text{SB}(\alpha,c) = 2P \int_0^{1-y} F_S(s) \mathrm{d}s = 2 \left( \frac{c}{c+1+\alpha} \right)^a \int_0^{1-y} \frac{c^a}{(c+(1+\alpha)(1-s))^a} \mathrm{d}s.$$

Evaluating this fractional integral yields $c^a \left[ (c+1+\alpha)^{1-a} - (c+(1+\alpha)y)^{1-a} \right]$. Factoring out $(c+1+\alpha)^{1-a}$ and utilizing the relation $\frac{c+(1+\alpha)y}{c+1+\alpha} = (2y)^{1/a}$ derived from the budget constraint, the expression becomes:
$$\mathrm{SB}(\alpha,c) = 2(c+1+\alpha) \left( \frac{c}{c+1+\alpha} \right)^{2a} \left[ 1 - (2y)^{\frac{1}{\alpha}} \right].$$

As $c\to 0^{+}$, the budget constraint $2y=\left( \frac{c+(1+\alpha)y}{c+1+\alpha} \right)^a$ simplifies to $2y\sim y^a$, yielding $y\sim \left( \frac{1}{2} \right)^{1+\alpha}$. 
Substituting the invariant $1-(2y)^{1/\alpha}=1/2$, we obtain:
$$\text{SB}(\alpha,c) \sim 2(1+\alpha) \cdot c^{2a} (1+\alpha)^{-2a} \cdot \frac{1}{2} = c^{2a} (1+\alpha)^{1-2a}.$$

\paragraph{Limiting approximation ratio.}

Then, we define the limit ratio $R_0(\alpha) = \lim_{c \to 0^+} \frac{\text{SB}(\alpha,c)}{\text{FB}(\alpha,c)}$.
Dividing the asymptotic expression of SB by FB, the $c^{2a}$ terms cancel. Recalling that $-2a = -\frac{2\alpha}{1+\alpha}$, we simplify the remaining $(1+\alpha)$ exponents:
$$R_0(\alpha) = \frac{(1+\alpha)^{1-2a}}{(1+\alpha)^{-2a} B\left( \frac{1}{1+\alpha}, \frac{1}{1+\alpha} \right)} = \frac{1+\alpha}{B\left( \frac{1}{1+\alpha}, \frac{1}{1+\alpha} \right)}.$$
Expressing the Beta function in terms of Gamma functions $\Gamma(\cdot)$ using the identity $B(z,z) = \frac{\Gamma(z)^2}{\Gamma(2z)}$, we acquire the ratio by the parameter $\alpha$:
$$R_0(\alpha) = (1+\alpha) \frac{\Gamma\left( \frac{2}{1+\alpha} \right)}{ \Gamma\left( \frac{1}{1+\alpha} \right)^2}.$$
Finally, we evaluate $R_0(\alpha)$ as $\alpha \to \infty$. Let $z = \frac{1}{1+\alpha}$. As $\alpha \to \infty$, $z \to 0^+$. The ratio function rewrites to:
$$R_0(z) = \frac{1}{z} \frac{\Gamma(2z)}{\Gamma(z)^2}.$$
To resolve the $\frac{\infty}{\infty}$ singularity at $z \to 0^+$, we apply the elementary recurrence relation of the Gamma function:$\Gamma(z) = \frac{\Gamma(z+1)}{z}$, $\Gamma(2z) = \frac{\Gamma(2z+1)}{2z}$.
Substituting these into the ratio:
$$R_0(z) = \frac{1}{z} \cdot \frac{\frac{\Gamma(2z+1)}{2z}}{\left( \frac{\Gamma(z+1)}{z} \right)^2} = \frac{1}{2} \cdot \frac{\Gamma(2z+1)}{\Gamma(z+1)^2}.$$Because the Gamma function is strictly continuous for positive real numbers and $\Gamma(1) = 1$, we can now directly evaluate the limit as $z \to 0^+$:
$$\lim_{z \to 0^+} R_0(z) = \frac{1}{2} \cdot \frac{\Gamma(1)}{(\Gamma(1))^2}  = \frac{1}{2}.$$
Therefore, taking the sequential limits $\lim_{\alpha \to \infty} \lim_{c \to 0^+} \frac{\text{SB}(\alpha,c)}{\text{FB}(\alpha,c)}$ yields precisely $1/2$, completing the proof.
\end{proof}

\printbibliography

\end{document}